\PassOptionsToPackage{colorlinks,linkcolor=BlueViolet,citecolor=OliveGreen,urlcolor=RawSienna,hyperfootnotes=false}{hyperref} 

\documentclass[11pt,a4paper]{article}

\pdfoutput=1

\usepackage[british]{babel}
\addto\captionsbritish{}

\usepackage{mathtools} 
\numberwithin{equation}{section} 

\usepackage{amssymb}
\usepackage{bbm} 

\usepackage{url, verbatim} 

\usepackage[a4paper,top=2in,bottom=1.5in,left=1in,right=1in]{geometry}

\usepackage{bold-extra} 

\usepackage{enumitem} 

\makeatletter
\renewcommand\subparagraph{\@startsection{subparagraph}{5}%
	{\parindent}
	{0pt}
	{-1em}
	{\normalfont\itshape}
} 
\makeatother

\usepackage[dvipsnames]{xcolor}

\usepackage{graphicx} 

\usepackage[bf, small]{caption} 
\usepackage{tikz}
\usetikzlibrary{calc} 
\usetikzlibrary{intersections} 
\usetikzlibrary{arrows} 
\usetikzlibrary{decorations.markings} 

\usetikzlibrary{matrix,arrows,decorations.pathmorphing} 

\usepackage{amsthm} 

\newtheorem{thm}{Theorem}

\newtheoremstyle{remarkitalicnumber}{3pt}{3pt}{}{}{\itshape}{.}{.5em}{\thmname{#1}\thmnumber{ \itshape#2}\thmnote{ (#3)}} 
\theoremstyle{remarkitalicnumber} 
\newtheorem{exercise}{Exercise}[section]

\usepackage[numbers,merge,sort&compress]{natbib}
\let\originalleft\left
\let\originalright\right
\renewcommand{\left}{\mathopen{}\mathclose\bgroup\originalleft}
\renewcommand{\right}{\aftergroup\egroup\originalright}

\newcommand{\cH}{\mathcal{H}}
\newcommand{\cN}{\mathcal{N}}

\newcommand{\svA}{\tilde{A}}
\newcommand{\svDelta}{\tilde{\Delta}}
\newcommand{\svcH}{\tilde{\mathcal{H}}}
\newcommand{\svV}{\tilde{V}}
\newcommand{\svp}{\tilde{p}}
\newcommand{\svPsi}{\tilde{\Psi}}
\newcommand{\svS}{\tilde{S}}

\newcommand{\choice}[2]{#2} 

\newcommand{\To}{\cdots\mspace{-1mu}}
\DeclareMathOperator*{\ordprod}{\prod\limits^{\vbox to -.5ex{\kern-0.5ex\hbox{$\leftharpoonup$}\vss}}}

\newcommand{\Vector}[1]{\boldsymbol{#1}}

\DeclarePairedDelimiter{\bra}{\langle}{\rvert}
\DeclarePairedDelimiter{\ket}{\lvert}{\rangle}
\DeclarePairedDelimiterX{\braket}[2]{\langle}{\rangle}{#1\vert#2}

\DeclareMathOperator{\E}{e}
\newcommand{\I}{\mathrm{i}}

\DeclareMathOperator{\End}{End}
\DeclareMathOperator{\inv}{inv}

\DeclareMathOperator{\id}{\mathbbm{1}} 
\DeclareMathOperator{\perm}{P}

\newcommand{\D}{\mathrm{d}}
\let\Im\undefined\DeclareMathOperator{\Im}{Im}

\DeclareMathOperator{\Res}{Res}
\DeclareMathOperator{\sgn}{sgn}
\DeclareMathOperator{\tr}{tr}

\DeclareRobustCommand{\w}[4]{w \begin{psmallmatrix} & \!#3\! & \\[-0.2em] \!#2\! & & \!#4\! \\[-0.2em] & \!#1\! & \end{psmallmatrix} } 
\DeclareRobustCommand{\wspec}[5]{w \left(#1 \middle| \begin{smallmatrix} & \!#4\! & \\[-0.2em] \!#3\! & & \!#5\! \\[-0.2em] & \!#2\! & \end{smallmatrix}\right) }

\title{\textbf{A pedagogical introduction to quantum integrability} \\ \textit{with a view towards theoretical high-energy physics}}

\author{Jules Lamers \\ \\ Institute for Theoretical Physics \\ Center for Extreme Matter and Emergent Phenomena, Utrecht University \\
Leuvenlaan 4, 3584 CE Utrecht, The Netherlands \\ \\ \texttt{j.lamers@uu.nl}}

\date{}

\begin{document}

\begin{flushright}
\footnotesize
ITP-UU-15/01 
\end{flushright}

{\let\newpage\relax\maketitle} 

\thispagestyle{empty}

\begin{abstract}
These are lecture notes of an introduction to quantum integrability given at the Tenth Modave Summer School in Mathematical Physics, 2014, aimed at PhD candidates and junior researchers in theoretical physics.

We introduce spin chains and discuss the coordinate Bethe ansatz (\textsc{cba}) for a representative example: the Heisenberg \textsc{xxz} model. The focus lies on the structure of the \textsc{cba} and on its main results, deferring a detailed treatment of the \textsc{cba} for the general $M$-particle sector of the \textsc{xxz} model to an appendix. Subsequently the transfer-matrix method is discussed for the six-vertex model, uncovering a relation between that model and the \textsc{xxz} spin chain. Equipped with this background the quantum inverse-scattering method (\textsc{qism}) and algebraic Bethe ansatz (\textsc{aba}) are treated. We emphasize the use of graphical notation for algebraic quantities as well as computations. 

Finally we turn to quantum integrability in the context of theoretical high-energy physics. We discuss factorized scattering in two-dimensional \textsc{qft}, and conclude with a qualitative introduction to one current research topic relating quantum integrability to theoretical high-energy physics: the Bethe/gauge correspondence.
\end{abstract}

\tableofcontents


\section{Introduction}\label{s1}

Quantum integrability is a beautiful and rich topic in mathematical physics, lying at the interface between condensed-matter physics, theoretical high-energy physics and mathematics. Usually, a (quantum) statistical model is considered `solved' if the ground states, elementary excitations, and various thermodynamic quantities are known. Quantum-integrable models possess a deep underlying structure that often allows for \emph{exact} computation of such quantities. At the same time several of these models are quite \emph{realistic}, and theoretic results may be tested with experiments.
Inevitably, then, the theory of quantum integrability is rather technical, which may obscure its beauty to newcomers. These notes aim to give a pedagogical introduction to quantum integrability and help the reader cross that first potential barrier.

\paragraph{Historical overview.} Quantum-integrable models emerged in two different branches of physics. The first example came from \emph{quantum mechanics}: the isotropic Heisenberg `\textsc{xxx}' spin chain for (ferro)magnetism. In a seminal paper from 1931, Bethe solved this model using a method that now goes under the name of \emph{coordinate Bethe ansatz}~(\textsc{cba}), turning the problem of finding the model's spectrum into the problem of solving certain coupled equations, called the \emph{Bethe-ansatz equations}~(\textsc{bae}). In the subsequent decades Bethe's work was developed further by others, and in the 1960s Yang and Yang applied the \textsc{cba} to the more general `\textsc{xxz}' spin chain.

The second source of quantum-integrable models was \emph{statistical mechanics}. Here the prototype is the six-vertex or ice-type model for two-dimensional hydrogen-bonded crystals. In the late 1960s Lieb and Sutherland were able to solve the six-vertex model via the transfer-matrix method --- famously used by Onsager to tackle the 2d square-lattice Ising model in 1944 --- together with the \textsc{cba} as in the work of Yang and Yang. This solution uncovered several striking similarities between the six-vertex model and the \textsc{xxz} spin chain, and shed light on the reason \emph{why} these models could be solved.

In the late 1970s these two stories were unified by the \emph{quantum inverse-scattering method}~(\textsc{qism}) developed by the `Leningrad group' of Faddeev et al, and others. Using ideas from classical integrability and soliton theory, the \textsc{qism} provides an algebraic framework for quantum-integrable models, in particular yielding the \textsc{bae} via the \emph{algebraic Bethe ansatz}~(\textsc{aba}).

\paragraph{Outline.} These notes are organized as follows. Sections \ref{s2}, \ref{s3} and~\ref{s4} contain an introduction to quantum integrability, roughly following the above historical account. (The four-hour Modave lectures on which these notes are based covered most of this material.) The quantum-mechanical side of the story is treated in Section~\ref{s2}. We introduce spin chains like the \textsc{xxx} and \textsc{xxz} models, present the \textsc{cba} for such models, and discuss the main results for the \textsc{xxz} spin chain. In Section~\ref{s3} we switch to the statistical-mechanical side. We introduce the six-vertex model, treat it using the transfer-matrix method and \textsc{cba}, and provide the results. By examining the outcome more closely we uncover the correspondence with the \textsc{xxz} model. Equipped with this background, the \textsc{qism} is developed in Section~\ref{s4}. This provides the precise relation between the \textsc{xxz} and six-vertex models and, via the \textsc{aba}, allows us to rederive the results of the \textsc{cba} for these models using a single computation. 

In Section~\ref{s5} we move on to \textsc{qft} and theoretical high-energy physics. After providing an overview of the various relations that have been found with quantum integrability, and a discussion of factorized scattering in \textsc{qft} in two dimensions, we give a qualitative introduction to the Bethe/gauge correspondence as a recent example of such a relation.

There are three appendices containing further details and background. In Appendix~\ref{sY} we present the Yang-Yang function. The details of the \textsc{cba} are worked out for the \textsc{xxz} spin chain in Appendix~\ref{sM}. Finally, in Appendix~\ref{sR} the $R$-matrix of the six-vertex model is found.

Although none of the material in these notes is new, this introduction is somewhat different from most other introductory texts. For example, in Sections \ref{s2:method} and~\ref{s2:results} we focus on the conceptual basis and the physics of the \textsc{cba} and its results rather than on computations. Still, the \textsc{cba} is worked out not just for the two-particle sector but, following~\cite[\textsection8.4]{Bax07}, also for the general case in Appendix~\ref{sM}. Our presentation of the transfer-matrix method and \textsc{qism} in Sections \ref{s3:method} and~\ref{s4} consistently exploits a graphical notation adapted from~\cite{GRS96}. Though not always the most practical way to perform computations, this diagrammatic notation is a convenient way to understand what is going on algebraically.

\paragraph{Further references.} Many important topics in quantum integrability are barely touched in these notes; examples include Baxter's $TQ$-method, the thermodynamic limit, correlation functions, and quantum groups. Luckily the literature on quantum-integrable models is extensive, ranging from introductory texts to very technical papers. The following references, here ordered alphabetically, have been useful for preparing these notes: 
\begin{itemize}
\item The renowned book by Baxter~\cite{Bax07} gives a very detailed account of the \textsc{cba} and the $TQ$-method for several quantum-integrable models in statistical mechanics, including the six-vertex model. The notation is perhaps a bit old fashioned at times.
\item Faddeev's famous Les Houches lecture notes~\cite{Fad95a} provide a good basis for the \textsc{aba} and the \textsc{xxx} model. Some familiarity with quantum integrability may be useful.
\item Gaudin's book~\cite{Gau14} was recently translated into English. Amongst others the \textsc{xxz} spin chain and the six-vertex model are treated using the \textsc{cba}, and the thermodynamic limit is studied.
\item Chapters 1--3 of the book by G\'omez, Ruiz-Altaba and Sierra~\cite{GRS96} treat the \textsc{cba} and \textsc{aba} for the \textsc{xxz} spin chain and the six-vertex model. The underlying quantum-algebraic structure is pointed out, though perhaps somewhat vaguely at times, and there are nice diagrammatic computations.
\item Chapters 0--2 of the book by Jimbo and Miwa~\cite{JM99} form a neat concise introduction to statistical physics, the \textsc{xxz} spin chain and the six-vertex model. Although the \textsc{aba} is not discussed, the \textsc{qism} is essentially treated in Sections 2.4--3.3 and 3.7.
\item Karbach, Hu and M\"uller~\cite{KM98,KHM98,*KHM00} have written a nice three-part introduction to the \textsc{cba} for the \textsc{xxx} model, including a discussion of the low-lying excitations in the physical spectrum for both the ferromagnetic and antiferromagnetic regime.
\item The well-known book by Korepin, Bogoliubov and Izergin~\cite{KBI93} contains a lot of information about the \textsc{qism} and its applications to correlation functions. The discussion of the basics is quite condensed.
\end{itemize}
A standard reference for classical integrability and soliton theory is the book by Babelon, Bernard and Talon~\cite{BBT03}. For more about the history of quantum integrability see e.g.~\cite{Bat07,*Bax95,*Fad95b}. Experimental realizations of quantum-integrable models are described in~\cite{ME+13,*GB+13,*KWW06}. Numerical methods for the \textsc{xxx} spin chain are discussed e.g.~in \cite{KM98,KHM98,*KHM00} and \cite[Eds.~1 and 4]{MSSTP}. For quantum groups see e.g.\ the chatty introduction~\cite[\textsection1--6]{DB+08} and the mathematics books~\cite{Maj02,CP94,Kas95}.

\paragraph{Acknowledgements.} I thank the organisers of the Tenth Modave Summer School in Mathematical Physics for giving me the opportunity to share my enthusiasm for quantum integrability with my peers. I am grateful to the participants of the school for their interest and questions. In preparing the lectures and these notes I benefited from discussions with G.~Arutyunov, R.~Borsato, W.~Galleas, A.~Henriques, R.~Klabbers and D.~Schuricht. 

I gratefully acknowledge the support of the Netherlands Organization for Scientific Research~(\textsc{nwo}) under the \textsc{vici} grant 680-47-602. This work is part of the \textsc{erc} Advanced Grant no.~246974, \textit{Supersymmetry: a window to non-perturbative physics}, and of the \textsc{d-itp} consortium, a program of the \textsc{nwo} funded by the Dutch Ministry of Education, Culture and Science~(\textsc{ocw}).

\newpage

\section{Bethe's method for the \textsc{xxz} model}\label{s2}

The pioneering work of Bethe on the one-dimensional Heisenberg model for ferromagnetism is one of the corner stones of the theory of quantum integrability. Although nowadays many quantum-integrable models can be tackled in more sophisticated ways, as we will e.g.\ see in Section~\ref{s4}, Bethe's method remains a concrete and physical way to introduce the basic ingredients and obtain the main results.

\subsection{The \textsc{xxz} spin chain and its symmetries}\label{s2:spin chains}

At the dawn of the 20th century Maxwell had formulated his laws describing the connection between electric and magnetic forces and optics, but the microscopic mechanism behind magnetism was not understood. The advent of quantum mechanics brought new insights, and Heisenberg and Dirac independently showed in~1926 that Pauli's exclusion principle leads to an effective interaction between electron spins of atoms with overlapping wave functions~\cite{Hei26,*Dir26}. This \emph{exchange interaction} formed the basis for an important model for ferromagnetism published by Heisenberg two years later~\cite{Hei28} (see also \cite[\textsection{8}]{HH95}). In one spatial dimension this is an example of a \emph{spin chain} --- a special class of quantum-mechanical models that are rather simple in their set-up, yet lead to a wide variety of interesting physics and mathematics.

\paragraph{Spin chains.} Consider a one-dimensional array of $L$~atoms, modelled by a lattice of length~$L$ with uniform lattice spacing that we take equal to one. We impose periodic boundary conditions, so that the lattice is $\mathbb{Z}_L \coloneqq \mathbb{Z}/L\mathbb{Z}$. This choice of boundary conditions is very convenient, and not unreasonable since one is typically interested in the physics in the \emph{thermodynamic limit} where $L\to\infty$ becomes macroscopically large.$^\#$\footnote{With the thermodynamic limit in mind one should not really distinguish between two spin chains that only differ in the numbers of lattice sites, but rather think of a spin chain as a \emph{family} of systems indexed by the length~$L\in\mathbb{N}\backslash\{1\}$.}

The microscopic degrees of freedom are quantum-mechanical spins, see Figure~\ref{s2:fg:spin chain}. Thus each site~$l \in \mathbb{Z}_L$ comes with a finite-dimensional vector space~$V_l$ and a spin operator~$\Vector{S}_l = (S^{x\vphantom{y}}_l,S^y_l,S^{z\vphantom{y}}_l)$ on $V_l$ satisfying the $\mathfrak{su}(2)$-relations. The periodic boundary conditions mean that $\Vector{S}_{l+L}=\Vector{S}_l$. We are interested in the case of spin~$1/2$: each~$V_l$ is a copy of~$\mathbb{C}^2$ with basis given by spin up and down, $V_l=\mathbb{C}\ket{\uparrow}_l \oplus\mathbb{C}\ket{\downarrow}_l$, and~$S_l^\alpha$ is represented via the Pauli matrices~$\sigma^\alpha$ as usual.

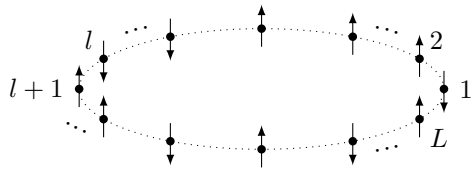
\begin{figure}[h]
	\centering
	\begin{tikzpicture}[>=latex,scale=0.8,font=\small]
		\draw[dotted] (0,0) ellipse (3 and 1);
		\fill [black] ($(0:3 and 1)$) circle (2pt) node[right=2pt]{1};
		\fill [black] ($(30:3 and 1)$) circle (2pt) node[above right]{2};
		\fill [black] ($(60:3 and 1)$) circle (2pt);
		\fill [black] ($(90:3 and 1)$) circle (2pt);
		\fill [black] ($(120:3 and 1)$) circle (2pt);
		\fill [black] ($(150:3 and 1)$) circle (2pt) node[above left]{$l$};
		\fill [black] ($(180:3 and 1)$) circle (2pt) node[left=2pt]{$l+1$};
		\fill [black] ($(210:3 and 1)$) circle (2pt);
		\fill [black] ($(240:3 and 1)$) circle (2pt);
		\fill [black] ($(270:3 and 1)$) circle (2pt);
		\fill [black] ($(300:3 and 1)$) circle (2pt);
		\fill [black] ($(330:3 and 1)$) circle (2pt) node[below right]{$L$};
		\draw[<-] ($(0:3 and 1)-(0,.4)$) -- ($(0:3 and 1)+(0,.3)$);
		\draw[->] ($(30:3 and 1)-(0,.3)$) -- ($(30:3 and 1)+(0,.4)$);
		\draw[->] ($(60:3 and 1)-(0,.3)$) -- ($(60:3 and 1)+(0,.4)$);
		\draw[->] ($(90:3 and 1)-(0,.3)$) -- ($(90:3 and 1)+(0,.4)$);
		\draw[<-] ($(120:3 and 1)-(0,.4)$) -- ($(120:3 and 1)+(0,.3)$);
		\draw[<-] ($(150:3 and 1)-(0,.4)$) -- ($(150:3 and 1)+(0,.3)$);
		\draw[->] ($(180:3 and 1)-(0,.3)$) -- ($(180:3 and 1)+(0,.4)$);
		\draw[->] ($(210:3 and 1)-(0,.3)$) -- ($(210:3 and 1)+(0,.4)$);
		\draw[<-] ($(240:3 and 1)-(0,.4)$) -- ($(240:3 and 1)+(0,.3)$);
		\draw[->] ($(270:3 and 1)-(0,.3)$) -- ($(270:3 and 1)+(0,.4)$);
		\draw[<-] ($(300:3 and 1)-(0,.4)$) -- ($(300:3 and 1)+(0,.3)$);
		\draw[->] ($(330:3 and 1)-(0,.3)$) -- ($(330:3 and 1)+(0,.4)$);
		\foreach \y in {-1,...,1} \draw ($(3*\y+55:3.6 and 1.2)$) node{$\cdot\mathstrut$};
		\foreach \y in {-1,...,1} \draw ($(3*\y+125:3.6 and 1.2)$) node{$\cdot\mathstrut$};
		\foreach \y in {-1,...,1} \draw ($(4*\y+212:3.6 and 1.2)$) node{$\cdot\mathstrut$};
		\foreach \y in {-1,...,1} \draw ($(3*\y+305:3.6 and 1.2)$) node{$\cdot\mathstrut$};
	\end{tikzpicture}
	\caption{One-dimensional spin chain of length~$L$ with spin $1/2$ and periodic boundary conditions. Cartoons like this, where the spin vector at each site either points up or down, should of course be taken with a grain of salt: really the spins may point in any direction in~$\mathbb{C}^2$.}
	\label{s2:fg:spin chain}
\end{figure}

The Hilbert space of the spin chain is the tensor product of the~$V_l$ over the lattice,
\begin{equation}\label{s2:eq:Hilb global}
	\cH = \bigotimes_{l\in\mathbb{Z}_L} V_l \ ,
\end{equation}
with (orthonormal) basis consisting of tensor products of the local spin vectors~$\ket{\uparrow}_l$ and~$\ket{\downarrow}_l$. The subscript of~$\Vector{S}_l$ keeps track of the factor $V_l$ in \eqref{s2:eq:Hilb global} on which this local spin operator acts nontrivially:
\begin{equation}\label{s2:eq:tensor-leg notation}
	\Vector{S}_l = \underset{\,1}{\vphantom{\tfrac{1}{2}}\id} \otimes \cdots\otimes\id\otimes\, \underset{l}{\tfrac{\hbar}{2}\, \Vector{\sigma}} \otimes\id\otimes\cdots\otimes \underset{\,L}{\vphantom{\tfrac{1}{2}}\id} \ .
\end{equation}
This \emph{tensor-leg notation} is used throughout the literature on quantum integrability and will be particularly helpful in Section~\ref{s4}. (Note that it does not make sense to use a summation convention for these subscripts as there is nothing special about precisely two operators acting nontrivially at the same~$V_l$.) While we are at it, let us also introduce the following common notation. For any vector space~$W$, `$\End(W)$' denotes the space of all linear operators~$W \longrightarrow W$, i.e.\ square matrices of size~$\dim(W)$. For example: $S^\alpha_l\in\End(V_l) \subseteq \End(\cH)$ for $\alpha=x,y,z$.

The relations between the local spin operators can be packaged together into a `global' spin Lie algebra governing the entire spin chain,
\begin{equation}\label{s2:eq:su(2) global}
	[S_k^{\alpha\vphantom{\beta}},S_l^\beta] = \I\,\hbar \, \delta^{\vphantom\gamma}_{k,l} \sum_{\gamma=x,y,z} \varepsilon^{\alpha\beta\gamma} \, S^\gamma_l \ ,
\end{equation}
where the totally antisymmetric $\mathfrak{su}(2)$-structure constant is fixed by~$\varepsilon^{xyz}=1$. The relation~\eqref{s2:eq:su(2) global} is sometimes called \emph{ultralocal} since the spin operators at different sites commute. For computations it is convenient to work with the ($\mathfrak{sl}(2)=\mathfrak{su}(2)_\mathbb{C}$) ladder operators $S_l^\pm \coloneqq S_l^{x\vphantom{y}} \pm \I\, S_l^y$ together with $S^z_l$, satisfying
\begin{equation}\label{s2:eq:su(2) ladder}
	[S_k^{z\vphantom\pm},S_l^\pm] = \pm\hbar\, \delta^{\vphantom\pm}_{k,l} \, S_l^\pm \ , \qquad [S_k^+,S_l^-] = 2 \, \hbar \, \delta^{\vphantom\pm}_{k,l} \, S_l^z \ , \qquad [S_k^\pm,S_l^\pm] = 0 \ .
\end{equation}
With respect to the basis $\{\ket{\uparrow}_l,\ket{\downarrow}_l\}$ of~$V_l$ these operators are given by
\begin{equation}\label{s2:eq:su(2) ladder matrices}
	S_l^+ = \hbar \, \begin{pmatrix} 0 & 1 \\ 0 & 0 \end{pmatrix} \ , \qquad S_l^- = \hbar \, \begin{pmatrix} 0 & 0 \\ 1 & 0 \end{pmatrix} \ , \qquad S_l^z = \frac{\hbar}{2} \, \begin{pmatrix} 1 & 0 \\ 0 & -1 \end{pmatrix} \ .
\end{equation}

The set-up so far can be summarized in more mathematical terms by saying that a \emph{spin chain} is a Hilbert space~$\cH$ as in~\eqref{s2:eq:Hilb global} carrying for each~$l\in \mathbb{Z}_L$ an irreducible $\mathfrak{su}(2)$-representation; for us this is the two-dimensional (defining) representation \eqref{s2:eq:tensor-leg notation}. In fact $\cH$ also carries a `global' $\mathfrak{su}(2)$-representation, given by the \emph{total spin operator} $\Vector{S}=(S^z,S^y,S^z)$ defined as 
\begin{equation}\label{s2:eq:total spin}
S^\alpha \coloneqq \sum_{l\in\mathbb{Z}_L} S_l^\alpha \in \End(\cH) \ , \qquad\qquad \alpha=x,y,z \ . 
\end{equation}
This representation is reducible, as we will see in~\eqref{s2:eq:M-part decomposition}.
\begin{exercise}\label{s2:ex:Sz block diagonal}
To practice with this notation, compute the matrix of $S^z$ with respect to the standard basis for $\cH$ for $L=2$ and $L=3$. 
\end{exercise}

The last piece of input is a (hermitean) Hamiltonian $H \in \End(\cH)$ describing the exchange interaction between the spins. We will need the following properties from these interactions: they are
\begin{enumerate}[label=\roman*),noitemsep]
	\item only \emph{nearest neighbour};
	\item \emph{homogeneous}, i.e.\ translationally invariant; and
	\item at least \emph{partially isotropic}, i.e.\ $[S^z,H]=0$. 
\end{enumerate}
\begin{exercise}
Argue that any spin-chain Hamiltonian obeying property~(i) can be written as $H = \sum_l H_{l,l+1}$. What does (ii) mean for the boundary conditions when $L$ is finite? 
Try to find the form of the most general local contributions~$H_{l,l+1}$ satisfying (ii)--(iii).
\end{exercise}

\paragraph{Examples.} The simplest spin chain satisfying~(i)--(iii) is the Heisenberg `\textsc{xxx}' model,
\begin{equation}\label{s2:eq:Ham xxx}
	H_\textsc{xxx} = -J \sum_{l\in\mathbb{Z}_L} \Vector{S}_l \cdot \Vector{S}_{l+1} \ ,
\end{equation}
where the \emph{exchange coupling}~$J$ sets the energy scale. Since $[\Vector{S},H_\textsc{xxx}]=0$ this model is completely isotropic and the spins have no preferred direction. Accordingly the spectrum is highly degenerate: the states come in an $\mathfrak{su}(2)$-multiplet for each energy eigenvalue. When $J>0$ the lowest energy is attained when each of the local terms in~\eqref{s2:eq:Ham xxx} contributes maximally, so the spins tend to align. This is the \emph{ferromagnetic} regime studied by Bethe in 1931. In contrast, for $J<0$, the spins tend to anti-align and the macroscopic magnetization vanishes. This \emph{antiferromagnetic} regime was first analyzed by N\'eel in 1948.

A more general spin chain obeying properties (i)--(iii) is the `\textsc{xxz}' or Heisenberg-Ising model,
\begin{equation}\label{s2:eq:Ham xxz}
	H_\textsc{xxz} = -J \sum_{l\in\mathbb{Z}_L} \bigl( S^x_l S^x_{l+1} + S^y_l S^y_{l+1} + \Delta \, S^z_l S^z_{l+1} \bigr) \ ,
\end{equation}
where $\Delta \in \mathbb{R}$ is the \emph{anisotropy} parameter. This model was introduced by Orbach in 1958 and thoroughly studied by Yang and Yang in the 1960s (see \cite{YY66} and references therein). In terms of the ladder operators~\eqref{s2:eq:su(2) ladder} the Hamiltonian \eqref{s2:eq:Ham xxz} reads
\begin{equation}\label{s2:eq:Ham xxz via S^pm}
	H_\textsc{xxz} = -\frac{J}{2} \sum_{l\in\mathbb{Z}_L} \bigl( S^+_l S^-_{l+1} + S^-_l S^+_{l+1} + 2 \, \Delta \, S^z_l S^z_{l+1} \bigr) \ .
\end{equation}
This form clearly shows that the first two terms describe the hopping of excited spins while the third term counts the number of (mis)aligned neighbouring spins.
\begin{exercise}
To get more feeling for the \textsc{xxz} Hamiltonian consider the summand in~\eqref{s2:eq:Ham xxz via S^pm}. By \eqref{s2:eq:tensor-leg notation} and~\eqref{s2:eq:su(2) ladder matrices}, we have e.g.\ $S^z_l S^z_{l+1} \propto (\sigma^z\otimes \id)_{l,l+1}(\id\otimes \, \sigma^z)_{l,l+1}=(\sigma^z\otimes \sigma^z)_{l,l+1}$. Use this to check that with respect to the standard basis of $V_l\otimes V_{l+1}$
\begin{equation}\label{s2:eq:xxz Ham L=2}
	H_{l,l+1} = -\frac{J}{2} \big( S^+_l S^-_{l+1} + S^-_l S^+_{l+1} + 2 \, \Delta \, S^z_l S^z_{l+1} \big) = - \frac{\hbar^2 J}{4} \begin{pmatrix} \Delta\! & & & \\ & \! -\Delta \! & 2 & \\  & 2 & \! -\Delta \!  & \\ & & & \!\Delta \end{pmatrix}_{\!l,l+1} \ , 
\end{equation}
where zeroes are suppressed.
\end{exercise}

\begin{exercise}
Show that for $L$ even it suffices to take $J>0$ and $\Delta\in\mathbb{R}$ by using \eqref{s2:eq:su(2) ladder} to compute $V \, H_\textsc{xxz} \, V^{-1}$ for $V\coloneqq \prod_l S^z_{2l} \in \End(\cH)$. Which value of $\Delta$ corresponds to the antiferromagnetic \textsc{xxx} model in this way?
\end{exercise}

\begin{exercise}\label{s2:ex:external magnetic field}
An external magnetic field in the $z$-direction can be included by adding $-h \sum_l S^z_l$ to the Hamiltonian, preserving properties (i)--(iii). Show that it is enough to consider $h>0$ by calculating $W \, H_\textsc{xxz}(h) \, W^{-1}$ with $W \coloneqq \prod_{l} S^x_l \in \End(\cH)$ the spin-flip operator.
\end{exercise}

There exists a further generalization, the `\textsc{xyz}' model, which has a different coupling constant for each spin-direction~$\alpha$. This spoils property~(iii) and the model cannot be treated using the Bethe ansatz (but see \cite[\textsection9--10]{Bax07}).

\paragraph{Symmetries.} Our goal is to find the spectrum of the \textsc{xxz} model, $H_\textsc{xxz}\ket{\Psi}=E\ket{\Psi}$. This will be achieved in Sections \ref{s2:method} and~\ref{s2:results} using the \textsc{cba}, and again in Section~\ref{s4:ABA} with a more slick method. As always, the symmetries come to our aid, and we can exploit properties~(ii)--(iii) to break our problem into smaller pieces. The following symmetries are at our disposal: translations along the lattice, by any amount of sites, and rotations around the $z$-axis, generated by $S^z$. Thus the symmetry group is
\begin{equation}\label{s2:eq:Ham symm group}
	 G = \mathbb{Z}\times U(1)_z \subseteq \mathbb{Z}\times SU(2) \ .
\end{equation}
In mathematical terms these symmetries can be used to decompose $\cH$ into a direct sum of irreducible $G$-representations, or `sectors', which are preserved by the Hamiltonian. Let us see what this means concretely.
\begin{exercise}
Without reading any further, find the consequence of partial isotropy for $L=2$ by comparing the result of Exercise~\ref{s2:ex:Sz block diagonal} with \eqref{s2:eq:xxz Ham L=2}. What are the sectors corresponding to~$U(1)_z$?
\end{exercise}

\paragraph{$M$-particle sectors.} First we exploit the partial isotropy. $H$ and $S^z$ can be simultaneously diagonalized by (iii), so the eigenvectors of $S^z$ form a basis for $\cH$ in which the Hamiltonian is block diagonal. Let us show that it has the following form with respect to this basis:
\begin{equation}\label{s2:eq:block diagonal}
	H \ = \ \begin{pmatrix}	\ \tikz[scale=0.8]{\pgfmathsetmacro{\A}{-42}
		\fill[gray!30] (0,0) rectangle (\A:.2) rectangle +(\A:.6);
		\foreach \y in {-1,...,1} \draw (\A:.8+.5+.2*\y) node{$\cdot$};
		\fill[gray!30] (\A:.8+1) rectangle +(\A:2.4);
		\foreach \y in {-1,...,1} \draw (\A:.8+1+2.4+.5+.2*\y) node{$\cdot$};
		\fill[gray!30] (\A:.8+1+2.4+1) rectangle ++(\A:.6) rectangle +(\A:.2);
	} \ \end{pmatrix} \ .
\end{equation}

The first block in \eqref{s2:eq:block diagonal} is $1\times1$ and corresponds to the \emph{pseudovacuum}, which we take to be
\begin{equation}\label{s2:eq:pseudovacuum}
	\ket{\Omega} \coloneqq \bigotimes_{l\in\mathbb{Z}_L} \ket{\uparrow}^{\phantom{j}}_l =  \ket{ \underset{1}{\uparrow} \To \underset{L}{\uparrow} } \in \cH \ .
\end{equation}
This vector happens to be a ground state of \eqref{s2:eq:Ham xxz} if $J\Delta>0$, as we will see in \eqref{s2:eq:xxz energy pseudovacuum}, but the point is that $\ket{\Omega}$ is an eigenvector of $S^z$ (with spin $\hbar\, L/2$) and killed by all $S_l^+$: it is a \emph{highest-weight vector}. This makes it a suitable reference point for constructing all other $S^z$-eigenvectors. For example, the second block in \eqref{s2:eq:block diagonal} is obtained by flipping any single spin,
\begin{equation}\label{s2:eq:M=1 basis}
	\ket{l} \coloneqq \hbar^{-1} \, S_l^- \ket{\Omega} = \ket{ \underset{1}{\uparrow} \To \uparrow \underset{l}{\downarrow} \uparrow \To \underset{L}{\uparrow} } \in \cH \ ,
\end{equation}
producing $L$ vectors (so the block has size $L\times L$) with spin $L/2-1$. Likewise the third block corresponds to flipping yet another spin; since $(S_l^-)^2=0$ and $S_k^- S_l^- = S_l^- S_k^-$ this yields $\binom{L}{2}$ different vectors $\ket{k,l}$ for $1\leq k<l \leq L$.

In general, by repeatedly applying lowering operators~$S_l^-$ to $\ket{\Omega}$ we construct an orthonormal basis describing configurations with $0 \leq M \leq L$ flipped spins:
\begin{equation}\label{s2:eq:M-part basis}
	\ket{l_1,\To,l_M} \coloneqq \hbar^{-M} \, S_{l_1}^- \cdots S_{l_M}^- \ket{\Omega} \in \cH \ , \qquad\qquad 1\leq l_1<\cdots<l_M \leq L\ .
\end{equation}
This is the \emph{coordinate basis} of $\cH$, which is responsible for the `coordinate' in `\textsc{cba}'. A nice aspect of this basis is that it is very physical; its elements can be depicted as in Figure~\ref{s2:fg:spin chain} (for which $M=5$). The price we pay is that we lose manifest periodicity by restricting ourselves to the `standard domain' $1\leq l_1<\cdots<l_M \leq L$ to avoid overcounting. Consequently the periodic boundary conditions~$\Vector{S}_{l+L} = \Vector{S}_l$ must be imposed explicitly when working with the coordinate basis. This will be important in Section~\ref{s2:method} and Appendix~\ref{sM}.

From \eqref{s2:eq:su(2) ladder} it follows that all spin configurations~\eqref{s2:eq:M-part basis} are eigenvectors of the total spin-$z$ operator:
\begin{equation}\label{s2:eq:M-part Sz action}
	S^z \, \ket{l_1,\To,l_M} = \hbar \, (L/2 - M) \,  \ket{l_1,\To,l_M} \ .
\end{equation}
Let us write $\cH_M\subseteq \cH$ for the \emph{$M$-particle sector} consisting of all vectors with $M$ spins down. The (weight) decomposition of our Hilbert space into these subspaces,
\begin{equation}\label{s2:eq:M-part decomposition}
	\cH = \bigoplus_{M=0}^L \cH_M \ ,
\end{equation}
corresponds to the block-diagonal form of $H$ in~\eqref{s2:eq:block diagonal}.
\begin{exercise}
Compute the size of the $M$th block in \eqref{s2:eq:block diagonal}. Check that the dimensions on both sides of \eqref{s2:eq:M-part decomposition} agree.
\end{exercise}

The upshot is that partial isotropy allows us to focus on diagonalizing the Hamiltonian in the $M$-particle sector: our new goal is to solve the eigenvalue problem
\begin{equation}\label{s2:eq:eigenproblem}
	H_\textsc{xxz} \, \ket{\Psi_M} = E_M \, \ket{\Psi_M} \ , \qquad \ket{\Psi_M} \in \cH_M \ .
\end{equation}

\paragraph{Magnons.} Next we exploit the homogeneity; let us see how far that gets us. By~(ii) the Hamiltonian satisfies $U H U^{-1} = H$ where the \emph{shift operator}~$U\in\End(\cH)$ shifts all  sites to the \choice{right}{left}, mapping each $V_l$ to \choice{$V_{l+1}$}{$V_{l-1}$}.$^\#$\footnote{In Section~\ref{s4:conserved quantities} we will see that \choice{our convention, where the shift operator acts by translation to the right, is}{defining the shift operator as acting by translations to the \emph{right} would be} more natural from the viewpoint of the \textsc{qism}, cf.~\eqref{s4:eq:t at u*}. \choice{This convention results}{However, using that convention would result} in a sign in the exponent in \eqref{s2:eq:magnon}, and similarly in e.g.\ \eqref{s2:eq:CBA M=2} and \eqref{s2:eq:CBA M} for higher~$M$. At any rate, this choice of convention essentially only affects the sign of the (quasi)momentum; of course the physical results do not depend on it.} In analogy with continuous (as opposed to lattice) models one often writes $U\eqqcolon \E^{\I P}$. Since $U$ is unitary its eigenvalues are of the form~$\E^{\I p}$ for some real momentum~$p$. (Note that $p$ is defined mod~$2\pi$.) Periodic boundary conditions imply that $U^L=\id$ is the identity operator on~$\cH$, leading to momentum quantization $p\in\frac{2\pi}{L}\mathbb{Z}_L$ as expected for particles on a circle.

\begin{exercise}
For the zero-particle sector~$\cH_0$ use homogeneity to find the momentum of the pseudovacuum~\eqref{s2:eq:pseudovacuum}. Check that $H_\textsc{xxz} \, \ket{\Omega} = E_0 \, \ket{\Omega}$ with `vacuum' energy
\begin{equation}\label{s2:eq:xxz energy pseudovacuum}
	E_0 = - \hbar^2 J \Delta \, L /4 \ .
\end{equation}
\end{exercise}

The one-particle sector is fixed by homogeneity as well. Indeed, any vector in $\cH_1$ can be expressed in terms of the coordinate basis. Translational invariance means that $\ket{\Psi_1;p} = \sum_l \Psi_p(l) \, \ket{l}$ should be an eigenvector of $U$ for some momentum~$p$. Using $U^\dag = U^{-1}$ it follows that the wave functions satisfy the recursion $\Psi_p(\choice{l-1}{l+1}) = \bra{l} U \ket{\Psi_1;p} = \E^{\I p}\Psi_p(l)$, yielding a plane-wave expansion:
\begin{equation}\label{s2:eq:magnon}
	\ket{\Psi_1;p} = \frac 1{\sqrt{L}} \sum_{l\in\mathbb{Z}_L} \E^{\choice{-}{}\I p \, l } \ket{l} \in \cH_1 \ .
\end{equation}
\begin{exercise}
Use $\braket{k}{l} = \delta_{k,l}$ to check that the $\#\mathbb{Z}_L = L = \dim(\cH_1)$ vectors \eqref{s2:eq:magnon} constitute an orthonormal basis for $\cH_1$.
\end{exercise}
The basis vectors~\eqref{s2:eq:magnon} describe excitations around the pseudovacuum~$\ket{\Omega}$ called \emph{magnons}: spin waves with quantized wavelength~$2\pi/p$ travelling along the chain. With respect to the magnon basis~\eqref{s2:eq:magnon} for~$\cH_1$ any translationally-invariant Hamiltonian is diagonal in the one-particle sector. Whether or not magnons are `quasiparticles' describing low-lying excitations in the physical spectrum depends on the model's parameters. 
\begin{exercise}
Compute the action of \eqref{s2:eq:Ham xxz via S^pm} on \eqref{s2:eq:M=1 basis} to check that the dispersion relation of a magnon in the \textsc{xxz} spin chain is
\begin{equation}\label{s2:eq:xxz energy magnon}
	\varepsilon_1(p) \coloneqq E_1(p) - E_0 = \hbar^2 J \, (\Delta-\cos p) \ .
\end{equation}
\end{exercise}
Notice that for the \textsc{xxx} case ($\Delta=1$) the magnon with vanishing momentum contributes zero to the energy. This is a direct consequence of the symmetries: $H_\textsc{xxx}$ has full rotational symmetry, so its eigenstates come in $\mathfrak{su}(2)$-multiplets. The zero-momentum magnon is simply the first $\mathfrak{su}(2)$-descendant of the pseudovacuum: $\ket{\Psi_1;0} \propto S^- \ket{\Omega}$ with $S^- = \sum_l S^-_l$ the total lowering operator. Turning on the anisotropy lifts the degeneracy in the spectrum.

For $M\geq2$ we can again write any translationally-invariant vector with momentum~$p$ as
\begin{equation}\label{s2:eq:M vector via coord basis}
	\ket{\Psi_M;p} = \sum_{1\leq l_1<\cdots<l_M \leq L} \!\! \Psi_p(l_1,\To,l_M) \,  \ket{l_1,\To,l_M} \in \cH_M
\end{equation}
with respect to the coordinate basis~\eqref{s2:eq:M-part basis} of $\cH_M$. This time, however, the wave functions~$\Psi_p(\Vector{l})$ are not completely determined by the symmetries~\eqref{s2:eq:Ham symm group}.$^\#$\footnote{For example, expand $\ket{\Psi_2;p}\in\cH_2$ as in \eqref{s2:eq:M vector via coord basis}. Homogeneity again recursively relates wave functions for excited spins at equal separation~$d_{12}\coloneqq d(l_1,l_2)$, where $d(k,l)\coloneqq \min_{n\in\mathbb{Z}} | l - k + n \, L |$ is the distance function (metric) on $\mathbb{Z}_L$. Indeed, $\Psi_p(\choice{l_1-1,l_2-1}{l_1+1,l_2+1}) = \bra{l_1,l_2} U \ket{\Psi_2;p} = \E^{\I p}\Psi_p(l_1,l_2)$, so $\Psi_p(l_1,l_2) = \E^{\choice{-}{} \I p\,l_1} \Psi'_p(d_{12}) = \E^{\choice{-}{}\I p\,l_2} \Psi''_p(d_{12})$ for some function $\Psi'_p$ (or equivalently $\Psi''_p$) depending on the lattice only through the separation between the two flipped spins. However, the values of this function for different~$d_{12}$ are not related by the symmetries~\eqref{s2:eq:Ham symm group}.} To diagonalize the larger blocks of the Hamiltonian we have to be smart: this is where the Bethe ansatz comes in.

\subsection{The coordinate Bethe ansatz}\label{s2:method}

In a ground-breaking paper from 1931, Bethe solved the ferromagnetic regime of the \textsc{xxx} spin chain~\cite{Bet31}. This subsection introduces the coordinate Bethe ansatz (\textsc{cba}) for any spin chain obeying properties (i)--(iii) above. The focus lies on the structure and physics of the method; computational details can be found in Appendix~\ref{sM}.

The basic idea of the \textsc{cba} is to parametrize the states in the $M$-particle sector via parameters~$p_m$, $1\leq m\leq M$. In the simplest case this boils down to the result~\eqref{s2:eq:magnon}, while in general it is a symmetrized version of the Fourier transform. The spectrum of the Hamiltonian follows once the values of the~$p_m$ are found from a system of coupled nonlinear equations: the \emph{Bethe-ansatz equations}~(\textsc{bae}). Thus, in principle, the \textsc{cba} provides a concrete and physical way of converting the problem of diagonalizing the Hamiltonian to that of solving the \textsc{bae}.

\paragraph{Trials for $M=2$.} To understand where the \textsc{cba} comes from we first consider the case~$M=2$. Expand $\ket{\Psi_2;p}$ as in \eqref{s2:eq:M vector via coord basis}. Note that we only have to define the wave function $\Psi_{p}(l_1,l_2)$ for~$1\leq l_1<l_2 \leq L$. Inspired by the result~\eqref{s2:eq:magnon} for~$M=1$ a reasonable first guess is to include \emph{two} parameters, $p_1$ and~$p_2$, and try
\begin{equation}\label{s2:eq:naive M=2}
	\Psi_{p_1,\,p_2}(l_1,l_2) \stackrel{?}{=} \Psi_{p_1}(l_1) \Psi_{p_2}(l_2) \propto \E^{\choice{-}{} \I (p_1 l_1 + p_2 l_2)} \ , \qquad\qquad l_1 < l_2 \ ,
\end{equation}
where an overall normalization, not depending on the lattice sites, is suppressed. The periodic boundary conditions require
\begin{equation}\label{s2:eq:periodicity M=2}
	\Psi_{p_1,\,p_2}(l_2,l_1+L) = \Psi_{p_1,\,p_2}(l_1,l_2) \ , \qquad\qquad 1 \leq l_1 < l_2 \leq L \ .
\end{equation}
where for \eqref{s2:eq:naive M=2} the left-hand side is given by
\begin{equation}\label{s2:eq:naive M=2 periodicity}
	\E^{\choice{-}{} \I p_2 L} \, \E^{\choice{-}{} \I (p_1 l_2 + p_2 l_1)} \ .
\end{equation}
\begin{exercise}
Check that for \eqref{s2:eq:naive M=2} the only solution to \eqref{s2:eq:periodicity M=2} is $p_1=p_2=0$.
\end{exercise}

To improve our guess for the wave functions we notice that \eqref{s2:eq:naive M=2 periodicity} describes two excited spins, like \eqref{s2:eq:naive M=2}, but with the positions~$l_m$ of the excitations with parameters~$p_m$ interchanged. Thus \eqref{s2:eq:naive M=2 periodicity} can interpreted as the result of scattering of the two excitations from \eqref{s2:eq:naive M=2}. Correcting our initial trial to take into account such scattering we thus try $A\, \E^{\choice{-}{}\I (p_1 l_1 + p_2 l_2)} + A'\, \E^{\choice{-}{}\I  (p_1 l_2 + p_2 l_1)}$. Through \eqref{s2:eq:periodicity M=2}, periodicity now relates the~$p_m$ to the coefficients as $\E^{\choice{-}{}\I p_1 L} = A'/A$ and $\E^{\choice{-}{}\I p_2 L} = A/A'$. Setting $A=A'$ would result in the~$p_m$ each taking values as for a single, free, magnon. To allow for interactions between the flipped spins we promote the coefficients to functions:
\begin{equation}\label{s2:eq:CBA M=2}
	\Psi_{p_1,\,p_2}(l_1,l_2) = A(p_1,p_2)\, \E^{\choice{-}{}\I (p_1 l_1 + p_2 l_2)} + A'(p_1,p_2)\, \E^{\choice{-}{}\I  (p_1 l_2 + p_2 l_1)} \ , \qquad\qquad l_1 < l_2 \ .
\end{equation}
This is the ansatz (hypothesis, educated guess) proposed by Bethe for the two-particle sector.
\begin{exercise}
Check that the vector $\ket{\Psi_2;p_1,p_2}$ given by \eqref{s2:eq:CBA M=2} has momentum $p=p_1+p_2$.
\end{exercise}

The \emph{two-body $S$-matrix} (which, despite its name, is just $1\times1$)
\begin{equation}\label{s2:eq:S-matrix}
	S(p_1,p_2) \coloneqq \frac{A'(p_1,p_2)}{A(p_1,p_2)}
\end{equation}
describes the scattering of the two excitations, as can be seen by rewriting \eqref{s2:eq:CBA M=2} in the form $\Psi_{p_1,\,p_2}(l_1,l_2) \propto \E^{\choice{-}{}\I (p_1 l_1 + p_2 l_2)} + S(p_1,p_2)\, \E^{\choice{-}{}\I(p_1 l_2 + p_2 l_1)}$. 

The periodic boundary conditions \eqref{s2:eq:periodicity M=2} impose two equations:
\begin{equation}\label{s2:eq:BAE M=2}
	\E^{\choice{-}{}\I p_1 L} = S(p_1,p_2) \ , \qquad \E^{\choice{-}{}\I p_2 L} = S(p_1,p_2)^{-1} \ .
\end{equation}
These are the \textsc{bae} for the parameters~$p_m$ in the two-particle sector. Physically they say that when either of the excited spins is moved once around the chain (in clockwise direction) it scatters on the other excitation. Note that the \textsc{bae} together imply the periodicity condition~$\E^{\I (p_1 + p_2) L}=1$ and hence momentum quantization $p=p_1+p_2\in \frac{2\pi}{L}\mathbb{Z}_L$ for the two-particle sector. The dependence on the details of the spin chain are hidden in the two-body $S$-matrix on the right-hand side of the \textsc{bae}. Until a model is specified we cannot say whether the \textsc{bae} have the right amount of, or even \emph{any}, solutions for the~$p_m$. We will look at the results for the \textsc{xxz} model in Section~\ref{s2:results}.

Before proceeding to the $M$-particle sector let us quickly recap our notation for $M=2$. We write `$\ket{\Psi_2}$' for an arbitrary vector in the two-particle sector, `$\ket{\Psi_2;p}$' for any vector in $\cH_2$ that is translationally invariant (with momentum~$p$), and `$\ket{\Psi_2;p_1,p_2}$' for the specific vectors (with parameters $p_1$ and~$p_2$) given by \eqref{s2:eq:CBA M=2}.

\paragraph{\textsc{cba} for general $M$.} Expand $\ket{\Psi_M} \in \cH_M$  in terms of the coordinate basis~$\ket{l_1,\To,l_M}$ as in~\eqref{s2:eq:M vector via coord basis}. Again we associate to each excitation~$l_m$ a parameter~$p_m$ that is to be determined. We abbreviate $\Vector{l}\coloneqq (l_1,\To,l_M)$ and $\Vector{p}\coloneqq (p_1,\To,p_M)$.

By property~(i) above we only have (very) short-ranged interactions, so the $M$~excited spins do not interact when they are \emph{well separated}, i.e.\ when no two excitations are next to each other. Thus, for such well-separated configurations it makes sense to look for a wave function of the product form $\Psi_{p_1}(l_1)\cdots \Psi_{p_M}(l_M)\propto \exp(\choice{-}{}\I\Vector{p} \!\cdot \Vector{l})$: this is the generalization of \eqref{s2:eq:naive M=2} to the $M$-particle sector. 

Next, as for $M=2$, we include all scattered configurations, labelled by permutations in~$S_M$ describing the ordering after the scattering. A linear combination of these $M!$ configurations, with coefficients depending on $\Vector{p}$ to account for interactions between the excitations, gives the \textsc{cba} for the (unnormalized) wave function in the $M$-particle sector:
\begin{equation}\label{s2:eq:CBA M}
	\Psi_{\Vector{p}}(\Vector{l}) = \sum_{\pi\in S_M} A_\pi(\Vector{p})\, \E^{\choice{-}{}\I\Vector{p}_\pi \!\cdot \Vector{l}} \ , \qquad\qquad l_1 < \cdots < l_M \ .
\end{equation}
Here $\Vector{p}_\pi$ is a shorthand for the (right) action of $\pi\in S_M$ on~$\Vector{p}$; concretely this just means that $\Vector{p}_\pi \!\cdot \Vector{l} = \sum_m p_{\pi(m)}\,l_m$. The ansatz~\eqref{s2:eq:CBA M} is also referred to as the \emph{Bethe wave function}. As a check we note that \eqref{s2:eq:CBA M} reduces to the magnon~\eqref{s2:eq:magnon} when $M=1$, while for $M=2$ the sum in~\eqref{s2:eq:CBA M} runs over two elements, the identity and a transposition, correctly reproducing \eqref{s2:eq:CBA M=2}.

\paragraph{Strategy.} The Bethe wave functions \eqref{s2:eq:CBA M} yield eigenstates of the Hamiltonian in the $M$-particle sector if we can solve the equations
\begin{equation}\label{s2:eq:xxz eqns}
	\bra{\Vector{l}}H_\textsc{xxz}\ket{\Psi_M;\Vector{p}} = E_M(\Vector{p})\Psi_{\Vector{p}}(\Vector{l}) \ , \qquad\qquad 1\leq l_1 < \cdots < l_M \leq L \ .
\end{equation}
The unknowns are the energy eigenvalues~$E_M(\Vector{p})$ and the coefficients $A_\pi(\Vector{p})$ (up to an overall normalization) as functions of~$\Vector{p}$, together with the actual values of the~$\Vector{p}$. The strategy to determine these consists of three steps:
\begin{enumerate}
	\item Solve the equations in~\eqref{s2:eq:xxz eqns} for configurations~$\Vector{l}$ with well-separated excitations. This is quite easy and will yield the $M$-particle dispersion relation $\varepsilon_M(\Vector{p})\coloneqq E_M(\Vector{p})-E_0$.
	\item The equations in~\eqref{s2:eq:xxz eqns} with at least one pair of neighbouring excitations in~$\Vector{l}$. Although this is more tricky, it can be done, giving the coefficients $A_\pi(\Vector{p})/A_e(\Vector{p})$.
	\item Impose the periodic boundary conditions. This will result in one equation for each exponent in \eqref{s2:eq:CBA M}: the \textsc{bae}, a~priori $M!$ in total, for the allowed (`on shell') values of~$\Vector{p}$.
\end{enumerate}

Of course this is really only half of the work: the \textsc{bae} still have to be solved, which has not been done for general $M$ and $L$, and one has to let $L \to \infty$ to study the thermodynamic properties of the model. In addition it remains to be seen whether all eigenstates are of the Bethe-form, so that the \textsc{cba} does really produce the full spectrum. The above strategy is carried out for the $M$-particle sector of the \textsc{xxz} model in Appendix~\ref{sM}; let us turn to the results.

\subsection{Results and Bethe-ansatz equations}\label{s2:results}

For brevity we set $\hbar=J=1$. This essentially only affects the energy eigenvalues, which can be restored by multiplication by~$\hbar^2 J$.

\paragraph{Results for $M=2$.} We first collect the results of the strategy in the two-particle sector. The computations can be found in Appendix~\ref{sM}.

\subparagraph{Step 1.} For the well-separated case the equations \eqref{s2:eq:xxz eqns} are satisfied by the \textsc{cba}~\eqref{s2:eq:CBA M=2} provided the energy is given by the dispersion relation
\begin{equation}\label{s2:eq:xxz energy M=2}
	\varepsilon_2(p_1,p_2) = 2\,\Delta-\cos p_1 -\cos p_2 = \varepsilon_1(p_1) + \varepsilon_1(p_2)  \ .
\end{equation}
This is a nice result: the energy consists of two contributions, one from each of the magnons. Although the values of the~$p_m$ remain to be determined, and will be different from the free (single-magnon) case to give rise to an interaction energy, this means that the excited spins in the two-particle sector behave like magnons. In particular, the two parameters~$p_m$ can be interpreted as the \emph{quasimomenta} of these magnons.

\subparagraph{Step 2.} For neighbouring excitations the equations can be solved using a trick, see Appendix~\ref{sM}, yielding the two-body $S$-matrix~\eqref{s2:eq:S-matrix}:
\begin{equation}\label{s2:eq:xxz S-matrix}
	S(p_1,p_2) = -\frac{1-2\,\Delta\E^{\choice{-}{}\I p_2} + \E^{\choice{-}{}\I (p_1+p_2)}}{1-2\,\Delta\E^{\choice{-}{}\I p_1} + \E^{\choice{-}{}\I (p_1+p_2)}} 
\end{equation}
Two-body scattering is unitary by virtue of the property $|S(p_1,p_2)|^2=1$ for $p_m\in\mathbb{R}$. This property, sometimes called \emph{physical unitarity}, suggests defining the (real-valued) function $\Theta(p_1,p_2) \coloneqq \choice{}{-}\I \log S(p_1,p_2)$ known as the \emph{two-body scattering phase}.

The result \eqref{s2:eq:xxz S-matrix} has two important physical implications. As $S(p_2,p_1) = S(p_1,p_2)^{-1}$ the Bethe wave function $\Psi_{p_1,p_2}(l_1,l_2)$ is symmetric in the two $p_m$ upon normalizing~\eqref{s2:eq:CBA M=2} by $A(p_1,p_2)=S(p_1,p_2)^{-1/2}$. Thus $\ket{\Psi_2;p_1,p_2}=\ket{\Psi_2;p_2,p_1}$: the magnons obey \emph{Bose-Einstein statistics}. Interestingly, \eqref{s2:eq:xxz S-matrix} also satisfies the fermion-like property $S(p_1,p_1)=-1$. Therefore $\ket{\Psi_2;p_1,p_1}=0$, yielding a \emph{Pauli exclusion principle} for the quasimomenta.$^\#$\footnote{Anticipating the Pauli exclusing principle some authors include a sign in the \textsc{cba} from the start, replacing \eqref{s2:eq:CBA M} by $\Psi_{\Vector{p}}(\Vector{l}) = \sum_{\pi\in S_M} (-1)^{\sgn(\pi)} A_\pi(\Vector{p})\E^{\choice{-}{}\I\Vector{p}_\pi \!\cdot \Vector{l}}$.} However, there is no spin-statistics connection in $1+1$~dimensions, so these two properties are compatible.

\subparagraph{Step 3.} The remaining task is to find equations for the values of the $p_m$. Plugging \eqref{s2:eq:xxz S-matrix} into \eqref{s2:eq:BAE M=2} we obtain the \textsc{bae} for the two-particle sector of the \textsc{xxz} model:
\begin{equation}\label{s2:eq:xxz BAE M=2}
\begin{aligned}
	& \E^{\choice{-}{}\I p_1 L} = -\frac{1-2\,\Delta\E^{\choice{-}{}\I p_2} + \E^{\choice{-}{}\I (p_1+p_2)}}{1-2\,\Delta\E^{\choice{-}{}\I p_1} + \E^{\choice{-}{}\I (p_1+p_2)}} = \E^{-\I \Theta(p_1,p_2)} \ , \\
	& \E^{\choice{-}{}\I p_2 L} = -\frac{1-2\,\Delta\E^{\choice{-}{}\I p_1} + \E^{\choice{-}{}\I (p_1+p_2)}}{1-2\,\Delta\E^{\choice{-}{}\I p_2} + \E^{\choice{-}{}\I (p_1+p_2)}} = \E^{-\I \Theta(p_2,p_1)} \ .
\end{aligned}
\end{equation}
Taking the logarithm shows that these are just quantization conditions for the quasimomenta:
\begin{equation}\label{s2:eq:BAE M=2 log form}
	L \, p_1 = 2\pi I_1 - \Theta(p_1,p_2) \ , \qquad L \, p_2 = 2\pi I_2 + \Theta(p_1,p_2) \ , \qquad\qquad I_m \in\mathbb{Z}_L \ ,
\end{equation}
where the $I_m$ are known as the \emph{Bethe quantum numbers}. Thus the quasimomenta in the two-particle sector are no longer the `bare' quantities (valued in $\frac{2\pi}L\mathbb{Z}_L$) of a free theory: \eqref{s2:eq:BAE M=2 log form} shows that they are `renormalized' as a result of interactions between the two magnons.

\begin{exercise}\label{s2:ex:arctan relation}
Find an expression for $\Theta(p_1,p_2)$ by multiplying the numerator and denominator in \eqref{s2:eq:xxz S-matrix} by $\E^{\choice{}{-}\I (p_1+p_2)/2}$ and using $\log\frac{u+\I v}{u-\I v}=2\I\arctan\frac{v}{u}$. 
\end{exercise}

\begin{exercise}
Check \eqref{s2:eq:xxz energy M=2} and \eqref{s2:eq:xxz S-matrix} without consulting Appendix~\ref{sM}.
\end{exercise}

\paragraph{Two-particle spectrum for $\Delta=1$.} To see whether we have succeeded in diagonalizing the Hamiltonian for~$M=2$ one has to check whether the \textsc{bae} admit $\dim(\cH_2) = \binom{L}{2}$  solutions giving rise to linearly independent states. To get some feeling for the physics we briefly discuss the spectrum that Bethe found for the ferromagnetic \textsc{xxx} model. More details, including some nice plots, can be found in~\cite{KM98} (for the antiferromagnetic case see~\cite{KHM98,*KHM00}). The solutions fall into three classes:
\begin{itemize}
\item[i)] There are $L$ solutions describing a superposition of two free magnons, $p_1=0$ and $p=p_2= 2\pi I_2/L$. These can also be understood as the $\mathfrak{su}(2)$-descendants of the states in the one-particle sector, $\sum_{l_1<l_2}(\E^{\choice{-}{}\I p \, l_1}+\E^{\choice{-}{}\I p \, l_2})\ket{l_1,l_2}\propto S^- \ket{\Psi_1;p}$.
\item[ii)] The remaining solutions $0<p_1<p_2<2\pi$ can be interpreted as nearly free superpositions of magnons whose interactions vanish for $L\to\infty$. Together with the first class these are \emph{scattering states}. 
\end{itemize}
However, there are not enough of these solutions. To find the remaining states the \textsc{cba} has to be improved by extending the quasimomenta to \emph{complex} values, $p_m\in\mathbb{C}$. Unitarity of the shift operator~$U=\E^{\I P}$ requires the total momentum~$p=p_1+p_2$ to remain real. These account for the third class of solutions:
\begin{itemize}
\item[iii)] Quasimomenta with $\Im(p_1)=-\Im(p_2)\neq 0$ cause $|\Psi_{\Vector{p}}(l_1,l_2)|$ to decrease with increasing separation between the magnons. These solutions can be interpreted as \emph{bound states}.
\end{itemize}
Such results were confirmed by neutron-scattering experiments, also for other models solvable by Bethe-ansatz techniques, see e.g.~\cite{ME+13,*GB+13,*KWW06}.

Let us briefly comment on a source of possible confusion. The appearance of the second $\mathfrak{su}(2)$-descendent of the pseudovacuum in the $M=2$ spectrum (with $p_1=p_2=0$) may appear to conflict with the Pauli exclusion principle for the two-body $S$-matrix. This issue is resolved by noticing that in the isotropic case the two-body $S$-matrix,
\begin{equation}\label{s2:eq:xxx S-matrix}
	S(p_1,p_2)|_{\Delta=1} = -\frac{1-2\,\E^{\choice{-}{}\I p_2} + \E^{\choice{-}{}\I (p_1+p_2)}}{1-2\,\E^{\choice{-}{}\I p_1} + \E^{\choice{-}{}\I (p_1+p_2)}} = \frac{\choice{-}{}\cot \tfrac{p_1}{2} \choice{+}{-}\cot \tfrac{p_2}{2} + 2 \I}{\choice{-}{}\cot \tfrac{p_1}{2} \choice{+}{-} \cot \tfrac{p_2}{2} - 2 \I} \ ,
\end{equation}
is not continuous at the origin. Indeed, along the diagonal in the quasimomentum plane we have $S(p_1,p_1) = -1$, cf.~the Pauli exclusion principle. However, along either axis \eqref{s2:eq:xxx S-matrix} satisfies $S(p_1,0)|_{\Delta=1} = S(0,p_2)|_{\Delta=1} = +1$, so the $\mathfrak{su}(2)$-descendants do not vanish. This discontinuity hints at the fact that the \textsc{xxz} model with `generic' anisotropy $\Delta\neq\pm 1$ is mathematically better behaved than the isotropic spin chain (cf.~Appendix~\ref{sY}). From this perspective $\Delta$ can be seen as a regulator for the \textsc{xxx} spin chain.

\begin{exercise}
Check that, as $\Delta\to1$, the result of Exercise~\ref{s2:ex:arctan relation} matches the expression for $\Theta(p_1,p_2)|_{\Delta=1}$ obtained directly from \eqref{s2:eq:xxx S-matrix}.
\end{exercise}

\paragraph{Results for general $M$.} Working out the strategy for the general $M$-particle sector, see Appendix~\ref{sM}, gives the following results.

\subparagraph{Step 1.} The equations \eqref{s2:eq:xxz eqns} are solved by the \textsc{cba}~\eqref{s2:eq:CBA M} for well-separated excitations if the contribution to the energy is
\begin{equation}\label{s2:eq:xxz energy M}
	\varepsilon_M(\Vector{p}) = M\,\Delta - \sum_{m=1}^M \cos p_m = \sum_{m=1}^M \varepsilon_1(p_m) \ .
\end{equation}
Thus, the dispersion relation behaves additively for general~$M$ as well: the energy splits into separate contributions for each~$p_m$. Let us stress once more that this does not mean that~\eqref{s2:eq:xxz energy M} is simply the sum of free-magnon contributions; the quasimomenta~$p_m$ are determined by the \textsc{bae} to account for the interaction energy. At any rate, \eqref{s2:eq:xxz energy M} does justify our quasiparticle interpretation for all~$M$, so that we may conclude that $M$ counts the number of magnons. In particular, since the $\cH_M$ are preserved by the Hamiltonian, the magnon-number is conserved: there is no magnon production or annihilation.

\subparagraph{Step 2.} The coefficients $A_\pi(\Vector{p})$ in the Bethe wave function~\eqref{s2:eq:CBA M} also have to satisfy \eqref{s2:eq:xxz eqns} for one or more pairs of neighbouring excitations. For general $M$ there are many more such equations than unknowns, but it turns out that they can in fact be solved. Up to an overall factor the $A_\pi$ are expressed in terms of the two-body $S$-matrix~\eqref{s2:eq:xxz S-matrix}:
\begin{equation}\label{s2:eq:A solution}
	\frac{A_\pi(\Vector{p})}{A_e(\Vector{p})} \ = \ \prod_{\substack{1 \leq m < m' \leq M \\ \text{s.t.\ } \pi(m)>\pi(m')}} S(p_m,p_{m'}) \ .
\end{equation}
This is a very remarkable result: physically, \eqref{s2:eq:A solution} says that $M$-magnon scattering is \emph{two-body reducible}, i.e.\ factors into successive two-magnon scattering processes, governed by the two-body $S$-matrix~\eqref{s2:eq:xxz S-matrix}. This is a extremely useful aspect of the model; we essentially know everything about the many-body scattering of magnons once we understand two-magnon scattering. Models exhibiting this property are rather special. Indeed, note that if the scattering were \emph{one}-body reducible, magnons would move freely along the spin chain. The \textsc{xxz} model, with \emph{two}-body reducible scattering, is just one level up in complexity, allowing one to study interesting dynamics in a controlled setting.

The results that there is no magnon-production and that scattering is two-body reducible hint at the existence of hidden symmetries and conserved quantities that render the \textsc{xxz} spin chain quantum integrable, see Section~\ref{s5:fact scatt}. This is indeed the case; we will find these conserved quantities in Sections \ref{s3:results} and \ref{s4:conserved quantities}.

There is a nice graphical way to understand the result~\eqref{s2:eq:A solution}~\cite{LL63}. Depict the initial and final configurations of quasimomenta (magnons) in the scattering process~$\pi\in S_M$ as follows:
	\[ \begin{array}{lcccc} \text{final:} \qquad & p_{\pi(1)} & p_{\pi(2)} & \cdots & p_{\pi(M)} \\ & & & & \\ & & & & \\ \text{initial:} \qquad & p_1 & p_2 & \cdots & p_M \end{array} \]
Now connect equal quasimomenta by arrows in such a way that there are no points where three or more lines meet. Typically there are several ways in which this can be done; these (must and do) give the same result. For every pair~$n<m$ that is switched by~$\pi$ there is a crossing
	$\begin{tikzpicture}[>=latex,scale=0.5,font=\small]
		\draw (-125:.3) -- (55:.3);
		\draw (-55:.3) -- (125:.3);
		\fill [black] (0,0) circle (.5ex);
	\end{tikzpicture}$
of $p_n$ and $p_m$, contributing a factor of $S(p_n,p_m)$ to $A_\pi(\Vector{p})$. For example, the coefficient~$A_\pi$ describing a three-magnon scattering process is depicted in Figure~\ref{s2:fg:factorized scattering}. 
\begin{exercise}
Rewrite the result \eqref{s2:eq:A solution} in terms of the two-body scattering phase for the normalization $A_e(\Vector{p}) = \prod_{m<m'} S(p_m,p_{m'})^{-1/2}$.
\end{exercise}
\begin{exercise}
Prove the Pauli exclusion principle for the quasimomenta in the $M$-particle sector by showing that \eqref{s2:eq:CBA M} with \eqref{s2:eq:A solution} vanishes whenever $p_m=p_n$ for some $m\neq n$.
\end{exercise}

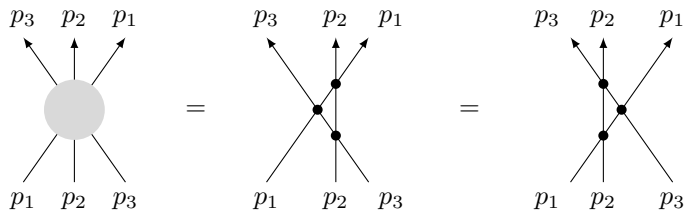
\begin{figure}[h]
	\centering
	\begin{tikzpicture}[>=latex,scale=0.8,font=\small]
	\pgfmathsetmacro\csc{1/sin(55)}
		\begin{scope}[shift={(-4,0)}]
			\draw[->] (-125:1.2*\csc) node[below]{$p_1$} -- (55:1.2*\csc) node[above]{$p_1$};
			\draw[->] (-90:1.2) node[below]{$p_2$} -- (90:1.2) node[above]{$p_2$};
			\draw[->] (-55:1.2*\csc) node[below]{$p_3$} -- (125:1.2*\csc) node[above]{$p_3$};
			
			\fill [gray!30] (0,0) circle (.5cm);
		\end{scope}
		
		\draw (-2,0) node{$=$};
		
		\draw[->] (-125:1.2*\csc) coordinate(a1) node[below]{$p_1$} -- (55:1.2*\csc) coordinate(b3) node[above right]{$p_1$};
		\draw[->] ($(-90:1.2)+(.3,0)$) coordinate(a2) node[below]{$p_2$} -- ($(90:1.2)+(.3,0)$) coordinate(b2) node[above]{$p_2$};
		\draw[->] (-55:1.2*\csc) coordinate(a3) node[below right]{$p_3$} -- (125:1.2*\csc) coordinate(b1) node[above]{$p_3$};
		
		\fill [black] (intersection of a1--b3 and a2--b2) circle (.5ex);
		\fill [black] (intersection of a1--b3 and a3--b1) circle (.5ex);
		\fill [black] (intersection of a2--b2 and a3--b1) circle (.5ex);
		
		\draw (2.5,0) node{$=$};
		
		\begin{scope}[shift={(5,0)}]
			\draw[->] (-125:1.2*\csc) coordinate(A1) node[below left]{$p_1$} -- (55:1.2*\csc) coordinate(B3) node[above]{$p_1$};
			\draw[->] ($(-90:1.2)-(.3,0)$) coordinate(A2) node[below]{$p_2$} -- ($(90:1.2)-(.3,0)$) coordinate(B2) node[above]{$p_2$};
			\draw[->] (-55:1.2*\csc) coordinate(A3) node[below]{$p_3$} -- (125:1.2*\csc) coordinate(B1) node[above left]{$p_3$};
			
			\fill [black] (intersection of A1--B3 and A2--B2) circle (.5ex);
			\fill [black] (intersection of A1--B3 and A3--B1) circle (.5ex);
			\fill [black] (intersection of A2--B2 and A3--B1) circle (.5ex);
		\end{scope}
	\end{tikzpicture}
	\caption{Diagrammatic representation of (the coefficient~$A_\pi$ describing) three-body scattering where magnons 1 and 3 are interchanged. There are two ways in which this can be done, so the second equality expresses a consistency condition --- which is trivially satisfied since the two-body $S$-matrix is $1\times 1$.}
	\label{s2:fg:factorized scattering}
\end{figure}

\subparagraph{Step 3.} Finally periodic boundary conditions must be imposed on the Bethe wave function~\eqref{s2:eq:CBA M} with coefficients~\eqref{s2:eq:A solution}. It turns out that there are $M$ independent \textsc{bae} for the quasimomenta $\Vector{p}$ in the $M$-particle sector:
\begin{equation}\label{s2:eq:xxz BAE M}
	\E^{\choice{-}{}\I p_m L} = \prod_{\substack{n=1 \\ n\neq m}}^M S(p_{n},p_m) \ , \qquad\qquad 1 \leq m \leq M \ .
\end{equation}
Like \eqref{s2:eq:BAE M=2} these have a nice physical interpretation. Since $p_m$ is the quasimomentum of the $m$th magnon, $\E^{\choice{-}{} \I p_m L}$ is the phase acquired by that magnon as it is moved once around the spin chain (in clockwise direction). The \textsc{bae} \eqref{s2:eq:xxz BAE M} say that this phase consists of contributions due to scattering on the other magnons.

To sum up, under the assumption that all states are of the Bethe form, the \textsc{cba} converts the problem of diagonalizing the Hamiltonian in the $M$-particle sector to that of solving the $M$~coupled equations \eqref{s2:eq:xxz BAE M} for $\Vector{p}\in\mathbb{C}^M$. The \textsc{bae} are rather complicated (see also Appendix~\ref{sY}), but that was to be expected: no approximations were used to obtain them; they are \emph{exact}. The \textsc{bae} can be studied numerically as well as analytically. By plugging in the resulting on-shell values for the~$p_m$ one finally obtains the actual eigenstates~$\ket{\Psi_M}$ and their energies.

Since the \textsc{bae} are usually hard to solve one may wonder what we have gained by all of this. Notice that, although the \textsc{bae} become more complicated as $M$ increases, they are not so sensitive to the length of the spin chain, unlike the size $\binom{L}{M}$ of the $M$-particle block of the Hamiltonian. This renders them useful for studying the ground states, elementary excitations and several thermodynamic quantities even as $L\to\infty$ (under certain assumptions on the nature of the solutions in that limit).

\paragraph{Rapidities.} To conclude this subsection we introduce alternative variables that are in some sense more natural than the quasimomenta~$p_m$. Indeed, due to the factorization of magnon scattering the two-body $S$-matrix plays an important role in the analysis of the model. Thus it is convenient to switch to coordinates for which the two-body $S$-matrix takes a simple form. 

For the \textsc{xxx} spin chain, \eqref{s2:eq:xxx S-matrix} suggests defining \emph{rapidities} as
\begin{equation}\label{s2:eq:xxx lambda}
	\lambda_m \coloneqq \choice{-}{} \frac 12 \cot\frac{p_m}{2}
\end{equation}
so that the two-body $S$-matrix simply depends on the rapidity difference (cf.~Section~\ref{s5:fact scatt}),
\begin{equation}\label{s2:eq:xxx S via lambda}
	S(\lambda_n,\lambda_m)|_{\Delta=1} = \frac{\lambda_m-\lambda_n+\I}{\lambda_m-\lambda_n-\I} \ .
\end{equation}
(Depending on the context, rescaled or shifted versions of these are also used in literature.)
\begin{exercise}
Invert \eqref{s2:eq:xxx lambda} and use \eqref{s2:eq:xxz energy magnon} to check that the quasimomentum and energy contribution of a magnon with rapidity~$\lambda$ are
\begin{equation}\label{s2:eq:xxx p and E via lambda} 
	p(\lambda)|_{\Delta=1} = \choice{\I}{\frac 1\I} \log \frac{\lambda+\I/2}{\lambda-\I/2} \ , \qquad \varepsilon_1(\lambda)|_{\Delta=1} = \frac{2}{4 \, \lambda^2 + 1} \ .
\end{equation}
\end{exercise}
We conclude that in terms of rapidities the \textsc{bae}~\eqref{s2:eq:xxz BAE M} for the $M$-particle sector of the \textsc{xxx} spin chain read
\begin{equation}\label{s2:eq:xxx BAE via lambda}
	\biggl( \frac{\lambda_m+\I/2}{\lambda_m-\I/2} \biggr)^{\!L} = \ \prod_{\substack{n=1 \\ n\neq m}}^M \frac{\lambda_m-\lambda_n+\I}{\lambda_m-\lambda_n-\I} \ , \qquad\qquad 1 \leq m \leq M \ .
\end{equation}

In the regime $|\Delta| \leq 1$ the \textsc{xxz} spin chain involves hyperbolic generalizations of the above expressions: parametrizing the anisotropy as $\Delta = \cos\gamma$ we have
\begin{align}
	& p(\lambda) = \choice{\I}{\frac 1\I} \log \frac{\sinh(\lambda+\I\gamma/2)}{\sinh(\lambda-\I\gamma/2)} \ , \label{s2:eq:xxz p via lambda} \\
	& \varepsilon_1(\lambda) = \frac{1}{2} \frac{\sin^2\gamma}{\sinh(\lambda + \I \gamma/2) \, \sinh(\lambda - \I \gamma/2)} \ , \label{s2:eq:xxz E via lambda}
\end{align}
and the \textsc{bae} become
\begin{equation}\label{s2:eq:xxz BAE via lambda}
	\biggl( \frac{\sinh(\lambda_m+\I\gamma/2)}{\sinh(\lambda_m-\I\gamma/2)} \biggr)^{\!L} = \ \prod_{\substack{n=1 \\ n\neq m}}^M \frac{\sinh(\lambda_m-\lambda_n+\I\gamma)}{\sinh(\lambda_m-\lambda_n-\I\gamma)} \ , \qquad\qquad 1 \leq m \leq M \ .
\end{equation}

\begin{exercise}
Compute $p'(\lambda_m)$ and compare the result with $\varepsilon_1(\lambda_m)$. (We will understand where this relation comes from in Section~\ref{s4:ABA}.)
\end{exercise}

\begin{exercise}\label{s2:ex:isotropic limit}
Note that \eqref{s2:eq:xxz p via lambda}--\eqref{s2:eq:xxz BAE via lambda} become trivial in the isotropic limit. Check that the equations for the \textsc{xxx} spin chain can nevertheless be recovered by rescaling the rapidities as $\lambda \mapsto \gamma \lambda$ before taking the limit~$\gamma\to 0$.
\end{exercise}

\section{Transfer matrices and the six-vertex model}\label{s3}

Before turning to the abstract algebraic but very powerful formalism to treat the \textsc{xxz} model once more in Section~\ref{s4} it is insightful to switch to the world of \emph{classical} statistical physics on a two-dimensional lattice. We focus on the six-vertex model, which will turn out to be intimately related to the \textsc{xxz} spin chain. Several concepts that we encounter along the way will play an important role in Section~\ref{s4} too.

\subsection{The six-vertex model}\label{s3:six-vertex model}

Any lattice can be turned into a statistical model by assigning some microscopic degrees of freedom to the lattice and specifying a rule $C\longmapsto w(\beta,C)$ that gives for each microscopic configuration~$C$ a temperature-dependent weight~$w(\tau,C)$, where $\tau\coloneqq k_\text{B}T$ as usual. Often these are Boltzmann weights, $w_\text{B}(\tau,C) \coloneqq \exp\bigl(-E(C)/\tau\bigr)$, and the energy~$E(C)$ of a configuration determines its weight. The main object in statistical physics is the partition function
\begin{equation}\label{s3:eq:partition function}
	Z(\tau) = \sum_C w(\tau,C)
\end{equation}
governing the statistical properties of the model.

A well-known class of examples are \emph{Ising models}, where the microscopic degrees of freedom are \emph{discrete} `spin' variables $\varepsilon_l =\pm1$ at the vertices of the lattice, labelled by~$l$, and the weights are determined by the energy $E(C)$ of the `spin' configurations~$C=\{\varepsilon_l\}_l$. These models describe molecules with highly anisotropic interactions in crystals. A discussion of several exactly solvable Ising models can be found in~\cite{Bax07}.
\begin{exercise}
Check that the (anti)ferromagnetic Ising model is in fact obtained from the \textsc{xxz} spin chain in the anisotropic limits $\Delta\to\pm\infty$.
\end{exercise}

More generally spin chains from Section~\ref{s2:spin chains} also fit within this formalism. The lattice is $\mathbb{Z}_L$, and the local degrees of freedom are the quantum-mechanical spins in the local spaces $V_l\simeq \mathbb{C}^2$. For a(n) (eigen)configuration $C\in\cH =\bigotimes_l V_l$ of spins the energy $E(C)$ is the eigenvalue of the Hamiltonian~$H$, from which we find the corresponding Boltzmann weight. The partition function arises as a trace over $\cH$:
\begin{equation}\label{s3:eq:partition function QM}
	Z(\tau)=\tr \, \exp\bigl(-H(C)/\tau\bigr) \ .
\end{equation}

For any model the goal is to get a grip on the typically \emph{huge} sum in~\eqref{s3:eq:partition function}. Indeed, interesting thermodynamics, like phase transitions, is related to non-smooth behaviour of $Z$ in~$1/\tau$. The weights usually depend smoothly on the temperature, so this can only occur in the limit where the lattice becomes infinite. For some statistical models there are methods that, in principle, allow for an \emph{exact} evaluation of \eqref{s3:eq:partition function}. The six-vertex model is an example of such an exactly solved model, and as we will see it can be tackled with the \textsc{cba} from Section~\ref{s2:method}. Since the thermodynamics will not be relevant for us the dependence on the temperature~$\tau$ will be suppressed from now on.

\paragraph{Vertex models.} The model that we are going to study is an example of a vertex model in two dimensions.$^\#$\footnote{In contrast to the quantum-mechanical spin chains, classical statistical models on a \emph{one}-dimensional lattice are usually not so interesting. Instead, the intriguing systems in statistical physics live in \emph{two} spatial dimensions. Actually, since $1+1=2+0$, this is not very surprising: through the (time-dependent) Schr\"odinger equation spin chains are really $(1+1)$-dimensional, whereas time does not play a role for statistical models in thermal equilibrium. In Section~\ref{s5:fact scatt} we will see that $2d$ is also special in \textsc{qft}.} Consider a finite square lattice consisting of $L$~rows and $K$~columns, with uniform lattice spacing. We impose periodic boundary conditions in both directions, yielding a discrete torus $\mathbb{Z}_L\times\mathbb{Z}_K$. The microscopic degrees of freedom are `spins', as for Ising models, but this time they are not assigned to the vertices of the lattice but rather to the edges, as shown on the left in Figure~\ref{s3:fg:example}.

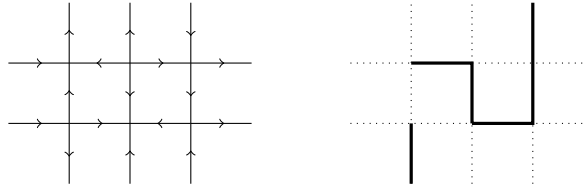
\begin{figure}[h]
	\centering
	\begin{minipage}{4.5cm}\centering
	\begin{tikzpicture}[scale=0.8,decoration={markings,mark=at position 0.55 with {\arrow{>}}}]
		\draw[postaction={decorate}] (0,1) -- (1,1);
		\draw[postaction={decorate}] (0,2) -- (1,2);
		
		\draw[postaction={decorate}] (1,1) -- (1,0);
		\draw[postaction={decorate}] (1,1) -- (1,2);
		\draw[postaction={decorate}] (1,2) -- (1,3);
			
		\draw[postaction={decorate}] (1,1) -- (2,1);
		\draw[postaction={decorate}] (2,2) -- (1,2);
		
		\draw[postaction={decorate}] (2,0) -- (2,1);
		\draw[postaction={decorate}] (2,2) -- (2,1);
		\draw[postaction={decorate}] (2,2) -- (2,3);
		
		\draw[postaction={decorate}] (3,1) -- (2,1);
		\draw[postaction={decorate}] (2,2) -- (3,2);
		
		\draw[postaction={decorate}] (3,0) -- (3,1);
		\draw[postaction={decorate}] (3,2) -- (3,1);
		\draw[postaction={decorate}] (3,3) -- (3,2);
		
		\draw[postaction={decorate}] (3,1) -- (4,1);
		\draw[postaction={decorate}] (3,2) -- (4,2);
	\end{tikzpicture}
	\end{minipage}\begin{minipage}{4.5cm}\centering
	\begin{tikzpicture}[scale=0.8]
		\draw[dotted] (0,1) -- (4,1) (0,2) -- (4,2) (1,0) -- (1,3) (2,0) -- (2,3) (3,0) -- (3,3);
		\draw[very thick] (1,0) -- (1,1) (1,2) -| (2,1) (2,1) -| (3,3);		
	\end{tikzpicture}
	\end{minipage}
	\caption{Example of a configuration of microscopic `spins' $\varepsilon = \pm1$ on the edges in a portion of a two-dimensional lattice. On the left the `spins' are indicated by arrows, with $\uparrow$ and $\rightarrow$ for $\varepsilon=-1$ and $\downarrow$ and $\leftarrow$ for $\varepsilon=+1$; on the right these values are represented by a dotted and thick line, respectively.}
	\label{s3:fg:example}
\end{figure}

For a \emph{vertex model} the weight of a configuration~$C$ on the entire lattice is obtained as the product of \emph{vertex weights}~$w(C,v)$ assigned to the vertices~$v$ of the lattice:
\begin{equation}\label{s3:eq:vertex weights}
	w(C)=\prod_{v\in\mathbb{Z}_L\times\mathbb{Z}_K} \!\! w(C,v) \ .
\end{equation}
For \emph{nearest-neighbour} interactions the vertex weights only depend on the four `spins' surrounding the vertex. If, in addition, the model is \emph{homogeneous} (translationally invariant in both directions) the vertex weights can be denoted as follows: given `spin' variables $\alpha,\beta,\gamma,\delta\in\{\pm 1\}$ on the four edges surrounding~$v$ as shown in Figure~\ref{s3:fg:vertex model} we write $w(C,v)=\w{\alpha}{\beta}{\gamma}{\delta}$. There are sixteen vertex weights~$\w{\alpha}{\beta}{\gamma}{\delta}$ that have to be specified, corresponding to the possible configurations of the `spins' on the surrounding edges.

\begin{figure}[h]
	\centering
	\begin{tikzpicture}[scale=0.8,font=\small]
	\draw (0,1) node[left]{$\beta$} -- (2,1) node[right]{$\delta$};
	\draw (1,0) node[below]{$\alpha$} -- (1,2) node[above]{$\gamma$};
	\end{tikzpicture}
	\caption{A vertex $v \in \mathbb{Z}_L \times \mathbb{Z}_K$ with `spin' variables $\alpha,\beta,\gamma,\delta \in \{\pm 1\}$ on the surrounding edges.}
	\label{s3:fg:vertex model}
\end{figure}
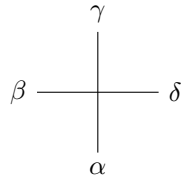

\paragraph{Main example.} The \emph{six-vertex} or \emph{ice-type} model describes hydrogen-bonded crystals. The vertices of the lattice represent larger atoms, oxygen in the case of water ice, and the edges model hydrogen bonds. (The square lattice is a reasonable two-dimensional approximation of the hexagonal structure of ice crystals found in nature, depicted in Figure~\ref{s3:fg:ice}.) The `spin' on the edge encodes at which end of each bond the proton is, say with `spin'~$-1$ corresponding to the right (top) of a horizontal (vertical) edge as in Figure~\ref{s3:fg:example}. For electric neutrality each oxygen atom should have precisely two hydrogen atoms close by. This translates to the \emph{ice rule} $\alpha+\beta=\gamma+\delta$, which leaves us with the six `allowed' vertices shown in Figure~\ref{s3:fg:six vertices}. For example, in Figure~\ref{s3:fg:example} the ice rule is only satisfied for the two vertices on the right.

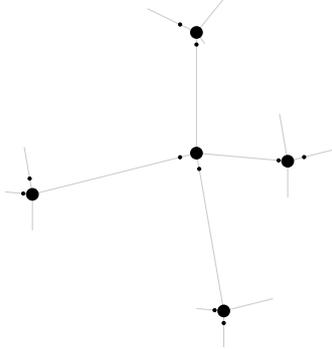
\begin{figure}[h]
	\centering
	\begin{tikzpicture}[x={(-0.866cm,-0.5cm)}, y={({1.5*0.866cm},-0.5cm)}, z={(0cm,1cm)}, scale=.8, font=\small]
		\pgfmathsetmacro\A{sin(120)}
		\pgfmathsetmacro\B{cos(120)}
		\pgfmathsetmacro\C{.1}
		
		\coordinate(v0) at (0,0,0); 
		\coordinate(v1) at ($2*(0,0,1)$);
		\coordinate(v2) at ($2*(-\A,0,\B)$); 
		\coordinate(v3) at ($2*(-\A*\B,-\A^2,\B)$); 
		\coordinate(v4) at ($2*(\A^3-\A*\B^2,-2*\A^2*\B,\B)$); 

		\coordinate(v1bar) at ($2*(0,0,-1)$);
		\coordinate(v2bar) at ($2*(-\A,0,-\B)$);
		\coordinate(v3bar) at ($2*(-\A*\B,-\A^2,-\B)$);
		\coordinate(v4bar) at ($2*(\A^3-\A*\B^2,-2*\A^2*\B,-\B)$);
		
		\draw[gray!40] (v1) -- (v0) -- (v2) -- ($(v2)+3*\C*(v1bar)$);
		\draw[gray!40] ($(v3)+3*\C*(v1bar)$) -- (v3) -- (v0) -- (v4) -- ($(v4)+3*\C*(v1bar)$);
		
		\draw[gray!40] ($(v1)+3*\C*(v2bar)$) -- (v1) -- ($(v1)+3*\C*(v3bar)$) (v1) -- ($(v1)+3*\C*(v4bar)$);
		\draw[gray!40] ($(v2)-3*\C*(v3)$) -- (v2) -- ($(v2)-3*\C*(v4)$);
		\draw[gray!40] ($(v3)-3*\C*(v2)$) -- (v3) -- ($(v3)-3*\C*(v4)$); \draw[gray!40] ($(v4)-3*\C*(v2)$) -- (v4) -- ($(v4)-3*\C*(v3)$);

		\fill [black] (v0) circle (3pt);
		\fill [black] (v1) circle (3pt);
		\fill [black] (v2) circle (3pt);
		\fill [black] (v3) circle (3pt);
		\fill [black] (v4) circle (3pt);
		
		\fill [black] ($\C*(v3)$) circle (1pt);
		\fill [black] ($\C*(v4)$) circle (1pt);
		\fill [black] ($(v1)+\C*(v1bar)$) circle (1pt);
		\fill [black] ($(v1)+\C*(v3bar)$) circle (1pt);
		\fill [black] ($(v2)-\C*(v2)$) circle (1pt);
		\fill [black] ($(v2)-\C*(v3)$) circle (1pt);
		\fill [black] ($(v3)-\C*(v2)$) circle (1pt);
		\fill [black] ($(v3)-\C*(v4)$) circle (1pt);
		\fill [black] ($(v4)-\C*(v1)$) circle (1pt);
		\fill [black] ($(v4)-\C*(v2)$) circle (1pt);
	\end{tikzpicture}
	\caption{In ordinary, `type \textsc{I}$_\text{h}$', ice the oxygens 	constitute a (nearly) perfect hexagonal crystal, where the four nearest neighbours of each oxygen form a tetrahedron centred at that oxygen. We have indicated the hydrogen bonds in grey. The protons near each oxygen satisfy the ice rule.}
	\label{s3:fg:ice}
\end{figure}	
	
In addition \emph{`spin'-reversal} or \emph{reflection} symmetry is often imposed: $\w{\alpha}{\beta}{\gamma}{\delta} = \w{\vphantom{\alpha}\smash{\bar\alpha}}{\vphantom{\beta}\smash{\bar\beta}}{\vphantom{\gamma}\smash{\bar\gamma}}{\vphantom{\delta}\smash{\bar\delta}}$, where the bar denotes negation. This can be interpreted as the absence of an external field, so that there is no preferred direction for the `spins'. This symmetry further cuts the number of independent vertex weights down to three, which are denoted by $a,b,c$ as shown in Figure~\ref{s3:fg:six vertices}. Thinking of these as (local) Boltzmann weights with energies $E_a$, $E_b$ and~$E_c$, we must have $a,b,c\geq 0$ for physical applications. The \emph{ice model} corresponds to the special case $a=b=c$ where each vertex is equally likely.

\begin{figure}[h]
	\centering
	\begin{minipage}{3cm}\begin{tikzpicture}[scale=0.8,font=\small]
		\draw[dotted] (1,0) -- (1,2) (0,1) -- (2,1);
		\draw (1,-.75) node{$\w{-}{-}{-}{-} = a$};
	\end{tikzpicture}\end{minipage}
	\begin{minipage}{3cm}\begin{tikzpicture}[scale=0.8,font=\small]
		\draw[dotted] (0,1) -- (2,1);
		\draw[very thick] (1,0) -- (1,2);
		\draw (1,-.75) node{$\w{+}{-}{+}{-} = b$};
	\end{tikzpicture}\end{minipage}
	\begin{minipage}{3cm}\begin{tikzpicture}[scale=0.8,font=\small]
		\draw[dotted] (0,1) -| (1,2);
		\draw[very thick] (1,0) |- (2,1);
		\draw (1,-.75) node{$\w{+}{-}{-}{+} = c$};
	\end{tikzpicture}\end{minipage}
	\\[\baselineskip]
	\begin{minipage}{3cm}\begin{tikzpicture}[scale=0.8,font=\small]
		\draw[very thick] (1,0) -- (1,2) (0,1) -- (2,1);
		\draw (1,-.75) node{$\w{+}{+}{+}{+} = a$};
	\end{tikzpicture}\end{minipage}	
	\begin{minipage}{3cm}\begin{tikzpicture}[scale=0.8,font=\small]
		\draw[dotted] (1,0) -- (1,2);
		\draw[very thick] (0,1) -- (2,1);
		\draw (1,-.75) node{$\w{-}{+}{-}{+} = b$};
	\end{tikzpicture}\end{minipage}
	\begin{minipage}{3cm}\begin{tikzpicture}[scale=0.8,font=\small]
		\draw[dotted] (1,0) |- (2,1);
		\draw[very thick] (0,1) -| (1,2);
		\draw (1,-.75) node{$\w{-}{+}{+}{-} = c$};
	\end{tikzpicture}\end{minipage}
	\caption{The `allowed' vertex configurations, with nonzero weights~$\w{\alpha}{\beta}{\gamma}{\delta}$, for the six-vertex model. The dotted and thick lines denote `spin' $-1$ and $+1$ on those edges, respectively.}
	\label{s3:fg:six vertices}
\end{figure}
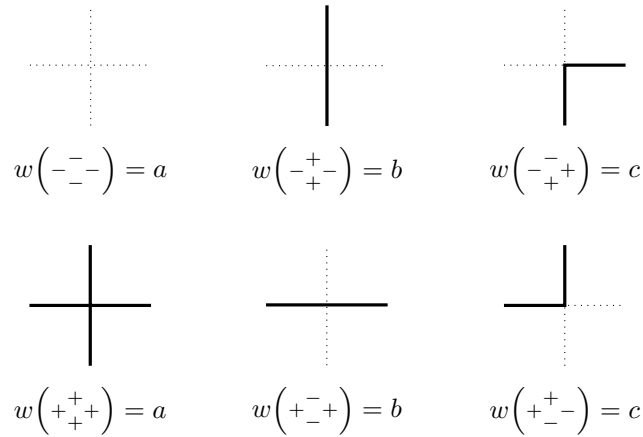

\begin{exercise}
Argue that, because of the periodic boundaries, the two vertices shown on the right in Figure~\ref{s3:fg:six vertices} must occur in equal amounts in any configuration contribution to the partition function. Conclude that one may take them to have equal vertex weights even without imposing `spin'-reversal symmetry.
\end{exercise}

Another example of a vertex model is the \emph{eight-vertex model}. This is a generalization of the six-vertex model where each edge still has two possible `spin' configurations but the ice rule no longer holds, thus allowing for two more vertices with vertex weight~$d$. These are the two vertices in the middle of the configuration from Figure~\ref{s3:fg:example}. We will briefly come back to the eight-vertex model at the end of Section~\ref{s3:results}.

\paragraph{Graphical notation.} To set up the formalism in Sections \ref{s3:method}, \ref{s4:conserved quantities} and~\ref{s4:YBA} we use a graphical notation. It is based on the following four rules:
\begin{enumerate}[label=\roman*)]
\item The basic building blocks are the vertex weights~$\w{\alpha}{\beta}{\gamma}{\delta}$, drawn as in Figure~\ref{s3:fg:vertex model}.
\item Fixed `spins' are depicted using dotted ($\varepsilon=-1$) and thick ($\varepsilon=+1$) lines like in Figure~\ref{s3:fg:six vertices}.
\item There is a summation convention for internal lines: whenever two vertices are connected by an ordinary (i.e.\ not dotted or thick) line there is an implicit sum over the two possible values of the `spins' on the connecting edge. Thus
\begin{equation}\label{s3:eq:convention}
	\tikz[baseline={([yshift=-.5*11pt*0.8]current bounding box.center)},scale=0.8,font=\small]{ 
		\draw (0,1) node[left]{$\beta\vphantom{'}$} -- (3,1) node[right]{$\delta'$};
		\draw (1,0) node[below]{$\alpha\vphantom{'}$} -- (1,2) node[above]{$\gamma$};
		\draw (2,0) node[below]{$\alpha'$} -- (2,2) node[above]{$\gamma'$};}
	\ \ \coloneqq \ \
	\tikz[baseline={([yshift=-.5*11pt*0.8]current bounding box.center)},scale=0.8,font=\small]{
		\draw (0,1) node[left]{$\beta\vphantom{'}$} -- (1,1);
		\draw[dotted] (1,1) -- (2,1);
		\draw (2,1) -- (3,1) node[right]{$\delta'$};
		\draw (1,0) node[below]{$\alpha\vphantom{'}$} -- (1,2) node[above]{$\gamma$};
		\draw (2,0) node[below]{$\alpha'$} -- (2,2) node[above]{$\gamma'$};}
	\ + \ 
	\tikz[baseline={([yshift=-.5*11pt*0.8]current bounding box.center)},scale=0.8,font=\small]{
		\draw (0,1) node[left]{$\beta\vphantom{'}$} -- (1,1);
		\draw[very thick] (1,1) -- (2,1);
		\draw (2,1) -- (3,1) node[right]{$\delta'$};
		\draw (1,0) node[below]{$\alpha\vphantom{'}$} -- (1,2) node[above]{$\gamma$};
		\draw (2,0) node[below]{$\alpha'$} -- (2,2) node[above]{$\gamma'$};}
\end{equation}
represents $\sum_{\varepsilon\in\{\pm1\}} \w{\alpha}{\beta}{\gamma}{\varepsilon} \, \w{\vphantom{\alpha}\smash{\alpha'}}{\varepsilon}{\vphantom{\gamma}\smash{\gamma'}}{\vphantom{\delta}\smash{\delta'}}$.
\item In view of the periodic boundary conditions we also need a way to indicate that opposite edges of a row or column in the lattice are connected. We draw little hooks to depict this periodicity. For example, the partition function specified by~\eqref{s3:eq:partition function} and \eqref{s3:eq:vertex weights} becomes
\begin{equation}\label{s3:eq:6v partition function}
	Z(C) \ \ = \ \ \
	\tikz[baseline={([yshift=-.5*11pt*0.8]current bounding box.center)},scale=0.8,font=\small]{
		\foreach \x in {1,2,4} {
			\draw (0,\x+.2) arc(90:270:.2 and .1) -- (2.6,\x) (3.4,\x) -- (5,\x) arc(-90:90:.2 and .1);
			\foreach \y in {-1,...,1} \draw (3+.2*\y,\x) node{$\cdot\mathstrut$};
			\draw (\x+.2,0) arc(0:-180:.1 and .2) -- (\x,2.6) (\x,3.4) -- (\x,5) arc(180:0:.1 and .2);
			\foreach \y in {-1,...,1} \draw (\x,3+.2*\y) node{$\cdot\mathstrut$};
		}
	}
\end{equation}
This represents a rather complicated expression involving $2KL$ sums as in \eqref{s3:eq:convention}, with the summands being products of $KL$ vertex weights $a$, $b$, $c$ from Figure~\ref{s3:fg:six vertices}.
\end{enumerate}

A nice exercise to get some feeling for this notation and the six-vertex model is the following~\cite[\textsection2.2]{JM99}. Suppose that $c\gg a,b$. The ground state, with maximal $w(C)$, in this regime only involves the two vertices on the right in Figure~\ref{s3:fg:six vertices}. There are two such states: one is shown in Figure~\ref{s3:fg:ground state} and the other is obtained from this by a translation by one unit in the horizontal or vertical direction. To compute the leading correction to the partition function we may therefore restrict our attention to one ground state and calculate $Z/2$ instead.
\begin{exercise}
Use the graphical notation to verify that to ninth order in~$a,b$ the partition function is given by
\begin{equation}
	\tfrac{1}{2} Z = 1 + V\,a^2b^2 + V \, a^2 b^2(a^2+b^2) + \tfrac{1}{2} V(V+1)\,a^4b^4+ V\, a^2b^2(a^4+b^4)+\cdots \ ,
\end{equation}
where $V\coloneqq KL$ and we have set $c=1$ for convenience.
\end{exercise}

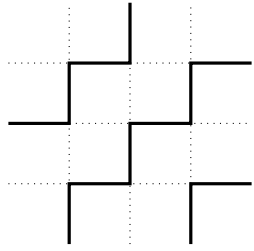
\begin{figure}[h]
	\centering
	\begin{tikzpicture}[scale=0.8,font=\small]
		\draw[dotted] (0,3) -| (1,4);
		\draw[very thick] (0,2) -| (1,3) -| (2,4);
		\draw[dotted] (0,1) -| (1,2) -| (2,3) -| (3,4);
		\draw[very thick] (1,0) |- (2,1) |- (3,2) |- (4,3);
		\draw[dotted] (2,0) |- (3,1) |- (4,2);
		\draw[very thick] (3,0) |- (4,1);
	\end{tikzpicture}
	\caption{A portion of one of the two ground states of the `$F$-model' in the `low-temperature' regime, arising as the special case of the six-vertex model for which $c\gg a,b$.}
	\label{s3:fg:ground state}
\end{figure}

\subsection{The transfer-matrix method and \textsc{cba}}\label{s3:method}

The six-vertex model was solved for three important special cases in 1967 by Lieb~\cite{Lie67a,*Lie67b,*Lie67c} and then in general by Sutherland~\cite{Sut67}. Their solutions combine the transfer-matrix method, allowing one to rewrite the partition function in a quantum-mechanical (linear-algebraic) language, with the \textsc{cba} from Section~\ref{s2:method}. We use tildes to distinguish the six-vertex model's set-up, developed in this subsection, from that of the \textsc{xxz} spin chain.

\paragraph{Transfer-matrix method.} The transfer-matrix method enables one to treat classical statistical systems as if they are quantum mechanical by rewriting the partition function as the trace of some operator to get something like in \eqref{s3:eq:partition function QM}. The basic idea of the method is to divide \eqref{s3:eq:6v partition function} into pieces corresponding to the rows of the lattice.

To define the Hilbert space over which the trace is taken we start locally, like we did for spin chains. Consider a \emph{vertical} edge in the $l$th column of the lattice. To this edge we assign a two-dimensional vector space $\svV_l$ with basis vectors $\ket{\alpha_l}$ labelled by the `spin' $\alpha_l\in\{\pm 1\}$ on that edge:
\begin{equation}\label{s3:eq:Hilb local}
	\svV_l = \mathbb{C} \ket{-}_l \oplus \mathbb{C} \ket{+}_l \ = \ \mathbb{C} \ \, \tikz[baseline={([yshift=-.5*11pt*0.8]current bounding box.center)},scale=0.8,font=\small]{\draw[dotted] (0,0) -- (0,1);} \ \oplus \mathbb{C} \ \, \tikz[baseline={([yshift=-.5*11pt*0.8]current bounding box.center)},scale=0.8,font=\small]{\draw[very thick] (0,0) -- (0,1);} \ \ .
\end{equation}
The Hilbert space associated to a row of vertical edges is constructed as a tensor product of these local vector spaces:
\begin{equation}\label{s3:eq:Hilb global}
	\svcH \coloneqq \bigotimes_{l\in\mathbb{Z}_L} \svV_l = \bigoplus_{\Vector{\alpha}\in\{\pm 1\}^L} \mathbb{C} \, \ket{\Vector{\alpha}} = \bigoplus_{\Vector{\alpha}\in\{\pm 1\}^L} \mathbb{C} \
	\Bigg( \ \smash{
		\tikz[baseline={([yshift=-.5*11pt*0.8+8pt]current bounding box.center)}, 
		scale=0.8,font=\scriptsize,triple/.style={postaction={draw,-,shorten >=.05},double,double distance=4pt,-implies}]{
			\draw (0,0) node[below]{$\alpha_1$} -- (0,1);
			\draw (1,0) node[below]{$\alpha_2$} -- (1,1);
			\draw (3,0) node[below]{$\alpha_L$} -- (3,1);
			\foreach \y in {-1,...,1} \draw (.2*\y+2,.5) node{$\cdot\mathstrut$};
		}
	} \, \Bigg) \ .
\end{equation}

The next step is to define the (row-to-row) \emph{transfer matrix} $t$ on $\svcH$, which counts the contribution to the partition function from the vertices in one row of the lattice: 
\begin{equation}\label{s3:eq:transfer matrix}
	t \ \coloneqq \ \
	\tikz[baseline={([yshift=-.5*11pt*0.8+8pt]current bounding box.center)},
	scale=0.8,font=\scriptsize]{
		\draw (0,1.2) arc(90:270:.2 and .1) -- (2.6,1) (3.4,1) -- (5,1) arc(-90:90:.2 and .1);
		\foreach \x in {1,2,4} \draw (\x,0) -- +(0,2);
		\foreach \y in {-1,...,1} \draw (.2*\y+3,1) node{$\cdot\mathstrut$};
		\draw (1,0) node[below]{$1$} (2,0) node[below]{$2$} (4,0) node[below]{$L$};
	} \ \in \End(\svcH) \ .
\end{equation}
More concretely, $t$ transfers $\ket{\Vector{\alpha}}\in\svcH$ (which we think of as a configuration below some row) to a linear combination of $\ket{\Vector{\gamma}}\in\svcH$ (which we imagine living above that row),
\begin{equation}\label{s3:eq:6v transfer matrix acting}
	t \, \ket{\Vector{\alpha}} \ = \ \sum_{\Vector{\gamma} \in \{\pm 1\}^L} \left(
	\tikz[baseline={([yshift=-.5*11pt*0.8]current bounding box.center)},scale=0.8,font=\small]{
		\draw (0,1.2) arc(90:270:.2 and .1) -- (2.6,1) (3.4,1) -- (5,1) arc(-90:90:.2 and .1);
		\draw (1,0) node[below]{$\alpha_1$} -- (1,2) node[above]{$\gamma_1$};
		\draw (2,0) node[below]{$\alpha_2$} -- (2,2) node[above]{$\gamma_2$};
		\draw (4,0) node[below]{$\alpha_L$} -- (4,2) node[above]{$\gamma_L$};
		\foreach \y in {-1,...,1} \draw (.2*\y+3,1) node{$\cdot\mathstrut$};
	} \right) \ \ket{\Vector{\gamma}} \ .
\end{equation}
The coefficients~$\bra{\Vector{\gamma}}\,t\,\ket{\Vector{\alpha}}$ are polynomials in $a,b,c$ (homogeneous of degree~$L$) that encode how likely each $\ket{\Vector{\gamma}}$ is for a given $\ket{\Vector{\alpha}}$. Taking into account all possible `spin' configurations on the intermediate horizontal edges, each of these polynomials in principle consists of $2^L$ terms; we will soon see that luckily many of these terms are zero.
\begin{exercise}
Use the graphical notation to compute the matrix of $t$ with respect to the basis $\ket{--}$, $\ket{-+}$, $\ket{+-}$, $\ket{++}$ for $L=2$. What do you notice about the form of this matrix? 
\end{exercise}

The use of the transfer matrix comes from the following observation. Powers $t^k$ of the transfer matrix can be depicted as
\begin{equation}
	t^k \ \ = \ \ \
	\tikz[baseline={([yshift=-.5*11pt*0.8+8pt]current bounding box.center)},
	scale=0.8,font=\scriptsize]{
		\foreach \x in {1,2,4} {
			\draw (0,\x+.2) arc(90:270:.2 and .1) -- (2.6,\x) (3.4,\x) -- (5,\x) arc(-90:90:.2 and .1);
			\foreach \y in {-1,...,1} \draw (3+.2*\y,\x) node{$\cdot\mathstrut$};
			\draw (\x,0) -- (\x,2.6) (\x,3.4) -- (\x,5);
			\foreach \y in {-1,...,1} \draw (\x,3+.2*\y) node{$\cdot\mathstrut$};
		}
		\foreach \y in {-1,...,1} \draw (7,3+.2*\y) node{$\cdot\mathstrut$};
		\draw (7,1) node{\small{$1$}} (7,2) node{\small{$2$}} (7,4) node{\small{$k$}};
		\draw (1,0) node[below]{$1$} (2,0) node[below]{$2$} (4,0) node[below]{$L$};
	}
\end{equation}
Since the partition function~\eqref{s3:eq:6v partition function} consists of $K$ such rows, with periodic boundary conditions also imposed in the vertical direction, we have
\begin{equation}\label{s3:eq:partition function as trace}
	Z = \sum_{\Vector{\alpha} \in \{\pm 1\}^L} \bra{\Vector{\alpha}} \, t^K \, \ket{\Vector{\alpha}} = \tr \bigl(t^K\bigr) \ ,
\end{equation}
where the trace is taken over $\svcH$. In this way the computation of the partition function amounts to the diagonalization of the transfer matrix. 

The transfer-matrix method was famously used by Onsager in 1944 to solve the Ising model on a two-dimensional square lattice. Although Onsager allowed for different horizontal and vertical interaction strengths it turned out that the model's behaviour near the critical temperature is universal in the sense that it does not depend on the ratio between the horizontal and vertical coupling constants. This led to the idea of \emph{universality} in statistical physics, and in the following decades more models were found exhibiting the same critical behaviour \cite[\textsection1.3]{Bax07}. Only in 1972, with Baxter's solution of the eight-vertex model, it became clear that there are several different universality classes. (It is generally believed that at criticality every universality class contains an integrable model, which may allow for the exact calculation of the order parameters for that universality class.)

\paragraph{Towards the \textsc{cba}.} Being familiar with the work on the \textsc{xxz} spin chain, Lieb and Sutherland realized that the transfer matrix of the six-vertex model can be diagonalized using the \textsc{cba}. To understand why this is so let us compare the settings of the two models. The first observation is that the two Hilbert spaces, $\cH$ from~\eqref{s2:eq:Hilb global} for the \textsc{xxz} spin chain and $\svcH$ from~\eqref{s3:eq:Hilb global} for the six-vertex model, clearly have the same form. Call the edges of $\mathbb{Z}_L \times \mathbb{Z}_K$ \emph{vacant} when they have `spin'~$-1$ (dotted line) and \emph{occupied} for `spin'~$+1$ (thick line). The two Hilbert spaces are isomorphic via the identification of spin down and up with vacant and occupied vertical edges, respectively:
\begin{equation}\label{s3:eq:isom}
	\cH \ni \ket{\Vector{l}} = \ket{\underset{1}{\uparrow} \cdots \uparrow \underset{l_1}{\downarrow} \uparrow \cdots } \quad \longleftrightarrow \quad \ket{\Vector{\alpha}} = \ket{\underset{1}{-} \cdots - \underset{l_1}{+} - \cdots} \in \svcH \ .
\end{equation}
(Note the difference between the labellings on the two sides: $\Vector{l}$ has just $M$ components $1\leq l_1 < \cdots < l_M \leq L$, while the corresponding $\Vector{\alpha}$ always has $L$ components, $M$ of which are a plus.)
\begin{exercise}\label{s3:ex:M=1}
Compute the matrix elements $\bra{k} t \ket{l}$ from \eqref{s3:eq:6v transfer matrix acting} for $\ket{k},\ket{l}\in\cH_1$, distinguishing the cases $k<l$, $k=l$, $k>l$.
\end{exercise}

The second thing to observe is that the operators that we want to diagonalize, $H_\textsc{xxz}$ from~\eqref{s2:eq:Ham xxz} and $t$ from~\eqref{s3:eq:transfer matrix}, are different. A quick way to see this is by comparing the parameters: the \textsc{xxz} Hamiltonian depends on two parameters~$(J,\Delta)$ while the transfer matrix depends on the values of the three vertex weights~$(a,b,c)$. Despite this difference it may of course be possible to diagonalize the two operators using the same Bethe-ansatz technique.
\begin{exercise}\label{s3:ex:t vs H}
For another difference between $H_\textsc{xxz}$ and $t$ compare the ways in which excitations and occupations in \eqref{s3:eq:isom} are moved by (a single application of) these operators, and compare the number of terms in $H_\textsc{xxz}\ket{\Vector{l}}$ and $t\ket{\Vector{l}}$.
\end{exercise}

Let us take a closer look at the properties of $H_\textsc{xxz}$ and~$t$. In Section~\ref{s2:spin chains} we exploited the fact that the \textsc{xxz} Hamiltonian is
\begin{enumerate}[label=\roman*),noitemsep]
	\item nearest neighbour;
	\item translationally invariant; and
	\item partially isotropic.
\end{enumerate}
Since the six-vertex model satisfies properties (i) and~(ii), so does the transfer matrix. The isomorphism~\eqref{s3:eq:isom} suggests that the six-vertex analogue of the $M$-particle sector~$\cH_M \subseteq \cH$ is the subspace $\svcH_M \subseteq \svcH$ with basis vectors $\ket{\Vector{\alpha}}$ containing precisely $M$ occupancies. By~(iii) the $M$-particle sectors are preserved by the \textsc{xxz} Hamiltonian. Let us now show that the `$M$-occupancy sectors'~$\svcH_M$ are similarly left invariant by the transfer matrix: thinking of the vertical direction of the lattice as (periodic and discrete) time, the occupancy number~$M$ is conserved for the six-vertex model as a consequence of the ice rule. This can be clearly seen by drawing the vertices from Figure~\ref{s3:fg:six vertices} as in Figure~\ref{s3:fg:line conservation}. Indeed, if there are $M$ occupancies below some row then, due to the horizontal periodicity, there must be $M$ occupancies above that row as well. This \emph{line conservation} is the six-vertex analogue of the $U(1)_z$-rotational symmetry of the \textsc{xxz} spin chain discussed in Section~\ref{s2:spin chains}. Therefore the decomposition $\svcH =\oplus_M \svcH_M$ is preserved by the transfer matrix too, so $t$ is block diagonal, and we can diagonalize it in these $M$-occupancy sectors separately.

\begin{figure}[h]
	\centering
	\begin{minipage}{3cm}\begin{tikzpicture}[scale=0.8,font=\small]
		\draw[dotted] (0,1) -- (2,1) (1,0) -- (1,2);
	\end{tikzpicture}\end{minipage}
	\begin{minipage}{3cm}\begin{tikzpicture}[scale=0.8,font=\small]
		\draw[dotted] (0,1) -- (2,1) (1,0) -- (1,2);
		\draw[very thick] (1,0) -- (1,2);
	\end{tikzpicture}\end{minipage}
	\begin{minipage}{3cm}\begin{tikzpicture}[scale=0.8,font=\small]
		\draw[dotted] (0,1) -- (2,1) (1,0) -- (1,2);
		\draw[very thick,rounded corners=8pt] (1,0) |- (2,1);
	\end{tikzpicture}\end{minipage}
	\\[\baselineskip]
	\begin{minipage}{3cm}\begin{tikzpicture}[scale=0.8,font=\small]
		\draw[dotted] (0,1) -- (2,1) (1,0) -- (1,2);
		\draw[very thick,rounded corners=8pt] (0,1) -| (1,2) (1,0) |- (2,1);
	\end{tikzpicture}\end{minipage}	
	\begin{minipage}{3cm}\begin{tikzpicture}[scale=0.8,font=\small]
		\draw[dotted] (0,1) -- (2,1) (1,0) -- (1,2);
		\draw[very thick] (0,1) -- (2,1);
	\end{tikzpicture}\end{minipage}
	\begin{minipage}{3cm}\begin{tikzpicture}[scale=0.8,font=\small]
		\draw[dotted] (0,1) -- (2,1) (1,0) -- (1,2);
		\draw[very thick,rounded corners=8pt] (0,1) -| (1,2);
	\end{tikzpicture}\end{minipage}
	\caption{The six vertex configurations with nonzero weights from Figure~\ref{s3:fg:six vertices} redrawn such that the ice rule for a vertex can be interpreted as line conservation at that vertex.}
	\label{s3:fg:line conservation}
\end{figure}
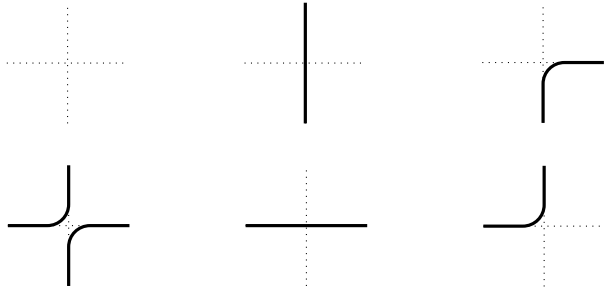

\begin{exercise}\label{s3:ex:at most two terms}
Convince yourself that, by line conservation, for given $\Vector{\alpha}$ and $\Vector{\gamma}$ \emph{at most two} terms in \eqref{s3:eq:6v transfer matrix acting} have nonzero statistical weight, and that there are \emph{two} precisely for the diagonal matrix entries ($\Vector{\gamma}=\Vector{\alpha}$). Compute these diagonal entries for the $M$-occupancy sector. (See Figure~\ref{s3:fg:transfer example} for some examples with $L=3$.)
\end{exercise}

\begin{figure}[h]
	\centering
	\begin{tikzpicture}[scale=0.8,font=\small]
		\draw (0,1.2) arc(90:270:.2 and .1) -- (4,1) arc(-90:90:.2 and .1);
		\draw[very thick] (1,0) -- (1,1); \draw[dotted] (1,1) -- (1,2);
		\draw[very thick] (2,0) -- (2,1.96); \draw (2,2) -- (2,2);
		\draw[dotted] (3,0) -- (3,2); \draw[very thick] (3,1) -- (3,2);
		\draw (5,1) node{=};
		\begin{scope}[shift={(6,0)}]
			\draw[dotted] (0,1.2) arc(90:270:.2 and .1) -- (4,1) arc(-90:90:.2 and .1);
			\foreach \x in {1,...,3} \draw[dotted] (\x,0) -- (\x,2);
			\draw[very thick,rounded corners=8pt] (1,0) -- (1,1) -- (2,1) -- (2,2);
			\draw[very thick,rounded corners=8pt] (2,0) -- (2,1) -- (3,1) -- (3,2);
		\end{scope}
	\end{tikzpicture}\\[\baselineskip]
	\begin{tikzpicture}[scale=0.8,font=\small]
		\draw (0,1.2) arc(90:270:.2 and .1) -- (4,1) arc(-90:90:.2 and .1);
		\draw[very thick] (1,0) -- (1,2); 
		\draw[very thick] (2,0) -- (2,2); 
		\draw[dotted] (3,0) -- (3,2);
		\draw (5,1) node{=};
		\begin{scope}[shift={(6,0)}]
			\draw[dotted] (0,1.2) arc(90:270:.2 and .1) -- (4,1) arc(-90:90:.2 and .1);
			\draw[very thick] (1,0) -- (1,2);
			\draw[very thick] (2,0) -- (2,2);
			\draw[dotted] (3,0) -- (3,2);
		\end{scope}
		\draw (11,1) node{+};
		\begin{scope}[shift={(12,0)}]
			\draw[very thick,rounded corners=8pt] (0,1.2) arc(90:270:.2 and .1); \draw[very thick,rounded corners=8pt] (4,1) arc(-90:90:.2 and .1); \foreach \x in {1,...,3} \draw[dotted] (\x,0) -- (\x,2);
			\draw[very thick,rounded corners=8pt] (0,1) -| (1,2);
			\draw[very thick,rounded corners=8pt] (1,0) |- (2,1) -- (2,2);
			\draw[very thick,rounded corners=8pt] (2,0) |- (4,1);
		\end{scope}
	\end{tikzpicture}
	\caption{Two typical examples of the graphical computation of the matrix elements of the transfer matrix for $L=3$ and $M=2$. The top shows that $\bra{-++}\,t\,\ket{++-}=ac^2$ and the bottom says that $\bra{++-}\,t\,\ket{++-} = a b^2 + a^2 b$.}
	\label{s3:fg:transfer example}
\end{figure}
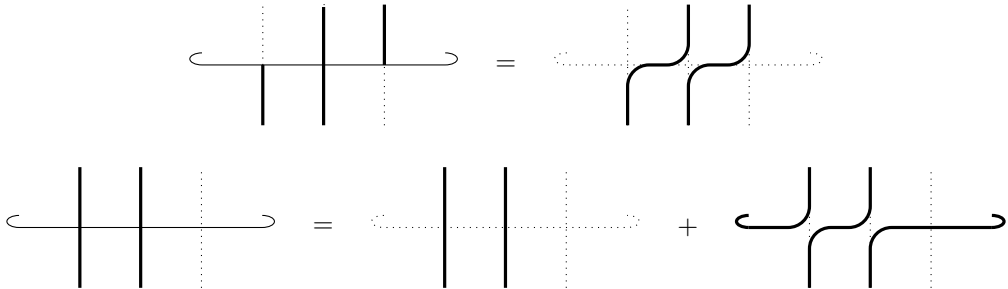

\paragraph{The \textsc{cba}.} Having identified the relevant properties of the transfer matrix we are all set to apply the \textsc{cba} for the diagonalization of the transfer matrix in $\svcH_M$. The basic idea is similar to what we described for \textsc{xxz} spin chain in Section~\ref{s2:method} so we keep it brief; the details of the \textsc{cba} for the six-vertex model can be found in \cite[\textsection8.3--8.4]{Bax07}, see also the end of Appendix~\ref{sM}.$^\#$\footnote{In \cite[\textsection8.2]{Bax07} the transfer matrix is thought of as acting from \emph{above} to \emph{below} a given row, as opposed to our convention in \eqref{s3:eq:6v transfer matrix acting}. The actual \textsc{cba} in \cite[\textsection8.3]{Bax07} agrees with \eqref{s3:eq:CBA M}. One can check that, as a consequence, Baxter's equations and results match those in Section~\ref{s3:results} upon replacing $z_m\leftrightarrow z_m^{-1}$.}

We want to find $\ket{\svPsi_M}\in\cH_M \subseteq \svcH$ solving the eigenvalue problem $t\,\ket{\svPsi_M} = \Lambda_M \, \ket{\svPsi_M}$. The identification $\svcH\simeq \cH$ from \eqref{s3:eq:isom} allows us to use the more convenient `occupancy basis' $\ket{\Vector{l}}$ instead of~$\ket{\Vector{\alpha}}$. Expand $\ket{\svPsi_M}$ in terms of the $M$-occupancy basis $\ket{\Vector{l}}\in \cH_M \simeq \svcH_M$ like in \eqref{s2:eq:M vector via coord basis}. The \textsc{cba} for the coefficients involves parameters~$\svp_m \in\mathbb{C}$ or equivalently $z_m \coloneqq \E^{\choice{-}{} \I \svp_m} \in\mathbb{C}$,
\begin{equation}\label{s3:eq:CBA M}
	\svPsi_{\Vector{z}}(\Vector{l}) = \sum_{\pi\in S_M} \svA_\pi(\Vector{z}) \Vector{z}_\pi {}^{\Vector{l}} \ , \qquad \Vector{z}_\pi {}^{\Vector{l}} \coloneqq \prod_{m=1}^M (z_{\pi(m)})^{l_m} \ , \qquad\qquad l_1 < \cdots < l_M \ .
\end{equation}
This produces eigenvectors for the transfer matrix provided we can solve the equations
\begin{equation}\label{s3:eq:6v eqns}
	\bra{\Vector{l}} \, t \, \ket{\svPsi_M; \Vector{z}} = \Lambda_M(\Vector{z}) \, \svPsi_{\Vector{z}}(\Vector{l}) \ , \qquad\qquad 1\leq l_1 < \cdots < l_M \leq L \ ,
\end{equation}
for the eigenvalues~$\Lambda_M(\Vector{z})$, the coefficients~$\svA_\pi(\Vector{z})$ and the values of the parameters~$\Vector{z}$. The strategy is roughly as before:
\begin{enumerate}
	\item Focus on the \emph{wanted} terms, i.e.\ terms proportional to $\Vector{z}_\pi {}^{\Vector{l}}$ as in the \textsc{cba}, to find $\Lambda_M(\Vector{z})$.
	\item Find $\svA_\pi(\Vector{z})$ so as to cancel certain unwanted (`internal') terms.
	\item Demand that remaining unwanted (`boundary') terms also cancel to get the \textsc{bae} determining the allowed values of~$\Vector{z}$.
\end{enumerate}
However, in accordance with Exercise~\ref{s3:ex:t vs H}, the left-hand side of \eqref{s3:eq:6v eqns} contains many more terms than its \textsc{xxz}-analogue~\eqref{s2:eq:xxz eqns}. Correspondingly the precise formulation of the strategy is a bit more involved than before too, see the end of Appendix~\ref{sM}. Let us proceed to the outcome.

\subsection{Unexpected results}\label{s3:results}

In Section~\ref{s3:method} we have already found some striking similarities between the \textsc{xxz} spin chain and the six-vertex model. Now we will see that the results of the \textsc{cba} uncover a much deeper relation between the two models.

\paragraph{Results for general $M$.} The results of the \textsc{cba} for the transfer matrix are as follows.

\subparagraph{Step 1.} The eigenvalues of the transfer matrix are
\begin{equation}\label{s3:eq:6v eigenvalues}
	\Lambda_M(\Vector{z}) = a^L \prod_{m=1}^M \frac{b\,(a-b\,z_m^{-1})+c^2\,z_m^{-1}}{a\,(a-b\,z_m^{-1})} + b^L \prod_{m=1}^M \frac{a\,(a-b\,z_m^{-1})-c^2}{b\,(a-b\,z_m^{-1})} \ .
\end{equation}
Not being additive, the result is clearly different from \eqref{s2:eq:xxz energy M}, once more showing that $H_\textsc{xxz}$ and~$t$ really are different operators.

\subparagraph{Step 2.} Again the coefficients $\svA_\pi$ in the \textsc{cba} factor into two-occupancy contributions. Interestingly they are very similar to what we found for the \textsc{xxz} spin chain:
\begin{equation}\label{s3:eq:A solution}
	\frac{\svA_\pi(\Vector{z})}{\svA_e(\Vector{z})} \ = \ \prod_{\substack{1 \leq m<m' \leq M \\ \text{s.t.\ } \pi(m)>\pi(m')}} \svS(z_m,z_{m'}) \ , \qquad 
	\svS(z,z') \coloneqq -\frac{1-2\,\svDelta(a,b,c)\, z' + z\,z'}{1-2\,\svDelta(a,b,c)\, z + z\,z'} \ .
\end{equation}
The only difference with \eqref{s2:eq:A solution} and \eqref{s2:eq:xxz S-matrix} is that the anisotropy parameter~$\Delta$ of the \textsc{xxz} model is replaced by a particular combination of the six-vertex weights,
\begin{equation}\label{s3:eq:Delta(a,b,c)}
	\svDelta(a,b,c) \coloneqq \frac{a^2 + b^2 - c^2}{2\,a\,b} \ .
\end{equation}
This striking similarity will play a crucial role in what follows.

\subparagraph{Step 3.} There are $M$ \textsc{bae} for the parameters $\Vector{z}$:
\begin{equation}\label{s3:eq:6v BAE M}
	z_m^L = (-1)^{M-1} \prod_{\substack{n=1 \\ n\neq m}}^M \frac{1-2\,\svDelta(a,b,c)\, z_m + z_m z_n}{1-2\,\svDelta(a,b,c)\,z_n + z_m z_n} \ , \qquad\qquad \ 1 \leq m \leq M \ .
\end{equation}
Up to the dependence on \eqref{s3:eq:Delta(a,b,c)} instead of $\Delta$ these are identical to the \textsc{xxz} \textsc{bae}~\eqref{s2:eq:xxz BAE M}.

\begin{exercise}
Use the results of Exercises \ref{s3:ex:M=1} and~\ref{s3:ex:at most two terms} to check \eqref{s3:eq:6v eigenvalues} and \eqref{s3:eq:6v BAE M} for $M=0,1$. For $M=1$ recognize geometric series in $z$ to sum many terms on the left-hand side of \eqref{s3:eq:6v eqns}.
\end{exercise}

\paragraph{Commuting transfer matrices.} The solution \eqref{s3:eq:A solution} for~$A_\pi(\Vector{z})$, and therefore~$\ket{\svPsi_M; \Vector{z}}$, is remarkable. Firstly, the coefficients~\eqref{s3:eq:A solution} and the \textsc{bae}~\eqref{s3:eq:6v BAE M} are \emph{precisely} the same as those of the \textsc{xxz} spin chain when the function~\eqref{s3:eq:Delta(a,b,c)} has fixed value $\svDelta(a,b,c)=\Delta$ equal to the \textsc{xxz} anisotropy parameter. Indeed, from \eqref{s3:eq:6v BAE M} we see that the allowed values of the parameters $\svp_m = \choice{}{-}\I \log z_m$ match those of the $p_m$. Secondly, the eigenvectors of the transfer matrix \emph{only} depend on the six-vertex weights through the combination~\eqref{s3:eq:Delta(a,b,c)}. Therefore varying the values of $a,b,c$ while keeping \eqref{s3:eq:Delta(a,b,c)} fixed does not change the (Bethe) eigenvectors of $t(a,b,c)$. 

Under the \emph{assumption} that all $2^L$ eigenvectors of $t$ and $H_\textsc{xxz}$ are of the Bethe form (see Appendix~\ref{sY}), these two facts mean that the \textsc{cba} \emph{simultaneously} diagonalizes the \textsc{xxz} Hamiltonian and all transfer matrices with matching value of~\eqref{s3:eq:Delta(a,b,c)}:
\begin{align}
	[\, t(a,b,c) , \, H_\textsc{xxz}(\Delta) \,] \: & = 0 \qquad\qquad \text{if} \quad   \svDelta(a,b,c)= \Delta \ , \label{s3:eq:commuting t H} \\
	[\, t(a,b,c) , \, t(a',b',c') \,] & = 0 \qquad\qquad \text{if} \quad  \svDelta(a,b,c)= \svDelta(a',b',c') \ . \label{s3:eq:commuting t's}
\end{align}
As we will soon see these two observations hold the key to understanding the integrability of the six-vertex model --- and that of the \textsc{xxz} spin chain.

To understand the consequences of \eqref{s3:eq:commuting t H}--\eqref{s3:eq:commuting t's} let us first look at the degrees of freedom contained in the six-vertex model's parameters~$(a,b,c)$. 
\begin{exercise}\label{s3:ex:rescaling}
Check that simultaneous nonzero rescalings $(a,b,c)\longmapsto (r\,a,r\,b,r\,c)$ do not affect the combination~\eqref{s3:eq:Delta(a,b,c)} and only modify the partition function~\eqref{s3:eq:6v partition function} by an overall factor.
\end{exercise}
Motivated by this let us fix the ratio $a:b:c$ and the value of the function \eqref{s3:eq:Delta(a,b,c)}. This leaves a single remaining degree of freedom, known as the \emph{spectral parameter}, which we denote by~$u$. Observe that, through the vertex weights, the transfer matrix also depends on the spectral parameter, $t(u) = t\bigl(a(u),b(u),c(u)\bigr)$. We can now recast \eqref{s3:eq:commuting t H}--\eqref{s3:eq:commuting t's} in the form
\begin{align}
	[\, t(u) , \, H_\textsc{xxz} \,] & = 0 \qquad\qquad \text{for all $u$} \ , \label{s3:eq:commuting t H v2} \\
	[\, t(u) , \, t(v) \,] \; & = 0 \qquad\qquad \text{for all $u,v$} \ . \label{s3:eq:commuting t's v2}
\end{align}
Therefore there is a one-parameter family of six-vertex models, with $a:b:c$ and $\svDelta$ fixed but varying $u$, whose transfer matrices~$t(u)$ commute with $H_\textsc{xxz}$ and each other.
\begin{exercise}
Check that the parametrization
\begin{equation}\label{s3:eq:parametrization}
	a = \rho \sinh(u+\I\gamma) \ , \qquad b = \rho \sinh u \ , \qquad c = \rho \sinh(\I\gamma) = \I \rho \sin\gamma 
\end{equation}
does the job, with $\svDelta(a,b,c)=\cos\gamma$. Determine the values of the \emph{crossing parameter}~$\gamma$ corresponding to the regimes $\Delta<-1$, $-1\leq \Delta \leq 1$ and $\Delta>1$ of the \textsc{xxz} model. (Correspondingly, shifted or rescaled parameters are also commonly used in the literature.)
\end{exercise}

\paragraph{$Z$-invariant models.} How should the commutator \eqref{s3:eq:commuting t's v2} be interpreted from the vertex-model viewpoint? In terms of our graphical notation it consists of two terms of the form
\begin{equation}\label{s3:eq:tu tv}
	t(u)\, t(v) \ \ = \ \ \
		\tikz[baseline={([yshift=-.5*11pt*0.8+8pt]current bounding box.center)},
		scale=0.8,font=\scriptsize]{
			\foreach \x in {1,2} {
				\draw (0,\x+.2) arc(90:270:.2 and .1) -- (2.6,\x) (3.4,\x) -- (5,\x) arc(-90:90:.2 and .1);
				\foreach \y in {-1,...,1} \draw (3+.2*\y,\x) node{$\cdot\mathstrut$};
			}
			\foreach \x in {1,2,4} {
				\draw (\x,0) -- (\x,3);
			}
			\draw (6,1) node{\small{$u$}} (6,2) node{\small{$v$}};
			\draw (1,0) node[below]{$1$} (2,0) node[below]{$2$} (4,0) node[below]{$L$};
		}
\end{equation}
with a separate spectral parameter associated to each row as indicated. This can be viewed as a portion of a vertex model with different values of the spectral parameter --- hence different vertex weights yielding the same value of \eqref{s3:eq:Delta(a,b,c)} --- for each row of horizontal edges in the lattice. By \eqref{s3:eq:commuting t's v2} the partition function~$Z$~\eqref{s3:eq:partition function} of such vertex models are invariant under the exchange of any two rows in the lattice; accordingly those models are called \emph{$Z$-invariant}. Thus the six-vertex model admits inhomogeneous generalizations that can still be tackled using the \textsc{cba}: the translational invariance in the vertical direction is broken in such a way that the model remains exactly solvable.

\paragraph{Analyticity.} Baxter realized that it is extremely useful to allow for complex vertex weights and let $u \in \mathbb{C}$. Indeed, \eqref{s3:eq:parametrization} then gives an \emph{analytic} parametrization of the six-vertex weights which is even \emph{entire} (in fact, this is how that parametrization can be found, see~\cite[\textsection9.7]{Bax07}). The real power of the transfer-matrix method lies in the fact that all functions $u \longmapsto \bra{\Vector{k}} t(u) \ket{\Vector{l}}$ are entire as well, because they are polynomial in~$a,b,c$. This highly constrains the properties of the \textsc{xxz} and six-vertex model and ultimately renders these models exactly solvable. For example, the \textsc{bae} have a natural interpretation in this context:
\begin{exercise}\label{s3:ex:z via a b}
Since the transfer matrix is entire in $u$, so must be its eigenvalues~$\Lambda_z(\Vector{z})$. However, the right-hand side of \eqref{s3:eq:6v eigenvalues} seems to have a simple pole for each $u$ such that $a(u)/b(u) = z_m^{-1}$ for some $1\leq m\leq M$. Use that $\Res(f,z_*) = g(z_*)/h'(z_*)$ when $f(z)=g(z)/h(z)$ with $h(z_*)=0$ but $h'(z_*)\neq 0$ to check that the residues at these poles satisfy
\begin{equation}
	\Res\big(\Lambda_M, z_m = \tfrac{b}{a}\big) \propto a^L \prod_{\substack{n=1 \\ n \neq m}}^M \left(1- \frac{b^2-c^2}{a\,b} z_n^{-1} \right) - b^L \prod_{\substack{n=1 \\ n \neq m}}^M \left(\frac{a^2-c^2}{a\,b} \, z_m^{-1} - z_m^{-1} \, z_n^{-1} \right) \ ,
\end{equation}
and conclude that the poles of $\Lambda_M$ disappear by virtue of the \textsc{bae}~\eqref{s3:eq:6v BAE M} with $z_m = b/a$.
\end{exercise}

In particular, if one would be able to obtain the eigenvalues \eqref{s3:eq:6v eigenvalues} of $t$ in a different way, one might be able to derive \textsc{bae} also for models for which the Bethe-ansatz techniques described in these lecture notes fail, such as the \textsc{xyz} spin chain and eight-vertex model. This is precisely what Baxter's \emph{$TQ$-method} manages to do by constructing another one-parameter family of commuting operators $Q(u)\in\End(\svcH)$ that satisfy certain `$TQ$-relations' determining $\Lambda_M$. For more about the $TQ$-method we refer to~\cite[\textsection9]{Bax07}; see \cite[\textsection4--5]{TF79} and \cite[4.2]{GRS96} for an account in the algebraic language of Section~\ref{s4}.

\paragraph{Quantum integrability.} To get a better understanding of the importance of the relations \eqref{s3:eq:commuting t H v2}--\eqref{s3:eq:commuting t's v2} let us parametrize the vertex weights as in~\eqref{s3:eq:parametrization}. Since the transfer matrix then is a Laurent polynomial in~$\E^u$, it makes sense to take logarithmic derivatives and define operators $H_k$ via the \emph{trace identities}
\begin{equation}\label{s3:eq:trace identities}
	H_k \coloneqq \left.\frac{\D^k}{\D u^k}\right|_{u=u_*} \log t(u) \ \in \  \End(\svcH) \simeq \End(\cH)
\end{equation}
for some value~$u_*$ of the spectral parameter. (In Section~\ref{s4:conserved quantities} we will see that $u_* = 0$ is a convenient choice for our parametrization.) 

The equations \eqref{s3:eq:commuting t H v2}--\eqref{s3:eq:commuting t's v2} then imply
\begin{align}
	[H_k , H_\textsc{xxz}] & = 0 \qquad\qquad \text{for all $k$} \ , \label{s3:eq:commuting Hi H} \\
	[H_j, H_k] \;\, & = 0 \qquad\qquad \text{for all $j,k$} \ . \label{s3:eq:commuting Hi's}
\end{align}
Now we can see the fruits of our labour more clearly. According to \eqref{s3:eq:commuting Hi H} the operators~$H_k$ are conserved quantities: we have found symmetries of the \textsc{xxz} spin chain! Moreover, by \eqref{s3:eq:commuting Hi's} these symmetry operators commute with each other (they are in involution). The presence of such a tower of commuting conserved charges is a very special property; it `proves' that the model is \emph{quantum integrable} in analogy with the notion of Liouville integrability and explains why magnon-scattering is two-body reducible, see Section~\ref{s5:fact scatt}. Thus, from the spin-chain viewpoint, there is a one-parameter family of six-vertex models whose transfer matrices produce symmetries~$H_k$ of the \textsc{xxz} model through the trace identities. (In more mathematical terms the $t(u)$ generate an abelian subalgebra in~$\End(\cH)$ that commutes with $H_\textsc{xxz}$.) What about the six-vertex model itself?

Consider a six-vertex model with vertex weights $(a_0,b_0,c_0)$ and transfer matrix $t_0 \coloneqq t(a_0,b_0,c_0)$. Setting $\Delta_0 \coloneqq \svDelta(a_0,b_0,c_0)$, by the above argument there exists a one-parameter family of six-vertex models with commuting transfer matrices, like in \eqref{s3:eq:commuting Hi's}, such that $t(u_0) = t_0$ for some~$u_0$. From the original six-vertex model's perspective each of these transfer matrices generates a discrete Euclidean `time' evolution with respect to which the $H_k$ are `conserved'. In particular it follows that 
\begin{equation}
	[H_k, t_0] = 0 \qquad\qquad \text{for all $k$} \ , \label{s3:commuting Hn t0}
\end{equation}
so $t_0$ enjoys the \emph{same} symmetries~\eqref{s3:eq:trace identities} as $H_\textsc{xxz}(\Delta_0)$!
\begin{exercise}
Go back to Section~\ref{s2:spin chains} to find a few operators that one may expect (or hope) to find amongst the $H_k\in \End(\cH)$. 
\end{exercise}

\pagebreak[4]

\paragraph{Summary.} The preceding discussion can be schematically summarized as follows:
\begin{equation}
	\tikz[baseline={([yshift=-.5*11pt*0.8]current bounding box.center)},scale=0.8,descr/.style={fill=white,font=\small}]{
		\matrix(m)[matrix of math nodes, text height=1.5ex, text depth=0.25ex, row sep=4em, column sep=4em,ampersand replacement=\&]{
			\& t(u) \& \\ H_\textsc{xxz} \& \text{symmetries} \& t_0 \\};
		\path[dotted,->](m-1-2) edge node[descr]{?} (m-2-1);
		\path[->](m-1-2) edge node[descr]{$u=u_*$} (m-2-2)
			(m-1-2) edge node[descr]{$u=u_0$} (m-2-3)
			(m-2-2) edge (m-2-1)
			(m-2-2) edge (m-2-3)
			(m-2-1) edge [bend left=40] node[descr]{$\svDelta=\Delta$} (m-1-2)
			(m-2-3) edge [bend right=40] node[descr]{$\svDelta=\Delta_0$} (m-1-2);
	}
\end{equation} 
Although we have not uncovered the exact relation between the \textsc{xxz} and six-vertex models yet, Table~\ref{s3:tb:table} contains a dictionary with our findings so far. Let us stress once more that the correspondence between the two models is \emph{not} a bijection: the \textsc{xxz} spin chain with anisotropy~$\Delta$ corresponds to a whole \emph{family} of six-vertex models parametrized by~$u$. The precise connection between the two sides will become clear in Section~\ref{s4:conserved quantities} when we compute the first few $H_k$ contained in the transfer matrix. 

\begin{table}[h]
	\centering
	\begin{tabular}{c|c}
		\textbf{\textsc{xxz} spin chain} \qquad & \textbf{(Family of) six-vertex models} \\ \hline \strut
		lattice $\mathbb{Z}_L$ & row of vertical edges in $\mathbb{Z}_L\times \mathbb{Z}_K$ \\ 
		basis vector $\ket{\Vector{l}} \in \cH$ & configuration $\ket{\Vector{\alpha}} \in \svcH$ on a row \\
		pseudovacuum~$\ket{\Omega}$ & configuration $\ket{--\,\cdots\, -}$ \\ 
		excited spin at site~$l$ & occupancy at edge~$l$ \\ \hline \strut
		translational symmetry & horizontal translational symmetry \\
		partial isotropy & ice rule/line conservation \\
		anisotropy $\Delta$ & $\svDelta(a,b,c) = (a^2 + b^2 - c^2)/2ab$
		\\ \hline \strut
		quasimomentum $p_m$ & parameter $\svp_m = \choice{}{-}\I \log z_m$ 
	\end{tabular}
	\caption{Comparison between the ingredients of the \textsc{xxz} spin chain and those of the corresponding one-parameter family of six-vertex models.}
	\label{s3:tb:table}
\end{table}

\section{The quantum inverse-scattering method}\label{s4}

In the previous sections we studied the \textsc{cba} for the \textsc{xxz} spin chain and the six-vertex model. The details were deferred to Appendix~\ref{sM}: solving the equations \eqref{s2:eq:xxz eqns} for the $M$-particle sector of $H_\textsc{xxz}$ is rather involved as all cases with neighbouring excitations must be taken into account, and the equations \eqref{s3:eq:6v eqns} for the $M$-occupancy sector of the transfer matrix are still harder to obtain. It would be nice if there is an easier way to diagonalize these operators and derive the \textsc{bae}.

Next, through the transfer-matrix method in Sections \ref{s3:method} and~\ref{s3:results} we found a correspondence between the \textsc{xxz} spin chain with anisotropy~$\Delta$ and a one-parameter family of six-vertex models parametrized by the spectral parameter~$u \in \mathbb{C}$. In particular, $H_\textsc{xxz}$ and $t$ are simultaneously diagonalized, and the family of commuting transfer matrices $t(u)$ generates a tower of symmetries~$H_k$ for both sides via the trace identities. However, the precise relation between $t(u)$ and $H_\textsc{xxz}$ is not clear yet. In addition, computing the $H_k$ directly from the transfer matrix is rather cumbersome; in fact the special value~$u_*$ of the spectral parameter still has to be determined.

In this section the \emph{quantum inverse-scattering method}~(\textsc{qism}) is introduced. This algebraic formalism can be used to rederive the results from Sections \ref{s2:results} and~\ref{s3:results} while addressing the above issues, and more:
\begin{itemize}[noitemsep]
\item it provides a convenient way to find an appropriate value $u_*$ and compute the~$H_k$ using the trace identities;
\item it has the all-important commutativity of the $t(u)$ built in;
\item the Fock space of (Bethe) states is constructed via creation and annihilation operators, and the eigenvalues and \textsc{bae} are derived with a single computation for general $M$;
\item it allows one to define several new quantum-integrable models.
\end{itemize}
Since the \textsc{xxz} and six-vertex models are treated simultaneously in the \textsc{qism} we no longer need to distinguish between $\svcH\simeq\cH$, et cetera, and can safely drop the tildes used in Sections \ref{s3:method} and~\ref{s3:results} from now on. We keep $\hbar=J=1$.

\subsection{Conserved quantities from Lax operators}\label{s4:conserved quantities}

We keep using the graphical notation introduced in Section~\ref{s3:six-vertex model}, with one additional rule:
\begin{enumerate}
\item[v)] To keep track of the ordering of the various operators in products we add little arrows to the end of lines. (Thus these arrows are \emph{not} related to the `spins'.) 
\end{enumerate}
We will see an example of this rule soon, e.g.\ in \eqref{s4:eq:transfer matrix}.

\paragraph{Lax operators.} From the six-vertex point of view the choice to associate a vector space to the rows of $\mathbb{Z}_L\times \mathbb{Z}_K$, but not its columns, is somewhat unnatural. To treat the horizontal and vertical edges on a more equal footing we introduce the vector space
\begin{equation}
	V_a \coloneqq \mathbb{C} \, \ket{-}_a \oplus \mathbb{C} \, \ket{+}_a \ = \ \mathbb{C} \ \, 
		\tikz[baseline={([yshift=-.5*11pt*0.8+1pt]current bounding box.center)},scale=0.8,font=\small]{
			\draw[dotted] (0,0) -- (1,0);
		} \,
	\oplus \mathbb{C} \ \,
		\tikz[baseline={([yshift=-.5*11pt*0.8+1pt]current bounding box.center)},scale=0.8,font=\small]{
			\draw[very thick] (0,0) -- (1,0);
		}
	\ .
\end{equation}
spanned by the two possible `spins' on a horizontal edge in the lattice. (The subscript in $V_a$ is not related to the vertex weight~$a(u)$.) Motivated by the spin-chain viewpoint, the `vertical' $V_l$ and $\cH$ from \eqref{s3:eq:Hilb local}--\eqref{s3:eq:Hilb global} are then often called \emph{physical} or \emph{quantum} spaces, while $V_a$ is an \emph{auxiliary} space. It is quite convenient to think of the spectral parameter~$u$ of the transfer matrix as being associated to~$V_a$.

The auxiliary space allows us to introduce `local' (vertex) operators acting at a single vertex of the lattice: the \emph{Lax operator},$^\#$\footnote{This operator is sometimes called the `$R$-matrix', but we follow Faddeev~\cite{Fad95a} and reserve the latter terminology for a closely related operator acting on $V_a\otimes V_b$, see Section~\ref{s4:YBA}.} defined as
\begin{equation}\label{s4:eq:Lax operator}
	L_{al}(u) \ \coloneqq \
		\tikz[baseline={([yshift=-.5*11pt*0.8]current bounding box.center)}, scale=0.8,font=\scriptsize]{
			\draw[->] (0,1) node[left]{$a$} -- (2,1) node[right]{$a$};
			\draw[->] (1,0) node[below]{$l$} -- (1,2) node[above]{$l$};
		} \ \in \ \End(V_a\otimes V_l) \ .
\end{equation}
The subscripts in $L_{al}(u)$ remind us on which vector spaces this operator acts nontrivially, in accordance with the tensor-leg notation from the start of Section~\ref{s2:spin chains}. The labels on the top and right of \eqref{s4:eq:Lax operator} will be omitted in the graphical notation from now on.

Explicitly, writing $\ket{\beta,\alpha}\coloneqq\ket{\beta}\otimes\ket{\alpha}$ for the (pure) vectors in $V_a\otimes V_l$, \eqref{s4:eq:Lax operator} means that
\begin{equation}\label{s4:eq:Lax operator acting}
	L_{al}(u) \, \ket{\beta,\alpha} \ = \sum_{\gamma,\delta \in \{\pm 1\}} \ \left(
	\tikz[baseline={([yshift=-.5*11pt*0.8]current bounding box.center)},scale=0.8,font=\small]{
		\draw (0,1) node[left]{$\beta$} -- (2,1) node[right]{$\delta$};
		\draw (1,0) node[below]{$\alpha$} -- (1,2) node[above]{$\gamma$};
	} \right) \, \ket{\gamma,\delta} \ = \sum_{\gamma,\delta \in \{\pm 1\}} \wspec{u}{\alpha}{\beta}{\gamma}{\delta} \, \ket{\delta,\gamma} \ ,
\end{equation}
where we also indicated the dependence of the vertex weights on~$u\in\mathbb{C}$ on the right-hand side. 

We point out that one has to be a bit careful when reading off the coefficients for the `outgoing' vector $\ket{\gamma,\delta}\in V_a \otimes V_l$ in our graphical notation: unlike for the `incoming' vector $\ket{\beta,\alpha}$, the order of the labels $\delta$ and $\gamma$ is reversed in the coefficients in the middle and on the right-hand side of \eqref{s4:eq:Lax operator acting}, cf.\ the labels in \eqref{s4:eq:Lax operator}. This reversal will come back in \eqref{s4:eq:monodromy acting} below.

\begin{exercise}
Use Figure~\ref{s3:fg:six vertices} to check that the matrix of the Lax operator with respect to the (standard) basis
\begin{equation}\label{s4:eq:tensor basis}
	\ket{--} = \ \tikz[baseline={([yshift=-.5*11pt*0.8+1pt]current bounding box.center)},scale=0.8,font=\small]{\draw[dotted] (0,0) -- (1,0);} \otimes \,  \tikz[baseline={([yshift=-.5*11pt*0.8]current bounding box.center)},scale=0.8,font=\small]{\draw[dotted] (0,0) -- (0,1);} \ , \quad 
	\ket{-+} = \ \tikz[baseline={([yshift=-.5*11pt*0.8+1pt]current bounding box.center)},scale=0.8,font=\small]{\draw[dotted] (0,0) -- (1,0);} \otimes \,  \tikz[baseline={([yshift=-.5*11pt*0.8]current bounding box.center)},scale=0.8,font=\small]{\draw[very thick] (0,0) -- (0,1);} \ , \quad 
	\ket{+-} = \ \tikz[baseline={([yshift=-.5*11pt*0.8+1pt]current bounding box.center)},scale=0.8,font=\small]{\draw[very thick] (0,0) -- (1,0);} \otimes \, \tikz[baseline={([yshift=-.5*11pt*0.8]current bounding box.center)},scale=0.8,font=\small]{\draw[dotted] (0,0) -- (0,1);} \ , \quad 
	\ket{++} = \ \tikz[baseline={([yshift=-.5*11pt*0.8+1pt]current bounding box.center)},scale=0.8,font=\small]{\draw[very thick] (0,0) -- (1,0);} \otimes \, \tikz[baseline={([yshift=-.5*11pt*0.8]current bounding box.center)},scale=0.8,font=\small]{\draw[very thick] (0,0) -- (0,1);}
\end{equation}
of $V_a\otimes V_l$ is given by
\begin{equation}\label{s4:eq:Lax operator matrix}
	L_{al}(u) \ = \ \begin{pmatrix}	a(u) & & & \\ & b(u) & c(u) & \\ & c(u) & b(u) & \\ & & & a(u) \end{pmatrix}_{al} \ .
\end{equation}
\end{exercise}

\begin{exercise}
Argue that the ice rule for the vertex weights is equivalent to the invariance of the Lax operator under simultaneous $U(1)_z$-rotations in $V_a$ and $V_l$:
\begin{equation}\label{s4:eq:ice rule Lax}
	[\, S^z_a + S^z_l , \, L_{al}(u) \,] = 0 \ .
\end{equation}
\end{exercise}

\begin{exercise}
Check that \eqref{s4:eq:Lax operator matrix} can be written in a basis-independent way as
\begin{equation}\label{s4:eq:Lax operator via sigmas}
	L_{al}(u) = \frac{a(u)+b(u)}{2} \, \id_{al} \, + \,  c(u) \, \sigma^+_a S^-_l \, + \, c(u) \, \sigma^-_a S^+_l \, + \, \big(a(u)-b(u)\big) \, \sigma^z_a S^z_l \ , 
\end{equation}
and use this to verify \eqref{s4:eq:ice rule Lax} directly via the $\mathfrak{su}(2)$-relations~\eqref{s2:eq:su(2) ladder}. 
\end{exercise}

\paragraph{Trace identities.} Lax operators can be used as building blocks of other, more involved operators. In particular, the transfer matrix~\eqref{s3:eq:transfer matrix} can be constructed as
\begin{equation}\label{s4:eq:transfer matrix}
	t(u) \ = \ \tr_a \big( L_{aL}(u)\,\cdots\,L_{a2}(u)L_{a1}(u) \big) \ = \
	\tikz[baseline={([yshift=-.5*11pt*0.8+8pt]current bounding box.center)}, 
	scale=0.8,font=\scriptsize]{
		\draw (0,1.2) arc(90:270:.2 and .1) -- (2.6,1); \draw[->] (3.4,1) -- (5,1)arc(-90:90:.2 and .1);
		\draw[->] (1,0) node[below]{$1$} -- (1,2);
		\draw[->] (2,0) node[below]{$2$} -- (2,2);
		\draw[->] (4,0) node[below]{$L$} -- (4,2);
		\foreach \y in {-1,...,1} \draw (3+.2*\y,1) node{$\cdot\mathstrut$};
	} \ \in \End(\cH) \ .
\end{equation}
This way of writing the transfer matrix is very useful for the computation of conserved charges~$H_k$ of the \textsc{xxz} and six-vertex models using the trace identities. 

As we will soon see, it is actually convenient to alter \eqref{s3:eq:trace identities} slightly by setting
\begin{equation}\label{s4:eq:trace identities}
	H_k \coloneqq  \choice{\frac{(-1)^k}{\I}}{\I} \left.\frac{\D^k}{\D u^k}\right|_{u=u_*} \log \frac{t(u)}{a(u)^L} \ \in \ \End(\cH) \ .
\end{equation}
From Section~\ref{s3:results} we know that the trace identities yield symmetries for any choice of the value~$u_*$ of the spectral parameter; let us see whether there are any particularly convenient choices. We parametrize vertex weights by \eqref{s3:eq:parametrization} with $\rho=1$, 
\begin{equation}\label{s4:eq:parametrization}
	a(u) = \sinh(u+\I\gamma) \ , \qquad b(u) = \sinh(u) \ , \qquad c = \sinh(\I\gamma) = \I \sin\gamma \ .
\end{equation}
Observe that $b(u)$ vanishes at $u_*=0$, while $a(u_*)=c$. At this point the Lax operator takes a particularly simple form:
\begin{equation}\label{s4:eq:Lax operator at u*}
	L_{al}(u_*) = c \, \perm_{al} \ , 
\end{equation}
where the \emph{permutation operator} (braiding) is defined as
\begin{equation}\label{s4:eq:perm}
	\perm_{al} \ = \ 
	\tikz[baseline={([yshift=-.5*11pt*0.8+1pt]current bounding box.center)},scale=0.8,font=\scriptsize]{
		\draw[->,rounded corners=8pt] (0,1) node[left]{$a$} -| (1,2) node[above]{$l$}; 
		\draw[->,rounded corners=8pt] (1,0) node[below]{$l$} |- (2,1) node[right]{$a$};
	} 
	\ \coloneqq \frac{1}{2} \id_{al} + \, \sigma^+_a S^-_l  + \sigma^-_a S^+_l + \sigma^z_a S^z_l \ \in \End(V_a \otimes V_l) \ .
\end{equation}
\begin{exercise}
To justify this graphical notation, check that the permutation operator switches vectors: $\perm_{al} \ket{\beta,\alpha} = \ket{\alpha,\beta}$ for (basis) vectors $\ket{\beta,\alpha} = \ket{\beta}\otimes\ket{\alpha} \in V_a \otimes V_l$. Show that $\tr_a \perm_{al} = \id_l$ both algebraically and using the graphical notation.
\end{exercise}

\paragraph{Conserved charges.} Now we are all set to compute the first few symmetries using \eqref{s4:eq:trace identities}. $H_0 = \log t(u_*)$ is easy to find. By \eqref{s4:eq:transfer matrix} and~\eqref{s4:eq:Lax operator at u*} we have $t(u_*) = c^L \tr_a\left( \perm_{aL} \cdots  \perm_{a2} \perm_{a1} \right)$. Focus on the product of permutation operators. It is quite standard to rearrange such products by repeated application of the rule~$\perm_{ak}\perm_{al} = \perm_{al} \perm_{kl}$, see e.g.\ Faddeev~\cite{Fad95a}. For a slightly more slick way to do this we exploit the relations in the permutation group by introducing for each permutation~$\pi \in S_{L+1}$ an operator~$\perm_\pi \in \End(V_a \otimes \cH )$ switching vectors in $V_a \otimes \cH = V_a \otimes \, \bigotimes_l V_l$ in the way specified by~$\pi$. For example, from the transposition $(al)\in S_{L+1}$ we recover $\perm_{(al)} = \perm_{al}$. (The above rule now is clear from $\perm_{ak}\perm_{al} = \perm_{(alk)} = \perm_{al} \perm_{kl}$.) We calculate
\begin{equation}\label{s4:eq:t at u*}
	t(u_*) = c^L \tr_a \left( \perm_{(a12\cdots L)} \right) = c^L \tr_a \left( \perm_{a1} \perm_{(12\cdots L)} \right) = c^L \tr_a(\perm_{a1}) \, \perm_{(12\cdots L)} = c^L \, \choice{U}{U^{-1}} \ ,
\end{equation}
where we have used $\tr_a \perm_{al} = \id_l$ and recognized the cyclic permutation operator~$\perm_{(12\cdots L)}$ as the \choice{}{inverse of the} shift operator $U = \E^{\I P} \in \End(\cH)$. Thus 
\begin{equation}\label{s4:eq:H_0}
	H_0 = \choice{\frac{1}{\I}}{\I} \log \frac{t(u_*)}{a(u_*)^L} = P
\end{equation}
is the momentum operator of the \textsc{xxz} spin chain!

\begin{exercise}\label{s3:ex:perm exercise}
Check that $X_i \perm_{\pi} = \perm_{\pi} X_{\pi(i)}$ for any $X_i \in \End(V_i)\subseteq\End(V_a\otimes\cH)$, $i=a,1,\To,L$, by applying both sides to any arbitrary (basis) vector $\ket{\beta,\alpha_1,\To,\alpha_L}=\ket{\beta}\otimes\ket{\Vector{\alpha}}\in V_a\otimes\cH$. 
\end{exercise}
\begin{exercise}
To find $H_1$ it is convenient to set $\check{L}_{al}(u)\coloneqq \perm_{al} L_{al}(u)$. Apply the result of Exercise~\ref{s3:ex:perm exercise} to $X_{al} = \check{L}_{al}'(u)$ to check that $t'(u_*) = c^{-1} t(u_*) \sum_l \check{L}_{l-1,l}'(u_*)$. Next compare $\check{L}_{l-1,l}'(u_*)$ for \eqref{s4:eq:parametrization} with \eqref{s2:eq:xxz Ham L=2}. Finally use \eqref{s4:eq:trace identities} to show that $H_1$ is nothing but the \textsc{xxz} Hamiltonian in disguise:
\begin{equation}\label{s4:eq:H_1}
	H_1 = \I\, t(u_*)^{-1} t'(u_*) - \I \, L \, \frac{a'(u_*)}{a(u_*)} \id = \frac{2}{\sin\gamma} \, ( H_\textsc{xxz} - E_0 \id ) \ ,
\end{equation}
where $E_0$ is the vacuum energy~\eqref{s2:eq:xxz energy pseudovacuum}. (Note that $H_1$'s eigenvalues are proportional to the $\varepsilon_M$.)
\end{exercise}
The higher~$H_k$ can in principle be computed in a similar fashion. The result is a sum of more and more nonlocal operators: $H_2$ consists of terms involving next-to-nearest neighbour interactions, see~\cite[Ex.\ 2.7]{GRS96}, and so on. 

We conclude that the one-parameter family of commuting transfer matrices~$t(u)$ of the six-vertex model contain important observables of the \textsc{xxz} spin chain. The trace identities in fact provide a concrete relation between the two sides, connecting physical properties of the \textsc{xxz} model, such as the momentum and energy, to the eigenvalues of the transfer matrix, determining the partition function. Together with an equation relating the correlation functions of the two models (see e.g.~\cite[\textsection~2.3]{JM99}), this establishes the precise correspondence between the models from Sections \ref{s2} and \ref{s3}.

\subsection{The Yang-Baxter algebra}\label{s4:YBA}

In Section~\ref{s3:results} we found, for fixed $\Delta$, a family of commuting operators $t(u)\in\End(\cH)$ that generate symmetries~$H_k$ rendering the \textsc{xxz} and six-vertex models quantum integrable. In this subsection we get to the heart of the \textsc{qism} starting from such commuting transfer matrices. It turns out that there is a sufficient `local' condition: the \emph{fundamental commutation relations} (\textsc{fcr}). These relations are closely related to the \emph{Yang-Baxter equation} (\textsc{ybe}) for the \emph{$R$-matrix}.

\paragraph{Monodromy matrix.} Rather than directly imposing horizontal periodicity to obtain the transfer matrix it is useful to define the `global' \emph{monodromy matrix} on $V_a\otimes \cH$ as an ordered product of Lax operators:
\begin{equation}\label{s4:eq:monodromy}
\begin{aligned}
	T_a (u) \ \coloneqq \ \ordprod_{l\in\mathbb{Z}_L } L_{al} (u) \ \coloneqq & \ L_{aL}(u)\,\cdots\,L_{a2}(u)L_{a1}(u) \\ = & \
	\tikz[baseline={([yshift=-.5*11pt*0.8+8pt]current bounding box.center)}, 
	scale=0.8,font=\scriptsize,triple/.style={postaction={draw,-,shorten >=.05},double,double distance=4pt,-implies}]{
		\draw (0,1) node[left]{$a$} -- (2.6,1); \draw[->] (3.4,1) -- (5,1);
		\draw[->] (1,0) node[below]{$1$} -- (1,2);
		\draw[->] (2,0) node[below]{$2$} -- (2,2);
		\draw[->] (4,0) node[below]{$L$} -- (4,2);
		\foreach \y in {-1,...,1} \draw (3+.2*\y,1) node{$\cdot\mathstrut$};
	} \ \in \End(V_a\otimes \cH) \ .
\end{aligned}
\end{equation}
The harpoon in `$\ordprod$' points in the direction of increasing~$l$; notice that this order of the Lax operators is consistent with the order indicated by the little arrows in our graphical notation. (Also note that, as always, subscripts corresponding to the `global' space~$\cH$ are omitted in the tensor-leg notation.) The transfer matrix $t(u) = \tr_a T_a(u)$ arises as a trace of the monodromy matrix over the auxiliary space, corresponding to horizontal periodic boundary conditions, cf.~\eqref{s3:eq:transfer matrix}.

Like in~\eqref{s4:eq:Lax operator acting} one should notice that the order of the labels $\delta$ and $\Vector{\gamma}\,$ in the coefficients of the `outgoing' vectors are reversed in the graphical notation:
\begin{equation}\label{s4:eq:monodromy acting}
	T_a(u) \, \ket{\beta,\Vector{\alpha}} \ = \ \sum_{\delta\in\{\pm1\}} \sum_{\Vector{\gamma} \in \{\pm 1\}^L} \left(
	\tikz[baseline={([yshift=-.5*11pt*0.8]current bounding box.center)},
	scale=0.8,font=\small,triple/.style={postaction={draw,-,shorten >=.05},double,double distance=4pt,-implies}]{
		\draw (0,1) node[left]{$\beta$} -- (2.6,1); \draw (3.4,1) -- (5,1) node[right]{$\delta$};
		\draw (1,0) node[below]{$\alpha_1$} -- (1,2) node[above]{$\gamma_1$};
		\draw (2,0) node[below]{$\alpha_2$} -- (2,2) node[above]{$\gamma_2$};
		\draw (4,0) node[below]{$\alpha_L$} -- (4,2) node[above]{$\gamma_L$};
		\foreach \y in {-1,...,1} \draw (3+.2*\y,1) node{$\cdot\mathstrut$};
	} \right) \ \ket{\delta,\Vector{\gamma}} \ .
\end{equation}

\paragraph{$RTT$-relation.} For graphical computations it is often convenient to depict vectors in the global Hilbert space~\eqref{s3:eq:Hilb global} simply by a single `fat' arrow, which we indicate by a triple line. This leads to following graphical shorthand:
\begin{equation}\label{s4:eq:shorthand}
	T_a(u) \ = \ \ 
	\tikz[baseline={([yshift=-.5*11pt*0.8+8pt]current bounding box.center)}, 
	scale=0.8,font=\scriptsize,triple/.style={postaction={draw,-,shorten >=.05},double,double distance=4pt,-implies}]{
		\draw[triple] (1,0) node[below]{$1\cdots L$} -- (1,2);
		\draw[->] (0,1) node[left]{$a$} -- (2,1);
	}
	\ \ , \qquad\qquad t(u) \ = \ \ 
	\tikz[baseline={([yshift=-.5*11pt*0.8+8pt]current bounding box.center)}, 
		scale=0.8,font=\scriptsize,triple/.style={postaction={draw,-,shorten >=.05},double,double distance=4pt,-implies}]{
		\draw[triple] (1,0) node[below]{$1\cdots L$} -- (1,2);
		\draw[->] (0,1.2) arc(90:270:.2 and .1) -- (2,1) arc(-90:90:.2 and .1);
	} \ \ .
\end{equation}
The commutativity~\eqref{s3:eq:commuting t's v2} of $t(u)$ and~$t(v)$ can then be depicted, cf.~\eqref{s3:eq:tu tv}, as
\begin{equation}\label{s4:eq:tu tv}
	\tikz[baseline={([yshift=-.5*11pt*0.8+8pt]current bounding box.center)}, 
	scale=0.8,font=\scriptsize,triple/.style={postaction={draw,-,shorten >=.05},double,double distance=4pt,-implies}]{
		\draw[triple] (1,0) node[below]{$1\cdots L$} -- (1,3);
		\foreach \y in {1,2} \draw[->] (0,\y+.2) arc(90:270:.2 and .1) -- (2,\y) arc(-90:90:.2 and .1);
		\draw (3,1) node{\small{$v$}} (3,2) node{\small{$u$}};
	}
	\qquad = \qquad
	\tikz[baseline={([yshift=-.5*11pt*0.8+8pt]current bounding box.center)}, 
	scale=0.8,font=\scriptsize,triple/.style={postaction={draw,-,shorten >=.05},double,double distance=4pt,-implies}]{
		\draw[triple] (1,0) node[below]{$1\cdots L$} -- (1,3);
		\foreach \y in {1,2} \draw[->] (0,\y+.2) arc(90:270:.2 and .1) -- (2,\y) arc(-90:90:.2 and .1);
		\draw (3,1) node{\small{$u$}} (3,2) node{\small{$v$}};
	} \ \ .
\end{equation}

Now consider two copies of the auxiliary space, $V_a$ and~$V_b$, with spectral parameters $u$ and~$v$. Of course \eqref{s4:eq:tu tv} holds if $T_a(u)$ and~$T_b(v)$ would commute in~$\End(V_a\otimes V_b \otimes \cH)$,
\begin{equation}\label{s4:eq:Tu Tv?}
	\tikz[baseline={([yshift=-.5*11pt*0.8+8pt]current bounding box.center)},	scale=0.8,font=\scriptsize,triple/.style={postaction={draw,-,shorten >=.05},double,double distance=4pt,-implies}]{
		\draw[triple] (1,0) node[below]{$1\cdots L$} -- (1,3);
		\draw[->] (0,2) node[left]{$a$} -- (2,2);
		\draw[->] (0,1) node[left]{$b$} -- (2,1);
	}
	\qquad \overset{?}{=} \qquad
	\tikz[baseline={([yshift=-.5*11pt*0.8+8pt]current bounding box.center)}, 
	scale=0.8,font=\scriptsize,triple/.style={postaction={draw,-,shorten >=.05},double,double distance=4pt,-implies}]{
		\draw[triple] (1,0) node[below]{$1\cdots L$} -- (1,3);
		\draw[->] (0,2) node[left]{$b$} -- (2,2);
		\draw[->] (0,1) node[left]{$a$} -- (2,1);
	} \ \ ,
\end{equation}
but a direct check (even for $L=1$, in which case $T_a$ is just the Lax operator) shows that this is not true for generic values of $u$ and~$v$.

Exploiting the horizontal periodicity, however, we can write down another equation that is not too restrictive while still guaranteeing~\eqref{s4:eq:tu tv}: the \textsc{fcr}. For this we need an operator $R_{ab}(w)\in\End(V_a\otimes V_b)$, rather unimaginatively called the \emph{$R$-matrix}. This operator is allowed to depend on some spectral parameter~$w$, and should be invertible for generic values of~$w$. We depict the $R$-matrix and its inverse as
\begin{equation}\label{s4:eq:R graphical}
	R_{ab}(w) \ = \
	\tikz[baseline={([yshift=-.5*11pt*0.8]current bounding box.center)},cross line/.style={-,preaction={draw=white,-,line width=6pt}},scale=0.8,font=\scriptsize]{
		\draw[->] (0,0) node[left]{$b$} -- (1,1);
		\draw[cross line,->] (0,1) node[left]{$a$} -- (1,0);
	} 
	\ \ , \qquad 
	R_{ab}^{-1}(w) \ = \
	\tikz[baseline={([yshift=-.5*11pt*0.8]current bounding box.center)},cross line/.style={-,preaction={draw=white,-,line width=6pt}},scale=0.8,font=\scriptsize]{
		\draw[->] (0,1) node[left]{$b$} -- (1,0);
		\draw[cross line,->] (0,0) node[left]{$a$} -- (1,1);
	}
	\ \ ,
\end{equation}
so that the products
\begin{equation}\label{s4:eq:R and R^-1}
	\tikz[baseline={([yshift=-.5*11pt*0.8]current bounding box.center)},cross line/.style={-,preaction={draw=white,-,line width=6pt}},scale=0.8,font=\scriptsize]{
		\draw[->] (0,1) node[left]{$b$} -- (1,0) -- (2,1);
		\draw[cross line,->] (0,0) node[left]{$a$} -- (1,1) -- (2,0);
	} 
	\ = \ 
	\tikz[baseline={([yshift=-.5*11pt*0.8]current bounding box.center)},scale=0.8,font=\scriptsize]{
			\draw[->] (0,0) node[left]{$a$} -- (2,0);
			\draw[->] (0,1) node[left]{$b$} -- (2,1);
	} \coloneqq \id_{ab} \ \ , \qquad
	\tikz[baseline={([yshift=-.5*11pt*0.8]current bounding box.center)},cross line/.style={-,preaction={draw=white,-,line width=6pt}},scale=0.8,font=\scriptsize]{
		\draw[->] (0,0) node[left]{$b$} -- (1,1) -- (2,0);
		\draw[cross line,->] (0,1) node[left]{$a$} -- (1,0) -- (2,1);
	}
	\ = \ 
	\tikz[baseline={([yshift=-.5*11pt*0.8]current bounding box.center)},scale=0.8,font=\scriptsize]{
		\draw[->] (0,1) node[left]{$a$} -- (2,1);
		\draw[->] (0,0) node[left]{$b$} -- (2,0);
	}
	\ \coloneqq \id_{ba}
\end{equation}
are the identity operators, while the square of either operator in~\eqref{s4:eq:R graphical} is not. (Since $u$ and~$v$ are associated to $V_a$ and~$V_b$ respectively, cf.~\eqref{s4:eq:tu tv}, one may hope to be able to express $w$ in terms of $u$ and~$v$; this will indeed be the case.) As for the Lax operator and the monodromy matrix, the order of the `incoming' and `outgoing' labels is reversed for the coefficients in the graphical notation:
\begin{equation}\label{s4:eq:R acting}
	R_{ab}(w) \, \ket{\beta,\beta'} \ = \sum_{\delta,\delta' \in \{\pm 1\}} \ \left(
	\tikz[baseline={([yshift=-.5*11pt*0.8]current bounding box.center)},cross line/.style={-,preaction={draw=white,-,line width=6pt}},scale=0.8,font=\small]{
		\draw[->] (0,0) node[left]{$\beta'$} -- (1,1) node[right]{$\delta'$};
		\draw[cross line,->] (0,1) node[left]{$\beta$} -- (1,0) node[right]{$\delta$};
	} \right) \, \ket{\delta,\delta'} \ . 
\end{equation}
These coefficients have to be determined.

The use of the $R$-matrix comes from two theorems that give `global' and `local' conditions on the $R$-matrix guaranteeing \eqref{s4:eq:tu tv}. We start with the `global' theorem:

\begin{thm}\label{s4:th:RTT thm}
If there exists an $R$-matrix~$R_{ab}(w) \in \End(V_a\otimes V_b)$ which
\begin{enumerate}[label=\roman*),noitemsep]
\item is generically invertible;
\item satisfies the following `global' \textsc{fcr} in~$\End(V_a\otimes V_b \otimes \cH)$:
\begin{equation}\label{s4:eq:RTT}
	R_{ab}(w) \, T_a(u) \, T_b(v) = T_b(v) \, T_a(u) \, R_{ab}(w) \ ,
\end{equation}
\end{enumerate} 
then the transfer matrices $t(u)$ and~$t(v)$ commute.
\end{thm}

\begin{proof}
Keeping track of the order in which the operators act, \eqref{s4:eq:RTT} can be depicted as 
\begin{equation}\label{s4:eq:RTT graphical}
	\tikz[baseline={([yshift=-.5*11pt*0.8+8pt]current bounding box.center)},	scale=0.8,font=\scriptsize,triple/.style={postaction={draw,-,shorten >=.05},double,double distance=4pt,-implies},cross line/.style={-,preaction={draw=white,-,line width=6pt}}]{
		\draw[triple] (1,0) node[below]{$1\cdots L$} -- (1,3);
		\draw (0,2) node[left]{$a$} -- (2,2);
		\draw[->] (0,1) node[left]{$b$} -- (2,1) -- (3,2);
		\draw[cross line,->] (2,2) -- (3,1);
	}
	\qquad = \qquad
	\tikz[baseline={([yshift=-.5*11pt*0.8+8pt]current bounding box.center)}, 
	scale=0.8,font=\scriptsize,triple/.style={postaction={draw,-,shorten >=.05},double,double distance=4pt,-implies},cross line/.style={-,preaction={draw=white,-,line width=6pt}}]{
		\draw[triple] (2,0) node[below]{$1\cdots L$} -- (2,3);
		\draw[->] (0,1) node[left]{$b$} -- (1,2) -- (3,2);
		\draw[cross line] (0,2) node[left]{$a$} -- (1,1);
		\draw[->] (1,1) -- (3,1);
	} \ \ .
\end{equation}
By multiplying both sides in~\eqref{s4:eq:RTT graphical} from the left by the inverse of the $R$-matrix we obtain the equivalent relation 
\begin{equation}\label{s4:eq:RTTR}
	\tikz[baseline={([yshift=-.5*11pt*0.8+8pt]current bounding box.center)},	scale=0.8,font=\scriptsize,triple/.style={postaction={draw,-,shorten >=.05},double,double distance=4pt,-implies},cross line/.style={-,preaction={draw=white,-,line width=6pt}}]{
		\draw[triple] (2,0) node[below]{$1\cdots L$} -- (2,3);
		\draw[->] (0,2) node[left]{$b$} -- (1,1) -- (3,1) -- (4,2);
		\draw[cross line] (0,1) node[left]{$a$} -- (1,2);
		\draw (1,2) -- (3,2);
		\draw[cross line,->] (3,2) -- (4,1);
	}
	\qquad = \qquad
	\tikz[baseline={([yshift=-.5*11pt*0.8+8pt]current bounding box.center)}, 
	scale=0.8,font=\scriptsize,triple/.style={postaction={draw,-,shorten >=.05},double,double distance=4pt,-implies}]{
		\draw[triple] (1,0) node[below]{$1\cdots L$} -- (1,3);
		\draw[->] (0,2) node[left]{$b$} -- (2,2);
		\draw[->] (0,1) node[left]{$a$} -- (2,1);
	} \ \ .
\end{equation}
Taking the trace over both auxiliary spaces we conclude that $t(u)$ and~$t(v)$ commute using the cyclic property of the trace.
\end{proof}

Thus we ask for the monodromy matrices $T_a(u)$ and~$T_b(v)$ to commute up to conjugation by the $R$-matrix (in more algebraic terms: we want the $R$-matrix to intertwine the actions of the two monodromies). Equation \eqref{s4:eq:RTT} is often referred to as the \emph{$RTT$-relation} for obvious reasons. It is global in the sense that it involves operators $T_a$ and~$T_b$ acting on the `global' Hilbert space~$\cH$. Since $T_a(u) \, T_b(v)=\big(T_a(u)\otimes\id\big)\big(\id\otimes\, T_b(v)\big)=T_a(u) \otimes T_b(v)$, \eqref{s4:eq:RTT} can rewritten as 
\begin{equation}\label{s4:eq:RTT alt}
	R_{ab}(w) \, \big(T(u) \otimes T(v)\big){}_{ab} = \big(T(v) \otimes T(u)\big){}_{ab} \, R_{ab}(w) \ .
\end{equation}

Before we continue let us simplify our graphical notation a bit. In \eqref{s4:eq:R graphical} we used under- and overcrossings to distinguish between the $R$-matrix and its inverse, both of which were necessary for the above proof. However we have no further need to depict $R^{-1}_{ab}(w)$ from now on. Thus we may drop this inverse from our graphical notation, and update \eqref{s4:eq:R graphical} to the simpler rule
\begin{equation}\label{s4:eq:R graphical new}
	R_{ab}(w) \ = \
		\tikz[baseline={([yshift=-.5*11pt*0.8]current bounding box.center)},scale=0.8,font=\scriptsize]{
			\draw[->] (0,1) node[left]{$a$} -- (1,0);
			\draw[->] (0,0) node[left]{$b$} -- (1,1);
		} \ \ .
\end{equation}
In this notation the $R$-matrix still differs from the Lax operator by the labels of the lines.

The task of finding a suitable $R$-matrix is simplified by the following `local' version of Theorem~\ref{s4:th:RTT thm}.

\begin{thm}\label{s4:th:FCR thm}
If there exists an $R$-matrix~$R_{ab}(w) \in \End(V_a\otimes V_b)$ which
\begin{enumerate}[label=\roman*),noitemsep]
\item is generically invertible;
\item satisfies for any (and hence all) $l\in\mathbb{Z}_L$ the following `local' \textsc{fcr} in~$\End(V_a\otimes V_b \otimes V_l)$:
\begin{equation}\label{s4:eq:FCR}
	R_{ab}(w) \, L_{al}(u) \, L_{bl}(v) = L_{bl}(v) \, L_{al}(u) \, R_{ab}(w) \ ,
\end{equation}
\end{enumerate} 
then the transfer matrices $t(u)$ and~$t(v)$ commute.
\end{thm}

\begin{proof}
By Theorem~\ref{s4:th:RTT thm} it suffices to show that the $RTT$-relation~\eqref{s4:eq:RTT} is equivalent to \eqref{s4:eq:FCR}. Since the latter is obtained from \eqref{s4:eq:RTT} in the special case with only one lattice site~($L=1$), which we may always label by~$l$, the $RTT$-relation implies the local \textsc{fcr}. To see that \eqref{s4:eq:FCR} is also sufficient we use the following graphical (yet rigorous!) `train argument'.

Suppose that \eqref{s4:eq:FCR} holds, so that Lax operators acting in the same local physical space, but different auxiliary spaces, commute up to conjugation by the $R$-matrix. Diagrammatically \eqref{s4:eq:FCR} says that the vertical line, corresponding to Lax operators acting on two vertices above each other, can be moved through the crossing representing the $R$-matrix:
\begin{equation}\label{s4:eq:FCR graphical}
	\tikz[baseline={([yshift=-.5*11pt*0.8+8pt]current bounding box.center)},
	scale=0.8,font=\scriptsize,triple/.style={postaction={draw,-,shorten >=.05},double,double distance=4pt,-implies}]{
		\draw[->] (0,1) node[left]{$b$} -- (2,1) -- (3,2);
		\draw[->] (0,2) node[left]{$a$} -- (2,2) -- (3,1);
		\draw[->] (1,0) node[below]{$l$} -- (1,3);
	}
	\qquad = \qquad
	\tikz[baseline={([yshift=-.5*11pt*0.8+8pt]current bounding box.center)},
	scale=0.8,font=\scriptsize,triple/.style={postaction={draw,-,shorten >=.05},double,double distance=4pt,-implies}]{
		\draw[->] (0,1) node[left]{$b$} -- (1,2) -- (3,2);
		\draw[->] (0,2) node[left]{$a$} -- (1,1) -- (3,1);
		\draw[->] (2,0) node[below]{$l$} -- (2,3);
	} \ \ .
\end{equation}
From the definition~\eqref{s4:eq:monodromy} of the monodromy matrix, $L$ applications of \eqref{s4:eq:FCR graphical} do the job: 
\begin{align*}
	\tikz[baseline={([yshift=-.5*11pt*0.8+8pt]current bounding box.center)},
	scale=0.8,font=\scriptsize,triple/.style={postaction={draw,-,shorten >=.05},double,double distance=4pt,-implies}]{
		\draw[->] (0,1) node[left]{$b$} -- (2.6,1) (3.4,1) -- (5,1) -- (6,2);
		\draw[->] (0,2) node[left]{$a$} -- (2.6,2) (3.4,2) -- (5,2) -- (6,1);
		\foreach \y in {-1,...,1} \draw (.2*\y+3,1) node{$\cdot\mathstrut$};
		\foreach \y in {-1,...,1} \draw (.2*\y+3,2) node{$\cdot\mathstrut$};
		\draw[->] (1,0) node[below]{$1$} -- (1,3);
		\draw[->] (2,0) node[below]{$2$} -- (2,3);
		\draw[->] (4,0) node[below]{$L$} -- (4,3);
	}
	\qquad & = \qquad
	\tikz[baseline={([yshift=-.5*11pt*0.8+8pt]current bounding box.center)},
	scale=0.8,font=\scriptsize,triple/.style={postaction={draw,-,shorten >=.05},double,double distance=4pt,-implies}]{
		\draw[->] (0,1) node[left]{$b$} -- (2.6,1) (3.4,1) -- (4,1) -- (5,2) -- (7,2);
		\draw[->] (0,2) node[left]{$a$} -- (2.6,2) (3.4,2) -- (4,2) -- (5,1) -- (7,1);
		\foreach \y in {-1,...,1} \draw (.2*\y+3,1) node{$\cdot\mathstrut$};
		\foreach \y in {-1,...,1} \draw (.2*\y+3,2) node{$\cdot\mathstrut$};
		\draw[->] (1,0) node[below]{$1$} -- (1,3);
		\draw[->] (2,0) node[below]{$2$} -- (2,3);
		\draw[->] (6,0) node[below]{$L$} -- (6,3);
	} \\
	& = \qquad \cdots \\
	& = \qquad
	\tikz[baseline={([yshift=-.5*11pt*0.8+8pt]current bounding box.center)},
	scale=0.8,font=\scriptsize,triple/.style={postaction={draw,-,shorten >=.05},double,double distance=4pt,-implies}]{
		\draw[->] (0,1) node[left]{$b$} -- (1,2) -- (3.6,2) (4.4,2) -- (5,2) -- (6,2);
		\draw[->] (0,2) node[left]{$a$} -- (1,1) -- (3.6,1) (4.4,1) -- (5,1) -- (6,1);
		\foreach \y in {-1,...,1} \draw (.2*\y+4,1) node{$\cdot\mathstrut$};
		\foreach \y in {-1,...,1} \draw (.2*\y+4,2) node{$\cdot\mathstrut$};
		\draw[->] (2,0) node[below]{$1$} -- (2,3);
		\draw[->] (3,0) node[below]{$2$} -- (3,3);
		\draw[->] (5,0) node[below]{$L$} -- (5,3);
	}  \ \ .
\end{align*} 
\\[-2\baselineskip]
\end{proof}

In statistical mechanics \eqref{s4:eq:FCR} is often referred to as the \emph{star-triangle relation}.$^\#$\footnote{We reserve the name `Yang-Baxter equation' for the rather similar and intimately related (but algebraically still more fundamental) equation that we will encounter soon.} If a model admits an $R$-matrix that satisfies conditions~(i)--(ii) one can construct symmetries~$H_k$ from the transfer matrix as described in Section~\ref{s4:conserved quantities}. Such models are moreover solvable via techniques like the algebraic Bethe ansatz, see Section~\ref{s4:ABA}. For this reason such an $R$-matrix is the \emph{integrability datum} allowing one to study quantum-integrable models from an algebraic point of view. Indeed, in practice one often uses this structure to define `quantum integrability'.

\paragraph{$R$-matrix.} The upshot of the preceding discussion is that if we can find an $R$-matrix satisfying the \textsc{fcr}~\eqref{s4:eq:FCR} for a given Lax operator then the transfer matrices constructed from that Lax operator commute. In Appendix~\ref{sR} the \textsc{fcr} of the \textsc{xxz}/six-vertex model, with Lax operator~\eqref{s4:eq:Lax operator}, is solved for a nontrivial $R$-matrix that respects the symmetries of the model: it satisfies both line conservation (the ice rule) and spin-reversal symmetry. The result is of the same form as the Lax operator~\eqref{s4:eq:Lax operator matrix}:
\begin{equation}\label{s4:eq:R-matrix solution}
	R_{ab}(w) \ = \ \begin{pmatrix}	a(w) & & & \\ & b(w) & c(w) & \\ & c(w) & b(w) & \\ & & & a(w) \end{pmatrix}_{ab} \ ,
\end{equation}
where the functions $a$, $b$ and~$c$ were defined in~\eqref{s4:eq:parametrization}. Accordingly, the entries of the $R$-matrix may be interpreted as the vertex weights of another six-vertex model with the same value of~$\Delta=\cos\gamma$ but with different spectral parameter~$w$.

Clearly \eqref{s4:eq:R-matrix solution} is indeed invertible for almost all values of the spectral parameter. In Appendix~\ref{sR} we further show that this $R$-matrix solves the \textsc{fcr} provided the spectral parameters are related by the difference property $w=u-v$. (Of course one can also directly check that in this case \eqref{s4:eq:R-matrix solution} does indeed satisfy the \textsc{fcr}, see e.g.~\cite[\textsection10]{Fad95a}.) Due to the difference property the \textsc{fcr} is often written as
\begin{equation}\label{s4:eq:FCR w=u-v}
	R_{ab}(u-v) \, L_{al}(u) \, L_{bl}(v) = L_{bl}(v) \, L_{al}(u) \, R_{ab}(u-v) \ ,
\end{equation}
and likewise for the $RTT$-relation. This result nicely fits in the graphical notation if we straighten out the lines in~\eqref{s4:eq:FCR graphical}:
\begin{equation}\label{s4:eq:FCR graphical alt}
	\tikz[baseline={([yshift=-.5*11pt*0.8]current bounding box.center)},
	scale=0.8,font=\scriptsize]{
		\pgfmathsetmacro\secA{1/cos(180-150)}
		\pgfmathsetmacro\tanA{tan(180-150)}
		\pgfmathsetmacro\secB{1/cos(180-140)}
		\pgfmathsetmacro\tanB{tan(180-140)}
			\draw[->] (150:1+\secA) coordinate(a1) node[left]{$a$} -- (-30:1) coordinate(b3);
			\draw[->] ($(-90:1+\tanB)-(1,0)$) coordinate(a2) node[below]{$l$} -- ($(90:1+\tanA)-(1,0)$) coordinate(b2);
			\draw[->] (-140:1+\secB) coordinate(a3) node[left]{$b$} -- (40:1) coordinate(b1);
			\coordinate(c1) at (intersection of a1--b3 and a2--b2); \draw[thin,densely dotted] ($(c1)+(150:.4)$) arc(150:270:.4); \draw ($(c1)+(150+60:.65)$) node{\small $u$};
			\coordinate(c2) at (intersection of a1--b3 and a3--b1); \draw[thin,densely dotted] ($(c2)+(150:.4)$) arc(150:220:.4); \draw ($(c2)+(-85:.6)$) node{\small $u-v$};
			\coordinate(c3) at (intersection of a2--b2 and a3--b1); \draw[thin,densely dotted] ($(c3)+(-140:.4)$) arc(-140:-90:.4); \draw ($(c3)+(-90-20:.65)$) node{\small $v$};
	} \qquad = \qquad
	\tikz[baseline={([yshift=-.5*11pt*0.8]current bounding box.center)},
	scale=0.8,font=\scriptsize]{
		\pgfmathsetmacro\secA{1/cos(180-150)}
		\pgfmathsetmacro\tanA{tan(180-150)}
		\pgfmathsetmacro\secB{1/cos(180-140)}
		\pgfmathsetmacro\tanB{tan(180-140)}
			\draw[->] (150:1) coordinate(a1) node[left]{$a$} -- (-30:1+\secA) coordinate(b3);
			\draw[->] ($(-90:1+\tanA)+(1,0)$) coordinate(a2) node[below]{$l$} -- ($(90:1+\tanB)+(1,0)$) coordinate(b2);
			\draw[->] (-140:1) coordinate(a3) node[left]{$b$} -- (40:1+\secB) coordinate(b1);
			\coordinate(c1) at (intersection of a1--b3 and a2--b2); \draw[thin,densely dotted] ($(c1)+(150:.4)$) arc(150:270:.4); \draw ($(c1)+(150+60:.65)$) node{\small $u$};
			\coordinate(c2) at (intersection of a1--b3 and a3--b1); \draw[thin,densely dotted] ($(c2)+(150:.4)$) arc(150:220:.4); \draw ($(c2)+(195:.4)$) node[left]{\small $u-v$};
			\coordinate(c3) at (intersection of a2--b2 and a3--b1); \draw[thin,densely dotted] ($(c3)+(-140:.4)$) arc(-140:-90:.4); \draw ($(c3)+(-90-20:.65)$) node{\small $v$};
	}
\end{equation}
Here the spectral parameters of the operators are included as angles, and the \textsc{fcr} says that any single line may be shifted past the intersection point of the other two lines if it is kept parallel to the original line. (Note that such shifts are in fact used in Section~\ref{s5:fact scatt} to \emph{derive} the \textsc{ybe} for factorized scattering.)

\paragraph{Yang-Baxter algebra.} The following algebraic construction lies at the core of the \textsc{qism} and provides the mathematical setting for the computations in the framework of the algebraic Bethe ansatz, as we will see in Section~\ref{s4:ABA}. Since the monodromy matrix $T_a(u) \in \End(V_a\otimes \cH)$ also acts in auxiliary space we can write it as a matrix on $V_a$,
\begin{equation}\label{s4:eq:monodromy matrix}
	T_a(u) =  \begin{pmatrix} A(u) & B(u) \\ C(u) & D(u) \end{pmatrix}_a \ ,
\end{equation}
whose entries act on the physical space~$\cH$ of the spin chain.

\begin{exercise}
Check that in our graphical notation
\begin{equation}\label{s4:eq:ABCD graphical}
\begin{aligned}
	& A(u) \ = \ \tikz[baseline={([yshift=-.5*11pt*0.8+8pt]current bounding box.center)}, 
	scale=0.8,font=\scriptsize,triple/.style={postaction={draw,-,shorten >=.05},double,double distance=4pt,-implies}]{
		\draw[dotted] (0,1) -- (2,1);
		\draw[triple] (1,0) node[below]{$1\cdots L$} -- (1,2);
	} \ \ , \qquad
	& B(u) \ = \ \tikz[baseline={([yshift=-.5*11pt*0.8+8pt]current bounding box.center)}, 
	scale=0.8,font=\scriptsize,triple/.style={postaction={draw,-,shorten >=.05},double,double distance=4pt,-implies}]{
		\draw[very thick] (0,1) -- (1,1);
		\draw[dotted] (1,1) -- (2,1);
		\draw[triple] (1,0) node[below]{$1\cdots L$} -- (1,2);
	} \ \ , \\
	& C(u) \ = \ \tikz[baseline={([yshift=-.5*11pt*0.8+8pt]current bounding box.center)}, 
	scale=0.8,font=\scriptsize,triple/.style={postaction={draw,-,shorten >=.05},double,double distance=4pt,-implies}]{
		\draw[dotted] (0,1) -- (1,1);
		\draw[very thick] (1,1) -- (2,1);
		\draw[triple] (1,0) node[below]{$1\cdots L$} -- (1,2);
	} \ \ , \qquad
	& D(u) \ = \ \tikz[baseline={([yshift=-.5*11pt*0.8+8pt]current bounding box.center)}, 
	scale=0.8,font=\scriptsize,triple/.style={postaction={draw,-,shorten >=.05},double,double distance=4pt,-implies}]{
		\draw[very thick] (0,1) -- (2,1);
		\draw[triple] (1,0) node[below]{$1\cdots L$} -- (1,2);
	} \ \ .
\end{aligned}
\end{equation}
\end{exercise}

These (one-parameter families of) operators in $\End(\cH)$ generate a (unital, associative) algebra, known as the \emph{Yang-Baxter algebra}~(\textsc{yba}), whose commutation rules are given by the $RTT$-relation~\eqref{s4:eq:RTT} with $w=u-v$. The latter encodes $2^2 \times 2^2=16$ relations in $\End(V_a\otimes V_b)$ for the generators~\eqref{s4:eq:ABCD graphical}. The explicit form of these relations can be found from~\eqref{s4:eq:RTT alt} by straightforward matrix multiplication, see \cite[\textsection4]{Fad95a}. Instead one can also use the graphical form~\eqref{s4:eq:RTT graphical} of the $RTT$-relation to find these relations. For example, the $(1,4)$-entry of \eqref{s4:eq:RTT alt} corresponds to
\begin{equation}\label{s4:eq:BB graphical} 
	\tikz[baseline={([yshift=-.5*11pt*0.8+8pt]current bounding box.center)}, 
		scale=0.8,font=\scriptsize,triple/.style={postaction={draw,-,shorten >=.05},double,double distance=4pt,-implies}]{
		\draw[dotted] (1,1) -- (2,1) -- (3,2) (1,2) -- (2,2) -- (3,1);
		\draw[very thick] (0,1) -- (1,1) (0,2) -- (1,2);
		\draw[triple] (1,0) node[below]{$1\cdots L$} -- (1,3);
	}
	\quad = \quad
	\tikz[baseline={([yshift=-.5*11pt*0.8+8pt]current bounding box.center)}, 
		scale=0.8,font=\scriptsize,triple/.style={postaction={draw,-,shorten >=.05},double,double distance=4pt,-implies}]{
		\draw[dotted] (2,1) -- (3,1) (2,2) -- (3,2);
		\draw[very thick] (0,1) -- (1,2) -- (2,2) (0,2) -- (1,1) -- (2,1);
		\draw[triple] (2,0) node[below]{$1\cdots L$} -- (2,3);
	} \quad , 
\end{equation}
implying that
\begin{equation}\label{s4:eq:BB}
	B(u) \, B(v) = B(v) \, B(u) \ . 
\end{equation}
Likewise, paying attention to the different ordering of `incoming' and `outgoing' auxiliary vectors, the $(1,3)$- and $(3,4)$-entries of \eqref{s4:eq:RTT alt} correspond to
\begin{align}
	& \tikz[baseline={([yshift=-.5*11pt*0.8+8pt]current bounding box.center)}, scale=0.8,font=\scriptsize,triple/.style={postaction={draw,-,shorten >=.05},double,double distance=4pt,-implies}]{
		\draw[dotted] (0,1) -- (2,1) -- (3,2) (1,2) -- (2,2) -- (3,1);
		\draw[very thick] (0,2) -- (1,2);
		\draw[triple] (1,0) node[below]{$1\cdots L$} -- (1,3);
	}
	\ \ = \ \ 
	\tikz[baseline={([yshift=-.5*11pt*0.8+8pt]current bounding box.center)}, scale=0.8,font=\scriptsize,triple/.style={postaction={draw,-,shorten >=.05},double,double distance=4pt,-implies}]{
		\draw[dotted] (0,1) -- (.5,1.5) (2,1) -- (3,1) (2,2) -- (3,2);
		\draw[very thick] (0,2) -- (.5,1.5);
		\draw (2,1) -- (1,1) -- (.5,1.5) -- (1,2) -- (2,2);
		\draw[triple] (2,0) node[below]{$1\cdots L$} -- (2,3);
	}
	\ \ = \ \ 
	\tikz[baseline={([yshift=-.5*11pt*0.8+8pt]current bounding box.center)}, scale=0.8,font=\scriptsize,triple/.style={postaction={draw,-,shorten >=.05},double,double distance=4pt,-implies}]{
		\draw[dotted] (0,1) -- (1,2) -- (3,2) (2,1) -- (3,1);
		\draw[very thick] (0,2) -- (1,1) -- (2,1);
		\draw[triple] (2,0) node[below]{$1\cdots L$} -- (2,3);
	}
	\ \ + \ \
	\tikz[baseline={([yshift=-.5*11pt*0.8+8pt]current bounding box.center)}, scale=0.8,font=\scriptsize,triple/.style={postaction={draw,-,shorten >=.05},double,double distance=4pt,-implies}]{
		\draw[dotted] (0,1) -- (.5,1.5) -- (1,1) -- (3,1) (2,2) -- (3,2);
		\draw[very thick] (0,2) -- (.5,1.5) -- (1,2) -- (2,2);
		\draw[triple] (2,0) node[below]{$1\cdots L$} -- (2,3);
	}
	\quad , \label{s4:eq:AB graphical} \\
	& \tikz[baseline={([yshift=-.5*11pt*0.8+8pt]current bounding box.center)}, scale=0.8,font=\scriptsize,triple/.style={postaction={draw,-,shorten >=.05},double,double distance=4pt,-implies}]{
		\draw[dotted] (1,1) -- (2,1) -- (3,2);
		\draw[very thick] (0,1) -- (1,1) (0,2) -- (2,2)-- (3,1);
		\draw[triple] (1,0) node[below]{$1\cdots L$} -- (1,3);
	}
	\ \ + \ \
	\tikz[baseline={([yshift=-.5*11pt*0.8+8pt]current bounding box.center)}, scale=0.8,font=\scriptsize,triple/.style={postaction={draw,-,shorten >=.05},double,double distance=4pt,-implies}]{
		\draw[dotted] (1,2) -- (2,2) -- (2.5,1.5) -- (3,2);
		\draw[very thick] (0,1) -- (2,1) -- (2.5,1.5) -- (3,1) (0,2) -- (1,2);
		\draw[triple] (1,0) node[below]{$1\cdots L$} -- (1,3);
	}
	\ \ = \ \ 
	\tikz[baseline={([yshift=-.5*11pt*0.8+8pt]current bounding box.center)}, scale=0.8,font=\scriptsize,triple/.style={postaction={draw,-,shorten >=.05},double,double distance=4pt,-implies}]{
		\draw[dotted] (2.5,1.5) -- (3,2);
		\draw[very thick] (0,1) -- (1,1) (0,2) -- (1,2) (2.5,1.5) -- (3,1);
		\draw (1,1) -- (2,1) -- (2.5,1.5) -- (2,2) -- (1,2);
		\draw[triple] (1,0) node[below]{$1\cdots L$} -- (1,3);
	}
	\ \ = \ \ 
	\tikz[baseline={([yshift=-.5*11pt*0.8+8pt]current bounding box.center)}, scale=0.8,font=\scriptsize,triple/.style={postaction={draw,-,shorten >=.05},double,double distance=4pt,-implies}]{
		\draw[dotted] (2,2) -- (3,2);
		\draw[very thick] (0,1) -- (1,2) -- (2,2) (0,2) -- (1,1) -- (3,1);
		\draw[triple] (2,0) node[below]{$1\cdots L$} -- (2,3);
	}
	\ \ . \label{s4:eq:DB graphical} 
\end{align}
Upon interchanging $u\leftrightarrow v$ in \eqref{s4:eq:AB graphical} these yield more complicated commutation rules:
\begin{align}
	& A(u) \, B(v) = \frac{a(v-u)}{b(v-u)} \, B(v) \, A(u) - \frac{c(v-u)}{b(v-u)} \, B(u) \, A(v) \ , \label{s4:eq:AB} \\
	& D(u) \, B(v) = \frac{a(u-v)}{b(u-v)} \, B(v) \, D(u) - \frac{c(u-v)}{b(u-v)} \, B(u) \, D(v) \ . \label{s4:eq:DB}
\end{align}
In both of these relations the first term on the right-hand side just contains the commuted operators (up to a factor), whereas in the second term the two operators have in addition interchanged their spectral parameters. 

The physical use of the \textsc{yba} stems from the \emph{quantum inverse-scattering problem}, which asks whether it is possible to reconstruct arbitrary operators in $\End(V_l)$, and thus those in $\End(\cH)$, from $T_a(u)$. The solution to this problem was found for many models, including the \textsc{xxz}/six-vertex model, in \cite{MT00,*KMT00}. The conclusion is that $A(u),\To,D(u)$ generate all of $\End(\cH)$. For example, the transfer matrix is an element of the Yang-Baxter algebra:
\begin{equation}\label{s4:eq:t=A+D}
	t(u) \ = \
	\tikz[baseline={([yshift=-.5*11pt*0.8+8pt]current bounding box.center)}, 
	scale=0.8,font=\scriptsize,triple/.style={postaction={draw,-,shorten >=.05},double,double distance=4pt,-implies}]{
		\draw[triple] (1,0) node[below]{$1\cdots L$} -- (1,2);
		\draw[->] (0,1.2) arc(90:270:.2 and .1) -- (2,1) arc(-90:90:.2 and .1);
	}
	\ = \ 
	\tikz[baseline={([yshift=-.5*11pt*0.8+8pt]current bounding box.center)}, 
	scale=0.8,font=\scriptsize,triple/.style={postaction={draw,-,shorten >=.05},double,double distance=4pt,-implies}]{
		\draw[dotted] (0,1) -- (2,1);
		\draw[triple] (1,0) node[below]{$1\cdots L$} -- (1,2);
	} 
	\ + \ 
	\tikz[baseline={([yshift=-.5*11pt*0.8+8pt]current bounding box.center)}, 
	scale=0.8,font=\scriptsize,triple/.style={postaction={draw,-,shorten >=.05},double,double distance=4pt,-implies}]{
		\draw[very thick] (0,1) -- (2,1);
		\draw[triple] (1,0) node[below]{$1\cdots L$} -- (1,2);
	}
	\ = \ A(u) + D(u) \ .
\end{equation}
\begin{exercise}
Find relations like \eqref{s4:eq:BB graphical} for $A$ and for $D$. Next use \eqref{s4:eq:RTT graphical} to compute $[A(u),D(v)]$. Check in this way that $t(u)$ and $t(v)$ do indeed commute.
\end{exercise}

\paragraph{Yang-Baxter equation.} The entries of the $R$-matrix play the role of structure constants for the Yang-Baxter algebra. In the present context there also is an analogue of the Jacobi identity for these `structure constants'. Indeed, consider one more copy of the auxiliary space, $V_c$, with associated spectral parameter~$w$. The $RTT$-relation can be used to reverse the order in the product $T_a(u) \, T_b(v) \, T_c(w)$ to get $T_c(w) \, T_b(v) \, T_a(u)$ up to conjugation by products of $R$-matrices. Now there are two ways in which this can be done, corresponding to the two decompositions $(13)=(12)(23)(12)=(23)(12)(23)$ of the permutation switching the first and third monodromy matrix. To avoid $2^3 \times 2^3 = 64$ additional relations for the Yang-Baxter algebra, the two results must coincide. This is true when the $R$-matrix satisfies the famous \emph{Yang-Baxter equation} (\textsc{ybe}) in~$\End(V_a\otimes V_b\otimes V_c)$:
\begin{equation}\label{s4:eq:YBE}
	R_{ab}(u-v)\,R_{ac}(u-w)\,R_{bc}(v-w)=R_{bc}(v-w)\,R_{ac}(u-w)\,R_{ab}(u-v) \ .
\end{equation}
Like the Jacobi equation, this relation is cubic in the `structure constants'. One can check that the solution~\eqref{s4:eq:R-matrix solution} of the six-vertex \textsc{fcr} does indeed satisfy the \textsc{ybe}. In our graphical notation \eqref{s4:eq:YBE} becomes
\begin{equation}\label{s4:eq:YBE graphical}
	\tikz[baseline={([yshift=-.5*11pt*0.8]current bounding box.center)},scale=0.8,font=\scriptsize]{
		\draw[->] (0,2) node[left]{$c$} -- (2,0) -- (3,0);
		\draw[->] (0,1) node[left]{$b$} -- (1,2) -- (2,2) -- (3,1);
		\draw[->] (0,0) node[left]{$a$} -- (1,0) -- (3,2);
	} \qquad = \qquad 
	\tikz[baseline={([yshift=-.5*11pt*0.8]current bounding box.center)},scale=0.8,font=\scriptsize]{
		\draw[->] (0,2) node[left]{$c$} -- (1,2) -- (3,0);
		\draw[->] (0,1) node[left]{$b$} -- (1,0) -- (2,0) -- (3,1);
		\draw[->] (0,0) node[left]{$a$} -- (2,2) -- (3,2);
	} \quad .
\end{equation}
Readers familiar with the braid group may recognize this as an analogue of the braid relation but involving spectral parameters. Of course the lines may again be straightened out like we did in \eqref{s4:eq:FCR graphical alt}.

\paragraph{Summary.} To conclude this subsection we present a brief overview of the formalism that we have set up. The operators of the \textsc{qism}, the equations that they satisfy, and relation between these operators is shown in Table~\ref{s4:tb:qism}. Any physical operator, in $\End(\cH)$, can be expressed as an element of the \textsc{yba}, i.e.\ in terms of the generators $A(u),\To,D(u)$. In particular the \textsc{yba} can be used to construct the Bethe vectors, which is our next topic.

\begin{table}[h]
	\centering
	\tikz[baseline={([yshift=-.5*11pt*0.8]current bounding box.center)},scale=0.8,>=implies]{
		\matrix(m)[matrix of math nodes, text height=1.5ex, text depth=0.25ex, row sep=0em, column sep=2em,ampersand replacement=\&]{
			 \& \text{auxiliary } V_b \& \text{local } V_l \&\text{global } \cH \\ 
			\text{auxiliary } V_a \& R_{ab}(u):\ \textsc{ybe} \& L_{al}(u):\ \textsc{fcr} \& T_a(u):\ RTT \\
			 \& \& \& \phantom{T_a(u)} \\ \& \& \& \phantom{T_a(u)} \\
			\text{physical result} \& \phantom{\text{auxiliary } V_b} \& \&  \begin{tabular}{c} $A(u),\To,D(u):$\ \textsc{yba} \\ $t(u):$\ commute \end{tabular} 
			\\
		};
		\draw ([xshift=-1.5em]m-1-2.north west) -- ([xshift=-1.5em,yshift=-.5em]m-5-2.south west);
		\draw ([yshift=.05cm]m-2-1.north west) -- ([xshift=1.5em,yshift=.05cm]m-2-4.north east);	
		\draw[<->,double equal sign distance] (m-2-3) -- (m-2-4);
		\draw[->,double equal sign distance] (m-2-4) -- (m-4-4.center);	
		}
	\caption{Summary of the \textsc{qism} in the spin-chain language, where $V_a$ plays an auxiliary role.}
	\label{s4:tb:qism}
\end{table}

\subsection{The algebraic Bethe ansatz}\label{s4:ABA}

Our final task is to reproduce the results of the \textsc{cba} from Sections \ref{s2:results} and~\ref{s3:results} in the context of the \textsc{qism}. The goal is to diagonalize the transfer matrix~\eqref{s4:eq:t=A+D}; the spectrum of the \textsc{xxz} Hamiltonian~\eqref{s4:eq:H_1} then follows from the trace identity \eqref{s4:eq:H_1}. We proceed along the lines of Faddeev~\cite{Fad95a}. Although there are still some nontrivial calculations involved in the algebraic Bethe ansatz (\textsc{aba}), it is much easier to get the eigenvalues $\Lambda_M$ and the \textsc{bae} for any $M$-particle sector than it is with the \textsc{cba}.

\paragraph{Second quantization.} The \textsc{cba} from Section~\ref{s2:method} features the Bethe wave function~\eqref{s2:eq:CBA M}: this is the first-quantized approach to the quantum-mechanical spin chains. In contrast, the \textsc{aba} corresponds to second quantization through the explicit construction of a Fock space of states for the model. As a first attempt to construct such a Fock space let us briefly go back to the \textsc{cba}. We already have a good candidate for the Fock vacuum: the pseudovacuum~$\ket{\Omega} \in \cH_0$ from~\eqref{s2:eq:pseudovacuum}. For $M=1$, \eqref{s2:eq:magnon} suggests that $\sum_l \E^{\choice{-}{}\I p_m l } S^-_l$ may serve as a creation operator for a magnon with quasimomentum~$p_m$. Unfortunately already for $M=2$ we see that this cannot be true. Indeed, applying that operator twice on $\ket{\Omega}$ gives $A(p_1,p_2)=A'(p_1,p_2)$ in~\eqref{s2:eq:CBA M=2}, only allowing for trivial two-body scattering. To proceed we have to exploit the Yang-Baxter algebra from the previous subsection.

Again we start from the pseudovacuum, which is depicted in our shorthand as
\begin{equation}\label{s4:eq:pseudovacuum}
	\ket{\Omega} \ = \ \
	\tikz[baseline={([yshift=-.5*11pt*0.8]current bounding box.center)},	scale=0.8,font=\scriptsize,triple/.style={postaction={draw,-,shorten >=.05},double,double distance=4pt,-implies}]{
		\draw[triple,-,dotted] (0,0) -- (0,1);
	} \ \ . 
\end{equation}
In view of \eqref{s4:eq:monodromy acting} and \eqref{s4:eq:ABCD graphical} we have
\begin{equation}\label{s4:eq:A Omega}
	A(u)\,\ket{\Omega} = \sum_{\Vector{\gamma}\in\{\pm 1\}^L}
	\tikz[baseline={([yshift=-.5*11pt*0.8-.5*11pt]current bounding box.center)}, 
	scale=0.8,triple/.style={postaction={draw,-,shorten >=.05},double,double distance=4pt,-implies}]{
		\draw[dotted] (0,1) -- (2,1);
		\draw[triple,-,dotted] (1,0) -- (1,1);
		\draw[triple,-] (1,1) -- (1,2) node[above]{$\boldsymbol\gamma$};
	}
	\ \ket{\Vector{\gamma}} \ = \
	\tikz[baseline={([yshift=-.5*11pt*0.8]current bounding box.center)},	scale=0.8,triple/.style={postaction={draw,-,shorten >=.05},double,double distance=4pt,-implies}]{
		\draw[dotted] (0,1) -- (2,1);
		\draw[triple,-,dotted] (1,0) -- (1,2);
	}
	\ \ket{\Omega} = a(u)^L \ket{\Omega} \ ,
\end{equation}
where in the second equality the sum over outgoing configurations collapses to a single term by line conservation.

\begin{exercise}
Show in the same way that $\ket{\Omega}$ is also an eigenvector of $D(u)$ and $C(u)$,
\begin{equation}\label{s4:eq:CD Omega}
	C(u)\,\ket{\Omega} = 0 \ , \qquad D(u)\,\ket{\Omega} = b(u)^L \, \ket{\Omega} \ , 
\end{equation}
while
\begin{equation}\label{s4:eq:B Omega}
	B(u)\,\ket{\Omega} = \frac{c(u)}{b(u)} \, a(u)^L \sum_{l \in\mathbb{Z}_L} \left(\frac{b(u)}{a(u)}\right)^{\! l} \, \ket{l} \ \in \ \cH_1 \ . 
\end{equation}
\end{exercise}

Thus $B(u)$ and $C(u)$ present themselves as candidates for raising and lowering operators, respectively. For this to make sense $B(u)$ must map $\cH_M$ into $\cH_{M+1}$ for \emph{each} $M$-particle sector, while $C(u)$ should act in the opposite direction. Graphically it is obvious that this is indeed the case: by line conservation $B(u)$ `injects' an excitation (occupancy) into the global quantum space~$\cH$, while $C(u)$ `absorbs' one. We conclude that $B(u)$ and $C(u)$ may indeed be used to build a Fock space starting from $\ket{\Omega}$.

\begin{exercise}
For an alternative argument check that the ice rule $[S^z_a+S^z , T_a(u)]=0$, cf.~\eqref{s4:eq:ice rule Lax}, implies that the generators of the \textsc{yba} satisfy
\begin{align}
 	& [S^z , A(u)] = [S^z , D(u)] = 0 \ , \label{s4:eq:Sz AD} \\
 	& [S^z , B(u)] = - B(u) \ , \qquad [S^z , C(u)] = C(u) \ , \label{s4:eq:Sz BC}
\end{align}
and compare \eqref{s4:eq:Sz BC} with \eqref{s2:eq:su(2) ladder}.
\end{exercise}

For this construction to reproduce the results of the \textsc{cba} we also need a way to include the parameters~$p_m = \choice{}{-} \I \log z_m$. In the present set-up there is already an obvious candidate to fulfil this role: the spectral parameter~$u$. If $B(u)$ is to create a physical state we should in particular be able to match \eqref{s4:eq:B Omega} with the magnon-solution~\eqref{s2:eq:magnon} to reproduce the spectrum for $M=1$. This requires
\begin{equation}\label{s4:eq:u via p}
	z(u) = \E^{\choice{-}{} \I p(u)} = \frac{b(u)}{a(u)} \ ,
\end{equation}
which is consistent with Exercise~\ref{s3:ex:z via a b} in Section~\ref{s3:results}.

\paragraph{Algebraic Bethe ansatz.} According to \eqref{s4:eq:Sz AD} the Hilbert space~$\cH$ splits into $M$-particle sectors as in \eqref{s2:eq:M-part decomposition}, where each $\cH_M$ is preserved by the transfer matrix~\eqref{s4:eq:t=A+D}. Motivated by the preceding discussion, for suitable values of the spectral parameters $\Vector{u}\in\mathbb{C}^M$, let us look for eigenvectors in the $M$-particle sector of the form
\begin{equation}\label{s4:eq:ABA via B}
	\ket{\Psi_M; \Vector{u}} \coloneqq B(u_1) \cdots B(u_M) \, \ket{\Omega} \ \in \cH_M \ .
\end{equation}
This is the \emph{algebraic Bethe ansatz} for the \emph{Bethe vectors} employed to diagonalize the transfer matrix and spin-chain Hamiltonian. The strategy is as follows:
\begin{enumerate}
	\item Use \eqref{s4:eq:t=A+D} and the relations from the Yang-Baxter algebra to work out $t(u_0) \,\ket{\Psi_M ; \Vector{u} }$.
	\item Read off $\Lambda_M(\Vector{z})$ from the \emph{wanted} terms, proportional to $\ket{\Psi_M; \Vector{u}}$ as in the \textsc{aba}.
	\item Demand that the unwanted terms cancel to get the \textsc{bae} for the allowed values of~$\Vector{u}$.
\end{enumerate}
Like for the \textsc{cba}, the ansatz~\eqref{s4:eq:ABA via B} will only work for specific values of~$\Vector{u}$, but unlike before there are no unknown coefficients that have to be determined. Thus, this time all effort goes into Step~1, which can be done using a nice trick based on \eqref{s4:eq:BB}.

\subparagraph{Step 1.} We have to compute the two terms in
\begin{equation}\label{s4:eq:t Psi}
	t(u_0) \, \ket{\Psi_M;\Vector{u}} = A(u_0) \, \prod_{m=1}^M B(u_m) \, \ket{\Omega} + D(u_0) \, \prod_{m=1}^M B(u_m) \, \ket{\Omega} \ .
\end{equation}
We start with the first term on the right-hand side. Using \eqref{s4:eq:AB} we can move $A(u_0)$ past $B(u_1)$:
\begin{equation}\label{s4:eq:A Psi computation}
	A(u_0) \, \prod_{m=1}^M B(u_m) = \left( \frac{a(u_1-u_0)}{b(u_1-u_0)} \, B(u_1) \, A(u_0) - \frac{c(u_1-u_0)}{b(u_1-u_0)} \, B(u_0) \, A(u_1) \right) \prod_{m=2}^M B(u_m) \ .
\end{equation}
Continuing in this way we obtain $2^M$ terms, each proportional to $\big(\prod_{\nu\neq \mu}B(u_\nu)\big) A(u_\mu)$ for some $0\leq \mu \leq M$. As $\ket{\Omega}$ is an eigenvector of $A(u_\mu)$, see \eqref{s4:eq:A Omega}, the result must be of the form
\begin{equation}\label{s4:eq:A Psi}
	A(u_0) \, \ket{\Psi_M ; \Vector{u} } = \sum_{\mu = 0}^{M} M_\mu(u_0,\Vector{u}) \prod_{\substack{\nu = 0 \\ \nu \neq \mu}}^M B(u_\nu) \, \ket{\Omega} \ .
\end{equation}

Two of the coefficients $M_\mu$ are easy to compute. Firstly, only one of the $2^M$ terms contributes to $\mu=0$: this is the term where we always pick up the first term in~\eqref{s4:eq:AB}, giving
\begin{equation}\label{s4:eq:M_0}
	M_0(u_0,\Vector{u}) = a(u_0)^L \, \prod_{m=1}^M \frac{a(u_m-u_0)}{b(u_m-u_0)} \ .
\end{equation}
Secondly, the coefficient for $\mu=1$ also only has one contribution: this comes from the second term on the right-hand side of \eqref{s4:eq:A Psi computation}, where we always pick up the first term in the subsequent steps of \eqref{s4:eq:AB}. Thus we find
\begin{equation}\label{s4:eq:M_1}
	M_1(u_0,\Vector{u}) = - a(u_1)^L \, \frac{c(u_1-u_0)}{b(u_1-u_0)} \, \prod_{n = 2}^M  \frac{a(u_n-u_0)}{b(u_n-u_0)} \ . 
\end{equation}

The other coefficients receive more and more contributions, and their calculation appears to be a complicated task. Luckily there is a neat trick that exploits the \textsc{yba} to obtain the other coefficients without much effort. Indeed, recall that by~\eqref{s4:eq:BB} the $B$'s commute. (We actually already used this in writing an ordinary product in the \textsc{aba}; else we should have specified an ordering.) Therefore we may rearrange the creation operators in~\eqref{s4:eq:A Psi computation} in any way we like; in particular we may put $B(u_m)$ in front. Then, by switching $1$ and~$m$ in \eqref{s4:eq:M_1}, the above argument immediately yields 
\begin{equation}\label{s4:eq:M_m}
	M_m(u_0,\Vector{u}) = - a(u_m)^L \, \frac{c(u_m-u_0)}{b(u_m-u_0)} \, \prod_{\substack{n = 1 \\ n \neq m}}^M  \frac{a(u_n-u_m)}{b(u_n-u_m)} \ . 
\end{equation}

The coefficients $N_\mu(u_0,\Vector{u})$ in
\begin{equation}\label{s4:eq:D Psi}
	D(u_0) \, \ket{\Psi_M ; \Vector{u} } = \sum_{\mu = 0}^{M} N_\mu(u_0,\Vector{u}) \prod_{\substack{\nu = 0 \\ \nu \neq \mu}}^M B(u_\nu) \, \ket{\Omega}
\end{equation}
are computed in a similar way, now using relation~\eqref{s4:eq:DB} from the \textsc{yba} together with \eqref{s4:eq:CD Omega} and of course the trick. The result is
\begin{align}
	& N_0(u_0,\Vector{u}) = b(u_0)^L \, \prod_{m=1}^M \frac{a(u_0-u_m)}{b(u_0-u_m)} \ , \label{s4:eq:N^0} \\
	& N_m(u_0,\Vector{u}) = - b(u_m)^L \, \frac{c(u_0-u_m)}{b(u_0-u_m)} \, \prod_{\substack{n = 1 \\ n \neq m}}^M  \frac{a(u_m-u_n)}{b(u_m-u_n)} \ . \label{s4:eq:N^m}
\end{align}

\subparagraph{Step 2.} Since only the terms with $\mu=0$ in \eqref{s4:eq:A Psi} and \eqref{s4:eq:D Psi} are of the wanted form, the eigenvalues are given by
\begin{equation}\label{s4:eq:t eigenvalue}
	\Lambda_M(u_0;\Vector{u}) = a(u_0)^L \prod_{m=1}^M \frac{a(u_m-u_0)}{b(u_m-u_0)} + b(u_0)^L \prod_{m=1}^M \frac{a(u_0-u_m)}{b(u_0-u_m)} \ .
\end{equation}

\subparagraph{Step 3.} The remaining terms in \eqref{s4:eq:A Psi} and \eqref{s4:eq:D Psi} cancel when $M_m(u_0,\Vector{u})+N_m(u_0,\Vector{u})=0$ for all $1\leq m \leq M$, that is, when
\begin{equation}\label{s4:eq:BAE full}
	\left(\frac{b(u_m)}{a(u_m)}\right)^{\!\!L} 
	= -\frac{c(u_m-u_0)}{b(u_m-u_0)} \, \frac{b(u_0-u_m)}{c(u_0-u_m)} \prod_{\substack{n = 1 \\ n \neq m}}^M  \frac{a(u_n-u_m)}{b(u_n-u_m)} \, \frac{b(u_m-u_n)}{b(u_m-u_n)} \ , \qquad 1 \leq m \leq M \ .
\end{equation}

\paragraph{Results.} Notice that \eqref{s4:eq:t eigenvalue}--\eqref{s4:eq:BAE full} have the same form as $\Lambda_M$ and the \textsc{bae} found in Section~\ref{s3:results}. Moreover, the left-hand side of \eqref{s4:eq:BAE full} matches with that in \eqref{s2:eq:xxz BAE M} and in \eqref{s3:eq:6v BAE M} when \eqref{s4:eq:u via p} holds. To make contact with the results obtained through the \textsc{cba} we use the parametrization~\eqref{s4:eq:parametrization}.

\begin{exercise}
Using some trigonometric identities, show that \eqref{s4:eq:t eigenvalue} precisely matches with \eqref{s3:eq:6v eigenvalues} if we recognize $a(u_0) = a$, $b(u_0) = b$, $c(u_0) = c$ and $b(u_m)/a(u_m) = z_m$.
\end{exercise}

\begin{exercise}
Check that \eqref{s4:eq:BAE full} now reduces to
\begin{equation}\label{s4:eq:BAE}
	\left(\frac{b(u_m)}{a(u_m)}\right)^{\!\!L} = (-1)^{M-1} \prod_{\substack{n=1 \\ n\neq m}}^M \frac{a(u_n-u_m)}{a(u_m-u_n)} \ ,
\end{equation}
and that this correctly reproduces \eqref{s2:eq:xxz BAE via lambda} when the spectral parameters are identified with rapidities via
\begin{equation}\label{s4:eq:u via lambda}
	u_m = -(\lambda_m + \I\gamma/2) \ .
\end{equation}
\end{exercise}

Now let us use the trace identities to compute the momentum and energy of the Bethe vectors. Notice that the second term in~\eqref{s4:eq:t eigenvalue}, and almost all of its derivatives, vanish at~$u_* = 0$. By \eqref{s4:eq:H_0} the momentum of the Bethe vector~\eqref{s4:eq:ABA via B} is
\begin{equation}
	p(\Vector{u}) = \choice{\frac{1}{\I}}{\I}  \log \frac{\Lambda_M(u_*,\Vector{u})}{a(u_*)^L} = \choice{\frac{1}{\I}}{\I}  \sum_{m=1}^M \log\frac{a(u_m)}{b(u_m)} = \sum_{m=1}^M p(u_m) \ ,
\end{equation}
nicely generalizing \eqref{s4:eq:u via p} to the $M$-particle sector.
\begin{exercise}
Use \eqref{s4:eq:H_1} to check that the energy of $\ket{\Psi_M;\Vector{u}}$ is given by
\begin{equation}
\begin{aligned}
	& \varepsilon_M(\Vector{u}) = \frac{\I \sin\gamma}{2} \, \Lambda_M(u_*,\Vector{u})^{-1} \, \left. \frac{\partial}{\partial u_0}\right|_{u_0=u_*} \Lambda_M(u_0,\Vector{u}) = \sum_{m=1}^M \varepsilon_1(u_m) \ , \\
	& \qquad \varepsilon_1(u) = \frac{\I \sin\gamma}{2} \, \frac{b(u)}{a(u)} \left(\frac{a(u)}{b(u)}\right)' = \frac{\sin\gamma}{2} p'(u) \ ,
\end{aligned}
\end{equation}
and plug in \eqref{s4:eq:u via lambda} to check that this agrees with \eqref{s2:eq:xxz E via lambda}.
\end{exercise}

In the framework of the \textsc{qism} it did not require much effort to derive these results even for an arbitrary $M$-particle sector. A comparison with the amount of work needed to obtain the same results using the \textsc{cba} in Appendix~\ref{sM} goes a long way to justify the abstract algebraic machinery developed in the previous subsections!

\paragraph{Rational limit.} Let us briefly turn to the isotropic limit~$\Delta\to 1$. Notice that the parametrization~\eqref{s3:eq:parametrization} yields $a=b$ and $c=0$ as $\gamma\to0$. Thus the Lax operator~\eqref{s4:eq:Lax operator matrix} reduces to the trivial operator $\rho\sinh(u) \id$ in this limit. This is directly related to the issue pointed out in Exercise~\ref{s2:ex:isotropic limit} at the end of Section~\ref{s2:results}. To study the isotropic limit one can use the parametrization obtained from \eqref{s3:eq:parametrization} by rescaling $u= \gamma \, u'$ and $\rho = 1/\gamma$ before taking $\gamma \to 0$. Dropping the primes we find that the result is rational in~$u$, and in our case even linear:
\begin{equation}\label{s4:eq:parametrization rational}
	a(u) = u+\I \ , \qquad b(u) = u \ , \qquad c(u) = \I \ .
\end{equation}
The Lax operator~\eqref{s4:eq:Lax operator via sigmas} thus becomes a simple linear combination of the identity operator and the permutation operator~\eqref{s4:eq:perm}:
\begin{equation}\label{s4:eq:Lax operator via sigmas rational}
	L_{al}(u) =  u \id_{al} + \, \I \perm_{al} \ .
\end{equation}
The $R$-matrix~\eqref{s4:eq:R-matrix solution} acquires the same form in the isotropic limit.

From the \textsc{xxx} \textsc{yba} it can be shown that on-shell Bethe vectors are always highest weight: they are annihilated by the total spin-raising operator, $S^+ \ket{\Psi_M;\Vector{u}} = 0$, by virtue of the \textsc{bae}; see e.g.~\cite[\textsection4]{Fad95a}. The $\mathfrak{su}(2)$-descendants in the spectrum are obtained by applying $S^-$ to the Bethe vectors. The \textsc{xxx} spin chain is analyzed using the \textsc{qism} in \cite[\textsection3]{FT84}.
\begin{exercise}
Check that for $R_{ab}$ like in \eqref{s4:eq:Lax operator via sigmas rational} the $RTT$-relations \eqref{s4:eq:RTT alt} can be written as
\begin{equation}\label{s4:eq:Yangian}
	(u-v) \, [T_{ij}(u),T_{kl}(v)] = \I \, \big( T_{kj}(v) T_{il}(u) - T_{kj}(u) T_{il}(v) \big)
\end{equation}
where $T_{11}(u)=A(u)$, $T_{12}(u)=B(u)$, $T_{21}(u)=C(u)$, $T_{22}(u)=D(u)$.
\end{exercise}

\paragraph{More spin chains.} To conclude this section we show how the \textsc{qism} allows one to define new quantum-integrable spin chains. In Section~\ref{s2:spin chains} we looked at spin chains whose interactions are
\begin{enumerate}[label=\roman*),noitemsep]
	\item only nearest neighbour;
	\item homogeneous (translationally invariant); and
	\item at least \emph{partially isotropic}. 
\end{enumerate}
The \textsc{xxx} and \textsc{xxz} magnets are the main examples of such models. Let us briefly recall where properties~(i)--(iii) were used in the analysis of these spin chains. Property~(iii) was necessary to define the $M$-particle sectors, forming the starting point for both the \textsc{cba} and the \textsc{aba}. For the diagonalization of the Hamiltonian in the one-particle sector property~(ii) came in handy, directly yielding the magnons~\eqref{s2:eq:magnon}. Property~(i) was also important for the \textsc{cba}, leading to the Bethe wave function~\eqref{s2:eq:CBA M}.

In the present section we have seen how, starting from the Lax operator~\eqref{s4:eq:Lax operator} for the \textsc{xxz} model, via the monodromy matrix~\eqref{s4:eq:monodromy} one obtains the Yang-Baxter algebra that allows one to solve the model via the \textsc{aba}. The latter still crucially depends on (iii), but it is possible to relax the other two properties in such a way that we can still use the \textsc{aba} to solve the resulting models.

\subparagraph{Twisted boundaries.} The periodic boundary conditions can be modified (`deformed') to allow for \emph{quasi-periodic} or \emph{twisted} boundary conditions, $\Vector{S}_{l+L}= \exp\bigl(\tfrac{1}{2} \I \vartheta \, \sigma^z\bigr) \, \Vector{S}_l \, \exp\bigl(-\tfrac{1}{2} \I \vartheta \, \sigma^z\bigr)$, where the \emph{twist parameter} $\vartheta$ is $2\pi$-periodic, cf.~$\exp(\pm \pi \I \, \sigma^z)=-\id$. Such boundary conditions are accounted for in the \textsc{qism} by introducing a \emph{twist operator}
\begin{equation}
	K_a(\vartheta) \coloneqq \exp(\tfrac{1}{2}\I \vartheta \, \sigma^z_a) = \mathrm{diag}(\E^{\I \,\vartheta/2},\E^{-\I \,\vartheta/2}) = \ 
	\tikz[baseline={([yshift=-.5*11pt*0.8+1pt]current bounding box.center)},	scale=0.8,font=\scriptsize]{
		\draw[->] (0,0) node[left]{$a$} -- (2,0); 
		\fill [black] (1,0) circle (3pt);
	}\ \in \End(V_a) \ .
\end{equation}
Although this operator breaks the full isotropy group when one starts with the \textsc{xxx} spin chain, the partial isotropy subgroup $U(1)_z \subseteq SU(2)$ corresponding to the ice rule, cf.~\eqref{s4:eq:ice rule Lax}, is preserved:
\begin{equation}\label{s4:eq:ice rule twisted}
	[K_a(\vartheta) K_b(\vartheta),R_{ab}(u)]= [\exp\tfrac{1}{2} \I \vartheta (\sigma^z_a+\sigma^z_b) , R_{ab}(u)] = 0 \ .
\end{equation}
This implies that the twisted monodromy matrix
\begin{equation}
	T_a(u;\vartheta) \coloneqq K_a(\vartheta) \ordprod_{l \in \mathbb{Z}_L} L_{al}(u) = \
	\tikz[baseline={([yshift=-.5*11pt*0.8+8pt]current bounding box.center)}, 
	scale=0.8,font=\scriptsize,triple/.style={postaction={draw,-,shorten >=.05},double,double distance=4pt,-implies}]{
		\draw[triple] (1,0) node[below]{$1\cdots L$} -- (1,2);
		\draw[->] (0,1) node[left]{$a$} -- (3,1); \fill [black] (2,1) circle (3pt);
	} \ \in \End(V_a\otimes \cH)
\end{equation}
satisfies the $RTT$-relation for the same $R$-matrix~\eqref{s4:eq:R-matrix solution}, so one can use the \textsc{aba} to diagonalize the twisted transfer matrix $t(u;\vartheta)=\tr_a T_a(u;\vartheta) = \E^{\I \,\vartheta/2} A(u) + \E^{-\I \,\vartheta/2} D(u)$.

\begin{exercise}
Extend the results from Sections \ref{s4:conserved quantities}, \ref{s4:YBA} and~\ref{s4:ABA} to the case of twisted boundary conditions: compute $H_0$ and $H_1$, use \eqref{s4:eq:ice rule twisted} to verify \eqref{s4:eq:RTT graphical}, check whether the relations \eqref{s4:eq:BB} and \eqref{s4:eq:AB}--\eqref{s4:eq:DB} of the \textsc{yba} are modified, and compute the eigenvalues and \textsc{bae} for the Bethe vectors~\eqref{s4:eq:ABA via B}.
\end{exercise}

\subparagraph{Inhomogeneities.} Translational invariance can be broken by considering Lax operators $L_{al}(u;\mu_l)\coloneqq L_{al}(u-\mu_l)$ that depend on \textit{inhomogeneity parameters} $\mu_l \in \mathbb{C}$. Integrability is preserved since the shifted arguments do not affect the \textsc{fcr}~\eqref{s4:eq:FCR w=u-v} in an essential way, and the same $R$-matrix~\eqref{s4:eq:R-matrix solution} does the job. As the shifts generically differ from site to site, however, this time there is no value~$u_*$ of the spectral parameter at which all Lax operators become proportional to the permutation operator as in~\eqref{s4:eq:Lax operator at u*}, and the $H_k$ cannot be expressed in a nice way; in particular the Hamiltonian $H_1$ does no longer involve only nearest-neighbour interactions. Nevertheless, one can still define the monodromy matrix as
\begin{equation}
	T_{a}(u;\Vector{\mu}) \coloneqq \ordprod_{l \in \mathbb{Z}_L} L_{al}(u-\mu_l) \ ,
\end{equation}
and proceed as before to diagonalize $t(u;\Vector{\mu})=\tr_a T_a(u;\Vector{\mu}) = A(u;\Vector{\mu}) + D(u;\Vector{\mu})$.

\begin{exercise}
Extend the results from Section~\ref{s4:YBA} and~\ref{s4:ABA} to the inhomogeneous \textsc{xxz} spin chain: check if the relevant relations of the \textsc{yba} are altered, find the vacuum eigenvalues \eqref{s4:eq:A Omega} and \eqref{s4:eq:CD Omega}, and compute the eigenvalues and \textsc{bae} for the Bethe vectors~\eqref{s4:eq:ABA via B}.
\end{exercise}

\subparagraph{Further generalizations.} Other quantum-integrable spin chains that can be tackled using the \textsc{qism} include models with higher spins ($V_l \cong \mathbb{C}^{2s+1}$, see e.g.~\cite[\textsection2.4, 3.4]{GRS96} or \cite[\textsection10]{Fad95a}), open spin chains with reflecting boundaries (see e.g.~\cite[\textsection3.5]{GRS96} or \cite{Skl88}), local spins that vary from site to site (obtained by `fusion'), and even exotic `spin' where $\mathfrak{su}(2)$ is replaced by any simple Lie (super)algebra \cite{KRS81}.

\section{Relation to theoretical high-energy physics}\label{s5}

In Section~\ref{s2}, \ref{s3} and~\ref{s4} we have dealt with quantum integrability in the context of quantum and statistical mechanics. In this section we turn to \textsc{qft}. The following (possibly biased and certainly incomplete) overview of applications of quantum integrability to \textsc{qft} may serve as a motivation for studying quantum integrability for the more `hep-th'-oriented reader:
\begin{itemize}
\item Exact $S$-matrix theory in 2d \textsc{qft} is governed by a Yang-Baxter equation (\textsc{ybe}) for the two-body $S$-matrix (Zamolodchikov-Zamolodchikov, 1979~\cite{ZZ79}). See Section~\ref{s5:fact scatt} and the nice lecture notes by Dorey~\cite{Dor98b}.
\item Also in two dimensions, conformal field theories (with central charge $c<1$) possess a quantum-integrable structure (Bazhanov-Lukyanov-Zamolodchikov, 1990s~\cite{BLZ96,*BLZ97,*BLZ99}).
\item The high-energy scattering of four-dimensional Yang-Mills theories, such as \textsc{qcd}, appears to exhibit hidden symmetries governed by quantum-integrable spin chains (Faddeev-Korchemsky, 1995~\cite{Kor12,*FK95}).
\end{itemize}
Moving up a notch and adding supersymmetry:
\begin{itemize}
\item The gauge/\textsc{ybe} correspondence interprets Seiberg duality for 4d $\mathcal{N}=1$ quiver gauge theories as a Yang-Baxter relation for vertex models on those quivers (Yamazaki, 2013~\cite{Yam13});
\item The Bethe/gauge correspondence relates supersymmetric vacua of certain $\cN=2$ gauge theories to quantum-integrable models (Nekrasov-Shatashvili, 2009 \cite{NSh09a,*NSh09b,NSh10,NRSh11,NSh14}). This topic is introduced in Section~\ref{s5:Bethe/gauge}.
\item The AdS/\textsc{cft} correspondence has led to what is currently the largest and most active research area relating quantum integrability to high-energy physics. The case that is understood best is based on the following two ingredients:
	\begin{itemize}
	\item[i)] It is well known that	the IIB superstring in the $AdS_5\times S^5$ background is classically integrable. The one-loop corrected two-body $S$-matrix satisfies a \textsc{ybe}, providing strong evidence that the theory is integrable at the quantum level too.
	\item[ii)] Correlators of gauge-invariant operators in planar 4d $\mathcal{N}=4$ Yang-Mills theory can be computed using spin chains. 
	\end{itemize}
Crucially, these two appear to yield exactly the \emph{same} quantum-integrable structure, and there is a quantum-integrable model interpolating between the two sides, as is indicated by several nontrivial tests. See the big review \cite{BA+10}, and \cite{AF09} for a detailed account of the string-theory side.
\end{itemize}
Let us also mention few more `math-ph'-oriented connections between quantum integrability and gauge theories:
\begin{itemize}
\item Certain deformations of 2d Yang-Mills theory are related to exactly solvable statistical-physical models (Migdal, 1975; Rusakov, 1990; Witten, 1991). See the review \cite{CMR95} and references therein.
\item Chern-Simons theory and other 3d topological \textsc{qft}s have quantum-group symmetries (Witten, 1989; Reshetikhin-Turaev, early 90s). See the review \cite{Fre95} and references therein.
\item Twisted deformed 4d $\cN=1$ gauge theories are in a similar way related to quantum-integrable models (Costello, 2013~\cite{Cos13b,*Cos13a}).
\end{itemize}

Thus, quantum integrability also appears to have close ties to theoretical high-energy and mathematical physics. To get some feeling for how quantum integrability may appear in such contexts we turn to two examples: one old --- the theory of exact $S$-matrices and factorized scattering in two-dimensional \textsc{qft}, which also sheds more light on the results from Sections \ref{s2:results} and~\ref{s3:results} --- and one much more recent: the Bethe/gauge correspondence.

\subsection{Quantum integrability and 2d \textsc{qft}}\label{s5:fact scatt}

Physics in two dimensions is special. This is well-known for conformal field theory, but more generally applies to two-dimensional \textsc{qft}. For example, the spin-statistics connection no longer holds in lower dimensions, in accordance with the results for magnons on a spin chain in Section~\ref{s2:results}. In condensed-matter physics this also shows up in two \emph{spatial} dimensions: the rotation group $SO(2)$ is abelian, so spin is not restricted to take half-integer values in the quantum theory, and at the same time there is a whole range of possible (`anyonic' or `braid') statistics. The situation for statistics in one spatial dimension is even more peculiar, as is illustrated by bosonization.$^\#$\footnote{For bosonization and other techniques in many-body quantum physics in one spatial dimension we refer to the book by Giamarchi~\cite{Gia04}. In particular, bosonization can also be applied to the \textsc{xxz} spin chain, see \cite[\textsection6.1]{Gia04}.} In the relativistic setting of high-energy physics in spacetime dimension two, the little group $SO(1) \subseteq SO(1,1)$ is trivial, so (Lorentz) spin in does not even have an intrinsic meaning in $1+1$ dimensions.

Another notable feature of two-dimensional physics that is more relevant for us is, of course, the presence of quantum-integrable models. However, there is in fact  \emph{no} generally accepted definition of `quantum integrability', so let us pause for a moment to think what this could actually mean.

Motivated by the definition of Liouville integrability in classical mechanics (see below) one often would like to ask for a maximal set of commuting conserved quantities, such as the $H_k$ in Section~\ref{s3:results}. However, for quantum-mechanical models with a finite-dimensional Hilbert space such a family always exists: like for any hermitean operator, the eigenstates~$\ket{\Psi_k}$ of the Hamiltonian can be taken to be orthogonal; then $H_k \coloneqq \ket{\Psi_k}\otimes \bra{\Psi_k}$ does the job. In addition there is an issue with the number of independent operators, and hence with the notion of `maximality', see~\cite{Wei92,*CM11}.

In \emph{practice}, then, one often demands the existence of an underlying $R$-matrix satisfying the \textsc{ybe}, see Section~\ref{s4:YBA}. This equation is closely related to some of the main results from Sections \ref{s2:results} and~\ref{s3:results}:
\begin{enumerate}[label=\roman*),noitemsep]
\item the \textsc{xxz} and six-vertex models have a tower of commuting symmetries;
\item the number $M$ of magnons/occupancies is conserved in scattering processes; 
\item such scattering is two-body reducible: it factorizes into two-body processes.
\end{enumerate}
Indeed, the discussion in Section~\ref{s3:results} and~\ref{s4:YBA} shows how the \textsc{ybe} is related to (i). In this subsection we explain how the \textsc{ybe} may appear in \textsc{qft}s with many conserved quantities in two dimensions, and how results (ii) and (iii) fit in.

\paragraph{Classical integrability.} Let us first summarize the situation in classical mechanics, where there is a well-defined notion of integrability. In the Hamiltonian framework classical mechanical models are described in terms of a finite-dimensional phase space, say with coordinates $q_i$ (positions) and $p_i$ (momenta), and Poisson brackets~$\{\,\cdot\,,\,\cdot\,\}$ between (functions of) those coordinates. The time evolution of any observable $f(\Vector{p},\Vector{q})$ is determined by the Hamiltonian~$H(\Vector{p},\Vector{q})$ through Hamilton's equation $\dot f \coloneqq \partial_t f = \{H,f\}$. A familiar example is a collection of point particles moving on a line in a potential $U(\Vector{q})$, for which $H(\Vector{p},\Vector{q})=\sum_i p_i^2/2m + U(\Vector{q})$ and the equations $\{p_i,q_j\}=\delta_{ij}$ yield $\dot q_i = p_i/m$ and $\dot p_i = -\partial_i U(\Vector{q})$, giving Newton's $m \ddot q_i = -\partial_i U(\Vector{q})$.

Now, in brief, if a classical mechanical system has a maximal set of independent conserved quantities, and if those quantities Poisson-commute with each other, then the system is solvable. Liouville's theorem makes this statement precise and provides a recipe for finding the solution. Systems that can be solved in this way are called \emph{Liouville integrable} and include the harmonic oscillator, Kepler problem and several spinning tops, see \cite[\textsection2]{BBT03} or the classic~\cite[\textsection49--50]{Arn97}. 

The next step is the notion of a \emph{Lax pair}, which opens up the road to the beautiful theory of classical integrable models. This theory contains the (semi)classical version of concepts that play an important role in quantum integrability, notably the \emph{classical $r$-matrix} and \emph{classical Yang-Baxter equation} as well as \emph{spectral parameters}. Unfortunately this takes us too far from our topic, but see e.g.~\cite{BBT03}. We suffice by mentioning that this story also has an infinite-dimensional incarnation in the theory of integrable nonlinear \textsc{pde}s in two dimensions, like the Korteweg-De Vries (KdV) equation, for which there exists an infinite set of conserved quantities that can be obtained from a generating function called the \emph{monodromy matrix}. These systems have solitonic solutions that can be obtained via the so-called \emph{classical inverse-scattering method}, whose quantization is the \textsc{qism} from Section~\ref{s4}, see also \cite[\textsection{V}]{KBI93}.

The main message to take away is that in classical mechanics the existence of sufficiently many conserved quantities signals its integrability. What happens if there exist many conserved quantities in a field theory?

\paragraph{Symmetries of the $S$-matrix.} One of the main goals in quantum field theory is to compute scattering amplitudes. All such amplitudes are contained in the \emph{$S$-matrix}, which relates asymptotic incoming states to all possible asymptotic final states: schematically $S\ket{i}=\sum_f S^f_i\ket{f}$, where the $S$-matrix entries~$\smash{S^f_i}$ depend on parameters such as the momenta of the asymptotic states. The $S$-matrix is a very complicated object and ordinarily one can at best hope to calculate its entries perturbatively. The presence of symmetries (and thus conserved quantities), however, may restrict the $S$-matrix to such an extent that its entries can be computed \emph{exactly}.
 
In many field theories the momentum $P_\mu$ is conserved. In addition there may be `internal' symmetry operators, which commute with the Lorentz algebra~$\mathfrak{o}(d-1,1)$, such as flavour or gauge symmetries. These operators are also symmetries of the $S$-matrix. Can the $S$-matrix have symmetries that are not Lorentz scalars or vectors? In three or more spacetime dimensions the \emph{Coleman-Mandula theorem} says the answer is negative~\cite{CM67,*PH97}. More precisely: if a Lorentz-invariant theory has a mass gap and its $S$-matrix is analytic, under some technical assumptions the presence of any `forbidden' infinitesimal (bosonic) symmetry operator forces the $S$-matrix of the theory to be trivial. Clearly such theories, having no interactions, are not very interesting.
 
In spacetime dimension two, though, the Coleman-Mandula theorem no longer holds, and there \emph{are} interesting two-dimensional theories, such as the sine-Gordon model, exhibiting (infinitely many!) conserved charges transforming in higher representations of the Lorentz algebra~$\mathfrak{o}(1,1)\cong\mathbb{R}$. The existence of such higher symmetries in a field theory does, however, severely constrain the $S$-matrix of the theory. There are two important restrictions; let us take a look at each of them.

\paragraph{Elastic scattering.} The first restriction arises in the following way. A local operator~$Q_s$ of `Lorentz spin'~$s$ acts on asymptotic states by multiplication with a polynomial, homogeneous of degree $s$, $\sum_a c_{a,s}\, (p_a^\pm)^s$ in the lightcone momenta ($p^\pm = p^0 \pm p^1$) of the asymptotic particles~$a$. (For $s=1$ the coefficients~$c_{a,1}$ all equal one and this is just the total momentum.) Conservation of $Q_s$ thus implies
\begin{equation}
	\sum_{i\in\text{in}} c_{i,s} \, (p_i^\pm)^s = \sum_{f\in\text{out}} c_{f,s} \,  (p'^\pm_f)^s \ .
\end{equation}
If, for $N\to N'$ scattering, the number of $Q_s$ with different~$s$ is large enough we get an overdetermined system of equations in the $N+N'$ asymptotic momenta, and the solutions are trivial up to relabelling: the sets of initial and final momenta must coincide, $\{p_i\mid i\in\text{in}\}=\{p'_f\mid f\in\text{out}\}$.

Now if there exist infinitely many symmetries~$Q_s$ then it follows that there is no macroscopic particle creation or annihilation in such theories: $N=N'$. (Microscopically there may be contributions due to virtual processes where the number of particles are not conserved.) Having infinitely many conserved charges thus constrains the $S$-matrix to be block-diagonal, with a block for each number~$N$ of asymptotic particles. In addition it follows that all scattering is \emph{elastic}: the energy of each asymptotic particle is conserved in the process.

\paragraph{Factorized scattering.} The second remarkable consequence of the presence of higher symmetries~$Q_s$ in a two-dimensional field theory is that  all blocks with $N\geq 3$ turn out to be determined by $\binom{N}{2}$ two-body scattering processes. This phenomenon, known as \emph{factorized scattering}, turns the computation of the full $S$-matrix into a finite task, bringing exact results within reach. It can be understood as a result of the fact that higher symmetries act on particles by momentum-dependent translations; let us sketch the argument.$^\#$\footnote{As for the Coleman-Mandula theorem (see Witten's \cite[Lect.~4]{Wit99}): in three or more dimensions we can use a `forbidden' symmetry to shift any two intersecting particle trajectories by momentum-dependent (and thus different) amounts to obtain two non-intersecting lines, so that there is no scattering. This argument fails in two dimensions since there two lines generically \emph{do} intersect.}

It is convenient to parametrize momentum by the \emph{rapidity}~$\lambda$ via $p^\pm = m \E^{\pm \lambda}$, which takes into account the mass-shell condition $\Vector{p}^2 = p^+ p^- = m^2$. In a spacetime diagram (with time increasing upwards) the worldline of a free particle with momentum~$\Vector{p}$ is a straight line with slope~$\lambda$. Lorentz invariance implies that the $S$-matrix entries only depend on rapidity differences. The two-body $S$-matrix for $i_1 i_2\to f_1 f_2$ corresponds to
\begin{equation}\label{s5:eq:two-body S-mat}
	S^{f_1 f_2}_{i_1 i_2}(\lambda_1-\lambda_2) \ = 
	\tikz[baseline={([yshift=-.5*11pt*0.8]current bounding box.center)},>=latex,scale=0.8,font=\small]{
	\pgfmathsetmacro\cscA{1/sin(130)}
	\pgfmathsetmacro\cscB{1/sin(60)}
		\draw[->] (-130:\cscA) node[below]{$i_1$} -- (50:\cscA) node[above]{$f_2$};
		\draw[->] (-60:\cscB) node[below]{$i_2$} -- (120:\cscB) node[above]{$f_1$};
		\draw[thin,densely dotted] (-130:.4) arc(-130:-60:.4); \draw (-95:.65) node{$\lambda_{12}$};
		\fill [black] (0,0) circle (.5ex);
	} , \qquad\qquad \lambda_{12} \coloneqq \lambda_1 - \lambda_2  \ . 
\end{equation}
To see that this quantity completely determines the full $S$-matrix of the theory let us consider three-body scattering with corresponding $S$-matrix element
\begin{equation}\label{s5:eq:three-body S-mat}
	S^{f_1 f_2 f_3}_{i_1 i_2 i_3}(\lambda_1 - \lambda_2,\lambda_2 - \lambda_3) \ = 
	\tikz[baseline={([yshift=-.5*11pt*0.8]current bounding box.center)},>=latex,scale=0.8,font=\small]{
	\pgfmathsetmacro\cscA{1/sin(130)}
	\pgfmathsetmacro\cscB{1/sin(92)}
	\pgfmathsetmacro\cscC{1/sin(60)}
		\draw[->] (-130:\cscA) node[below]{$i_1$} -- (50:\cscA) node[above]{$f_3$};
		\draw[->] (-92:\cscB) node[below]{$i_2$} -- (88:\cscB) node[above]{$f_2$};
		\draw[->] (-60:\cscC) node[below]{$i_3$} -- (120:\cscC) node[above]{$f_1$};
		\fill [gray!30] (0,0) circle (.5cm);
	} \ .
\end{equation}
In a local field theory this scattering must happen in one of following ways, depending on the initial positions:
\begin{equation}\label{s5:eq:factorized scattering}
	\begin{tikzpicture}[baseline={([yshift=-.5*11pt*0.8]current bounding box.center)},>=latex,scale=0.8,font=\small]
	\pgfmathsetmacro\cscA{1/sin(130)}
	\pgfmathsetmacro\cscB{1/sin(92)}
	\pgfmathsetmacro\cscC{1/sin(60)}
		\draw[->] (-130:1.2*\cscA) coordinate(i1) node[below]{$i_1$} -- (50:1.2*\cscA) coordinate(f3) node[above]{$f_3$};
		\draw[->] ($(-92:1.2*\cscB)+(.3,0)$) coordinate(i2) node[below]{$i_2$} -- ($(88:1.2*\cscB)+(.3,0)$) coordinate(f2) node[above]{$f_2$};
		\draw[->] (-60:1.2*\cscC) coordinate(i3) node[below right]{$i_3$} -- (120:1.2*\cscC) coordinate(f1) node[above]{$f_1$};
		\fill [black] (intersection of i1--f3 and i2--f2) circle (.5ex);
		\fill [black] (intersection of i1--f3 and i3--f1) circle (.5ex);
		\fill [black] (intersection of i2--f2 and i3--f1) circle (.5ex);
	\end{tikzpicture} 
	\ \ , \qquad
	\begin{tikzpicture}[baseline={([yshift=-.5*11pt*0.8]current bounding box.center)},>=latex,scale=0.8,font=\small]
	\pgfmathsetmacro\cscA{1/sin(130)}
	\pgfmathsetmacro\cscB{1/sin(92)}
	\pgfmathsetmacro\cscC{1/sin(60)}
		\draw[->] (-130:1.2*\cscA) node[below]{$i_1$} -- (50:1.2*\cscA) node[above]{$f_3$};
		\draw[->] (-92:1.2*\cscB) node[below]{$i_2$} -- (88:1.2*\cscB) node[above]{$f_2$};
		\draw[->] (-60:1.2*\cscC) node[below]{$i_3$} -- (120:1.2*\cscC) node[above]{$f_1$};
		\fill [black] (0,0) circle (.5ex);
	\end{tikzpicture}
	\ \ , \qquad
	\begin{tikzpicture}[baseline={([yshift=-.5*11pt*0.8]current bounding box.center)},>=latex,scale=0.8,font=\small]
	\pgfmathsetmacro\cscA{1/sin(130)}
	\pgfmathsetmacro\cscB{1/sin(92)}
	\pgfmathsetmacro\cscC{1/sin(60)}
		\draw[->] (-130:1.2*\cscA) coordinate(i1) node[below]{$i_1$} -- (50:1.2*\cscA) coordinate(f3) node[above]{$f_3$};
		\draw[->] ($(-92:1.2*\cscB)-(.3,0)$) coordinate(i2) node[below]{$i_2$} -- ($(88:1.2*\cscB)-(.3,0)$) coordinate(f2) node[above]{$f_2$};
		\draw[->] (-60:1.2*\cscC) coordinate(i3) node[below]{$i_3$} -- (120:1.2*\cscC) coordinate(f1) node[above left]{$f_1$};
		\fill [black] (intersection of i1--f3 and i2--f2) circle (.5ex);
		\fill [black] (intersection of i1--f3 and i3--f1) circle (.5ex);
		\fill [black] (intersection of i2--f2 and i3--f1) circle (.5ex);
	\end{tikzpicture} \ \ .
\end{equation}

Now the higher symmetries at our disposal enable us to shift the worldlines by an amount that is momentum dependent and thus different for each of the lines. Thus we can pick suitable $Q_s$'s to turn each situation from \eqref{s5:eq:factorized scattering} into either of the others. The conclusion is that all three situations must represent the same physical process, so that \eqref{s5:eq:three-body S-mat} indeed factorizes into two-body processes. In addition we obtain a consistency condition for the entries~$S^{f_1 f_2}_{i_1 i_2}$ of the two-body $S$-matrix:
\begin{equation}
	\sum_{j_1,\,j_2,\,j_3} \
	\begin{tikzpicture}[baseline={([yshift=-.5*11pt*0.8]current bounding box.center)},>=latex,scale=0.8,font=\small]
	\pgfmathsetmacro\cscA{1/sin(130)}
	\pgfmathsetmacro\cscB{1/sin(92)}
	\pgfmathsetmacro\cscC{1/sin(60)}
		\draw[->] (-130:1.5*\cscA) coordinate(i1) node[below]{$i_1$} -- (50:1.5*\cscA) coordinate(f3) node[above]{$f_3$} node[pos=.53,above,inner sep=2ex]{$j_2$};
		\draw[->] ($(-92:1.5*\cscB)+(.5,0)$) coordinate(i2) node[below]{$i_2$} -- ($(88:1.5*\cscB)+(.5,0)$) coordinate(f2) node[above]{$f_2$}  node[pos=.45,right]{$j_3$};
		\draw[->] (-60:1.5*\cscC) coordinate(i3) node[below right]{$i_3$} -- (120:1.5*\cscC) coordinate(f1) node[above]{$f_1$} node[pos=.45,below,inner sep=2ex]{$j_1$};
		\fill [black] (intersection of i1--f3 and i2--f2) circle (.5ex);
		\fill [black] (intersection of i1--f3 and i3--f1) circle (.5ex);
		\fill [black] (intersection of i2--f2 and i3--f1) circle (.5ex);
	\end{tikzpicture}
	\ \ = \ \sum_{j_1,\,j_2,\,j_3} \
	\begin{tikzpicture}[baseline={([yshift=-.5*11pt*0.8]current bounding box.center)},>=latex,scale=0.8,font=\small]
	\pgfmathsetmacro\cscA{1/sin(130)}
	\pgfmathsetmacro\cscB{1/sin(92)}
	\pgfmathsetmacro\cscC{1/sin(60)}
		\draw[->] (-130:1.5*\cscA) coordinate(i1) node[below]{$i_1$} -- (50:1.5*\cscA) coordinate(f3) node[above]{$f_3$} node[pos=.5,below,inner sep=2ex]{$j_2$};
		\draw[->] ($(-92:1.5*\cscB)-(.5,0)$) coordinate(i2) node[below]{$i_2$} -- ($(88:1.5*\cscB)-(.5,0)$) coordinate(f2) node[above]{$f_2$} node[pos=.53,left]{$j_1$};
		\draw[->] (-60:1.5*\cscC) coordinate(i3) node[below]{$i_3$} -- (120:1.5*\cscC) coordinate(f1) node[above left]{$f_1$} node[pos=.72,right,inner sep=1ex]{$j_3$};
		\fill [black] (intersection of i1--f3 and i2--f2) circle (.5ex);
		\fill [black] (intersection of i1--f3 and i3--f1) circle (.5ex);
		\fill [black] (intersection of i2--f2 and i3--f1) circle (.5ex);
	\end{tikzpicture}
\end{equation}
This is the \textsc{ybe} in the context of factorized scattering. Further guaranteeing the consistency of the factorization of \emph{any} $N$-body scattering processes, the \textsc{ybe} is of central importance to the theory of factorized scattering and exact $S$-matrices in two dimensions. For more about these topics we refer to the nice lecture notes by Dorey~\cite{Dor98b} and to \cite[\textsection1.1]{GRS96}.

\subsection{The Bethe/gauge correspondence}\label{s5:Bethe/gauge}

Over the last two decades it is becoming apparent that gauge theories enjoying $\cN=2$ supersymmetry seem to come with an integrable structure for free. One class of such supersymmetric gauge theories was studied intensively by Seiberg and Witten~\cite{SW94} in the mid-1990s. It was soon realized~\cite{GK+95,*DW96} that at low energies these theories give rise to so-called classical algebraic integrable systems, closely related to (Liouville) integrable models from \emph{classical} mechanics. These systems can be solved exactly, at least in principle, and such a low-energy classical integrable structure is a surprising and pleasant feature of the original gauge theories.

An interesting question is whether this story also has a \emph{quantum-mechanical} analogue: do there exist gauge theories which yield quantum-integrable models at low energies? In~2009, Nekrasov and Shatashvili showed that this question can be answered positively for a large class of supersymmetric gauge theories: this is the \emph{Bethe/gauge correspondence}~\cite{NSh09a,*NSh09b,NSh10,NRSh11,NPSh13,NSh14}.

This subsection provides a qualitative overview of the main idea presented in~\cite{NSh09a,*NSh09b}. Rather than discussing the correspondence in full generality we describe what it entails for its main example. Much more can be found in e.g.\ the following references. Nonperturbative \textsc{qft} and \textsc{susy} are introduced in the book by Shifman~\cite{Shi12}. A classic reference for supersymmetry is Wess and Bagger~\cite{WB92}. For two-dimensional $\cN=(2,2)$ gauge theories see Witten~\cite{Wit93} and the book by Hori et al.~\cite{Hor03}; see also the author's MSc thesis~\cite[\textsection3]{Lam12} and the references therein. Other, closely related, developments concerning exact results in $\cN=2$ gauge theories are reviewed in the recent series of papers by Teschner et al.~\cite{Tes14}. The Bethe/gauge correspondence is introduced in~\cite[\textsection4]{Lam12}. For a mathematical version of the Bethe/gauge correspondence see~\cite{MO12}.

\paragraph{Rough version.} We already know a lot about the `Bethe' side of the correspondence. The `gauge' side of the story is about a very different area of theoretical physics, namely that of \emph{supersymmetric gauge theory}. Although these arise naturally in string theory let us give another, more concrete, motivation coming from \textsc{qcd}. At high energies, the gauge coupling is very small, and the standard tools of perturbative quantum field theory are available: this is the asymptotically free regime. As we flow to the infrared, however, the coupling constant of the non-abelian gauge theory ceases to be small and perturbation theory breaks down. At the same time, the vacuum structure of a gauge theory determines the possible phases of that theory. Getting a grip on this low-energy regime is one of the great open problems in present-day theoretical physics. Out of the various approaches to try and overcome this problem we consider the following. Instead of studying \textsc{qcd} itself we shift our attention to its \emph{idealizations} possessing supersymmetry (\textsc{susy}). This leads to a class of toy models that allow for more control whilst at the same time keeping several key features of \textsc{qcd}, providing an arena to test ideas about quantum field theory and non-abelian gauge theory in a controlled setting. Supersymmetry provides us with exact tools, so that a better insight can be gained into the structure of these idealized models. One may hope that some of this insight persists into the real world, where it may ultimately shed more light on \textsc{qcd} itself. 

Roughly speaking the Bethe/gauge correspondence amounts to the observation that
\begin{quote}
	There exists a class of \textsc{susy} gauge theories for which the (`gauge-theoretic part' of the) low-energy effective theory has the structure of a quantum-integrable model.
\end{quote}
Of course, this statement has to be supplemented with the specification \emph{which} \textsc{susy} gauge theories this applies to. A more precise version of the statement is the following:
\begin{quote}
	For \textsc{susy} gauge theories with effective two-dimensional `$\cN=(2,2)$' super-Poincar\'e invariance at low energies, the (`Coulomb branch' of the) \textsc{susy} vacuum structure corresponds to a quantum-integrable model.
\end{quote}
Let us illustrate these statements via the main example presented in~\cite{NSh09a,*NSh09b}, which features the \textsc{xxx} spin chain with Hamiltonian \eqref{s2:eq:Ham xxx} on the `Bethe' side. Before getting to the actual correspondence we take a look the `gauge' side and describe the specific theories featuring in this example.

\paragraph{Gauge-theory set-up.} We begin with an ordinary non-abelian gauge theory with matter, rather like \textsc{qcd}, in $3+1$~dimensions. The field content of the theory we want to study is as follows. The gauge group is $G=U(N_\text{c})$, so the gauge field~$A_\mu$ describes $N_\text{c}$~colours of gluons. Next there are massive matter fields in the fundamental and antifundamental representations of the gauge group: these are the quarks and their antiparticles. Equally many of these fields are included, so that we have $N_\text{f}$ flavours of quarks and antiquarks. This is similar to the situation in the Standard Model in the absence of the Yukawa couplings, where there is a flavour symmetry mixing the different generations of fermions. Finally we add one further massive field, living in the adjoint representation of~$U(N_\text{c})$, these are like the $W$-bosons in the Standard Model.

The next step is to enhance the theory by making it \emph{supersymmetric}. Although \textsc{susy} is a beautiful topic we suffice by saying that it is a boson-fermion symmetry, so each field now has a \emph{superpartner} with the opposite statistics. The fields and their superpartners are nicely packaged together in representations of the \textsc{susy} algebra known as  \emph{supermultiplets}. For example, $A_\mu$ and its fermionic superpartner are contained in a `vector supermultiplet', and each of the matter fields and their bosonic superpartners `chiral supermultiplets', also in the (anti)fundamental or adjoint representation of the gauge group.

So far the set-up is quite standard for \textsc{susy} gauge theories. Now we proceed towards the more specific situation that we need for our example: the number of spacetime dimensions has to be reduced to two. This can be arranged via a procedure known as \emph{dimensional reduction} where we consider theories that are translationally invariant in two directions, so that we can restrict ourselves to physics in a $(1+1)$-dimensional slice of spacetime, say the $(t,z)$-plane. At the level of fields this reduction is achieved by simply forgetting the dependence of the fields on the coordinates $x$ and~$y$. Although the result may seem pathological, non-abelian gauge theories in two dimensions still exhibit interesting phenomena such as asymptotic freedom, confinement, dimensional transmutation, and topological effects such as solitons. Thus, two-dimensional models provide a playground to learn about such aspects in an easier setting.

The dimensional reduction has two consequences that are relevant for us here. Firstly, the original gauge theory the vector field~$A_\mu$ has four components, describing two physical degrees of freedom corresponding to the transverse polarizations. As in Kaluza-Klein reduction, upon going down to two dimensions these components recombine into a two-dimensional gauge field with components $A_0$ and~$A_1$ together with a complex scalar field~$\sigma$. Since there are no transverse directions in the reduced spacetime (there are no photons in two dimensions!), $A_0$ and~$A_1$ do not correspond to physical degrees of freedom; therefore, the field~$\sigma$ encodes the physics of the gauge field from the four-dimensional viewpoint.

The second consequence of the reduction is that the amount of \textsc{susy} is enhanced, and we get $\cN=2$ \emph{extended} \textsc{susy}. This is a powerful tool for computations, bringing exact methods within our reach. In brief the reason for this enhancement is the following: \textsc{susy} operators are spinorial quantities, and the four real components that a Majorana spinor has in four dimensions recombine into two Majorana spinors worth of \textsc{susy} operators in two dimensions. In fact, in two dimensions, the Majorana (reality) and Weyl (chirality) conditions for spinors are compatible and can be simultaneously imposed, so that the minimal spinors in $1+1$ dimension  have only one component. In terms of the chirality of the \textsc{susy} operators we now have two left-handed and two right-handed such components: this is called `$\cN=(2,2)$' \textsc{susy} in two dimensions.

\paragraph{Goal: finding \textsc{susy} vacua.} Even in lower dimensions and with extended \textsc{susy}, the full non-abelian gauge theory is still too complicated to solve. However, what can be computed \emph{exactly} is the effective theory in the infrared. When the energy is sufficiently low, the (massive) matter is effectively non dynamical, so the resulting effective theory is a \emph{pure} gauge theory: it does no longer involve any matter fields. Our goal is to find the \textsc{susy} vacuum structure on the \emph{Coulomb branch}, which is parametrized by (the vacuum expectation value of) the complex scalar field~$\sigma$ encoding the four-dimensional physical degrees of freedom of the gauge field.

The low-energy theory is governed by a scalar potential whose zeroes are the \emph{supersymmetric vacua}. Amongst others this requires $\sigma$, which lives in the adjoint representation of~$U(N_\text{c})$, to be diagonalizable. As a result the gauge group breaks down to its diagonal subgroup~$U(1)^{N_\text{c}}$, so that the low-energy effective theory constitutes $N_\text{c}$~copies of electrodynamics; this is where the name `Coulomb branch' comes from. Write $\sigma_m$ for the $m$th diagonal component of (the vacuum expectation value of)~$\sigma$, corresponding to the $m$th $U(1)$-factor. The goal is to determine the values of the \textsc{susy} vacua~$\sigma_m$, also known as `Coulomb moduli', through the \emph{vacuum equations} for $\Vector{\sigma}=(\sigma_1,\To,\sigma_{N_\text{c}})$:
\begin{equation}
	\exp \left( 2\pi \frac{\partial \tilde{W}_\text{eff}(\Vector{\sigma})}{\partial \sigma_m} \right) = 1 , \qquad \qquad 1 \leq m \leq N_\text{c} \ ,
\end{equation}
where $\tilde{W}_\text{eff}$ is known as the (shifted) effective \emph{twisted superpotential}. Thus, the supersymmetric vacua on the Coulomb branch are similar to critical points of $\tilde{W}_\text{eff}$; the exponential is a peculiarity of $\cN=(2,2)$ theories in two dimensions.

\paragraph{Approach: Wilsonian effective action.} To find the effective theory in the infrared all matter fields have to be integrated out, as well as the higher modes of the gauge fields, to find $\tilde{W}_\text{eff}$. This is where $\cN=2$ \textsc{susy} comes in handy: by general arguments, it allows one to derive certain (`non-renormalization' and `decoupling') theorems that highly restrict what can happen when fields are integrated out; in particular, $\tilde{W}_\text{eff}$ is one-loop exact. With the help of such theorems, the low-energy effective theory on the Coulomb branch can be computed exactly.

\paragraph{Result: vacuum equations.} When the dust has settled the vacuum equations for our theory turn out to be as follows:
\begin{equation}\label{s5:eq:vacuum equations}
	\biggl( \frac{\sigma_m+\I/2}{\sigma_m-\I/2} \biggr)^{\!\!N_\text{f}} = \ \prod_{\substack{n=1 \\ n\neq m}}^{N_\text{c}} \frac{\sigma_m-\sigma_n+\I}{\sigma_m-\sigma_n-\I} \ , \qquad 1 \leq m \leq N_\text{c} \ .
\end{equation}
On the left-hand side each factor in the numerator comes from a one-loop diagram due a single flavour of quarks; likewise, the denominator is due to the anti-quarks. The product on the right-hand side is the contribution from the off-diagonal components of the gauge field.

\paragraph{Guiding observation; dictionary.} Recalling the discussion at the end of Section~\ref{s2:results} it is clear that the vacuum equations \eqref{s5:eq:vacuum equations} are precisely the same as the \textsc{bae}~\eqref{s2:eq:xxx BAE via lambda} for the $M$-particle sector of the \textsc{xxx} spin chain. This is the guiding observation behind the Bethe/gauge correspondence. Comparing \eqref{s5:eq:vacuum equations} with \eqref{s2:eq:xxx BAE via lambda} we arrive at the identifications listed in Table~\ref{s5:tb:table}. Thus, the Bethe/gauge correspondence provides a \emph{dictionary} which allows us to go back and forth between the two sides of the correspondence. In this way, we may use our knowledge of one side to learn something about, or at least shed new light on, the other side.

\begin{table}[h]
	\centering
	\begin{tabular}{c|c}
		\textbf{Bethe} \qquad & \textbf{Gauge} \\ \hline \strut
		number of sites $L$ & number of flavours $N_\text{f}$ \\ 
		number of magnons $M$ & number of colours $N_\text{c}$ \\
		rapidities $\lambda_m$ & supersymmetric vacua $\sigma_m$ \\
		Yang-Yang function $Y$ & effective twisted superpotential $\tilde{W}_\text{eff}$
	\end{tabular}
	\caption{Dictionary relating quantities for the \textsc{xxx} spin chain and the \textsc{susy} vacuum structure on the Coulomb branch for a $\cN=(2,2)$ gauge theory in two dimensions. (The Yang-Yang function is defined in Appendix~\ref{sY}; more precisely it matches with $2\pi\, \tilde{W}_\text{eff}$ up to a shift.)}
	\label{s5:tb:table}
\end{table}

\paragraph{More examples.} Let us step back for a moment. So far we have looked at \emph{one} particular \textsc{susy} gauge theory and found that its vacuum structure on the Coulomb branch corresponds to the $M$-particle sector of \textsc{xxx} spin chain. Is this just a coincidence? After all, one swallow does not make a summer.

In the first two papers~\cite{NSh09a,*NSh09b}, Nekrasov and Shatashvili showed that the Bethe/gauge correspondence can accommodate for \emph{much} more, and the dictionary can accordingly be enriched. The various integrable spin chains discussed at the end of Section~\ref{s4:ABA} all fit in nicely: quasiperiodic boundary conditions for the spin chain correspond to the inclusion of topological (`Fayet-Iliopoulos' and `vacuum-angle') terms in the gauge theory; inhomogeneities and local spins straightforwardly correspond to different values of the (`twisted') mass parameters, implicit in the above, on the `gauge' side; anisotropy is related to gauge theories in three or four dimensions, and more exotic types of `spin' complement gauge groups other than~$U(N_\text{c})$.

Further examples of the correspondence, where the quantum-integrable models involve long-range (rather than just nearest-neighbour) interactions, were provided in \cite{NSh10}. Gauge theories with more involved Standard Model-like gauge groups fit in too~\cite{NPSh13}, as do gauge theories on curved spaces~\cite{NSh14}. The general pattern of the Bethe/gauge correspondence is the following:
\begin{quote}
	Consider a two-dimensional gauge theory, with
	\begin{itemize}[noitemsep]
		\item effectively, at low energies, two-dimensional $\cN=(2,2)$ super-Poincar\'e invariance
		\item the appropriate matter content, and
		\item suitable values for the parameters (implicit in the above),
	\end{itemize}
	and determine the low-energy effective theory. Then the vacuum equations for the Coulomb branch coincide with the \textsc{bae} of a quantum-integrable model.
\end{quote}
Arguably this correspondence can be extended to encompass \emph{all} quantum-integrable models that are solvable via a Bethe ansatz.

\paragraph{Achievements.} Let us conclude with some words about the \emph{use} of the Bethe/gauge correspondence. The ideas that we have described so far are of course nice and perhaps unexpected: they relate two seemingly very different areas of theoretical and mathematical physics. However it should be stressed that to the extent described so far this correspondence is not predictive. Indeed, in order to establish it one has to start with the appropriate models on both sides and calculate the \textsc{bae} and the vacuum equations in order to set up a dictionary between the two sides. Thus, a lot of information is needed as input, and one may wonder to what novel results it may lead. Here are a few examples of its early achievements: 
\begin{itemize}
	\item In \cite{NSh10} the Bethe/gauge correspondence and related ideas were used to find so-called thermodynamic Bethe-ansatz (\textsc{tba}) equations for several long-range quantum-integrable models. These results later confirmed for the Toda chain with integrability techniques~\cite{KT11}.
	\item Many new (quantum) integrable models have been obtained from $ADE$-quiver gauge theories in four dimensions (in the $\Omega$-background)~\cite{NP12}.
	\item Dualities between supersymmetric gauge theories have been used obtain relations between different quantum-integrable models \cite{GK13,*CHK13}; some of these are already known in the integrability literature but others appear to be novel.
\end{itemize}
Thus, in short, the Bethe/gauge correspondence provides a way to translate knowledge about one side into statements about the other side, leading to new insights.

\appendix

\section{Completeness and the Yang-Yang function}\label{sY}

In Sections \ref{s2}, \ref{s3} and~\ref{s4} we have seen how, for a finite system with periodic boundary conditions, the coordinate and algebraic Bethe Ans\"atze lead to the \textsc{bae}. For the eigenvectors in $\cH_M\subseteq\cH$ these consist of $M$ coupled equations determining the allowed values of the parameters in the Bethe ansatz. An important question is whether the Bethe ansatz is \emph{complete}: does it give all $\binom{L}{M}$ independent eigenstates in $\cH_M$? Note that the problem is more complicated than `just' trying to count the number of solutions to the \textsc{bae}: not all such solutions necessarily lead to physically acceptable states; in particular, the resulting Bethe vectors must be \emph{normalizable}, i.e.\ have nonzero norm. The three possible situations are illustrated in Figure~\ref{sY:fg:completeness}. 

The issue of completeness was already investigated for the \textsc{xxx} model by Bethe~\cite[\textsection8]{Bet31}, but it has remained a topic of debate to the present day, with approaches ranging from heavy numerics to combinatorics and algebraic geometry, see e.g.\ \cite{KS14,*Bax01} and references therein. To understand why, notice that symmetries lead to degeneracies, while counting becomes harder by the presence of degeneracies in the spectrum. Indeed, since the Hamiltonian is hermitean, its eigenvectors are orthogonal for different eigenvalues. When there are no degeneracies in the spectrum, one can count the number of distinct eigenvalues to find the number of eigenvectors. 

From this point of view it is not surprising that the situation is better for the case with only partial isotropy. In this brief appendix we introduce an important tool for the study of completeness for the \textsc{xxz} model: the Yang-Yang function.

\begin{figure}[h]
	\centering
	\begin{minipage}{5cm}
		\centering
		\begin{tikzpicture}[scale=0.8,font=\small]
			\draw (0,0) circle (1.6cm);
			\fill[gray!30] (0,0) circle (1cm);
			\draw (0,0) circle (1cm) node {\textsc{bae}};
			\node[right] at (30:1.7) {spec$(H)$};
		\end{tikzpicture}
	\end{minipage}
	\begin{minipage}{5cm}
		\centering
		\begin{tikzpicture}[scale=0.8,font=\small]
			\begin{scope} 
				\clip (-1.2,0) circle (1.6cm);
				\fill[gray!30] (1.2,0) circle (1.6cm);
			\end{scope} 
			\draw (-1.2,0) circle (1.6cm);
			\draw (1.2,0) circle (1.6cm);
			\node at (-1.4,0) {spec$(H)$}; 
			\node at (1.4,0) {\textsc{bae}};
		\end{tikzpicture}
	\end{minipage}
	\begin{minipage}{5cm}
		\centering
		\begin{tikzpicture}[scale=0.8,font=\small]
			\fill[gray!30] (0,0) circle (1cm);
			\draw (0,0) circle (1.6cm);
			\draw (0,0) circle (1cm) node {spec$(H)$};
			\node[right] at (30:1.7) {\textsc{bae}};
		\end{tikzpicture}
	\end{minipage}
	\caption{The three typical cases in the completeness problem of the Bethe ansatz. The left shows the \emph{in}complete case, where the \textsc{bae} do not have enough solutions to obtain the full spectrum of the Hamiltonian. The right, instead, corresponds to the \emph{over}complete case, where all eigenvectors are of the Bethe form but the \textsc{bae} have more than $2^L$ solutions and it has to be determined which of those are physically acceptable. The middle depicts the intermediate case.}
	\label{sY:fg:completeness}
\end{figure}
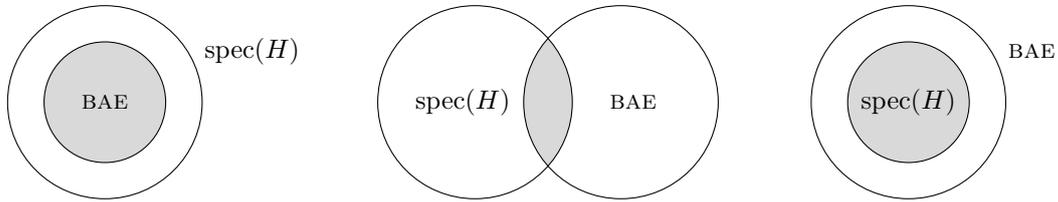

\paragraph{\textsc{xxz} case.} Consider the Bethe ansatz parametrized in terms of rapidities~$\Vector{\lambda}\,\in\mathbb{C}^M$ like at the end of Section~\ref{s2:results}; see also \eqref{s4:eq:u via lambda}. In logarithmic form the \textsc{bae} for the $M$-particle sector read
\begin{equation}\label{sY:eq:xxz BAE via lambda log form}
	L \, p(\lambda_m) = 2\pi I_m + \sum_{\substack{n=1 \\ n\neq m}}^M \Theta(\lambda_n,\lambda_m) \ , \qquad\qquad 1\leq m\leq M \ ,
\end{equation}
where $\Vector{I}=(I_1,\To,I_M)$ are the Bethe quantum numbers, cf.~\eqref{s2:eq:BAE M=2 log form}, and the two-body scattering phase~$\Theta$ depends on the rapidity difference and the anisotropy~$\Delta=\cos\gamma$, cf.~Exercise~\ref{s2:ex:arctan relation} in Section~\ref{s2:results}.

It may come as a surprise that the \textsc{bae} admit a `potential': there exists a function $Y\colon \mathbb{C}^M \longrightarrow \mathbb{C}$ such that \eqref{sY:eq:xxz BAE via lambda log form} is equivalent to the extremality conditions
\begin{equation}\label{sY:eq:YY potential}
	\frac{\partial Y(\Vector{\lambda})}{\partial \lambda_m} = 0 \ , \qquad\qquad 1 \leq m \leq M \ .
\end{equation}
This \emph{Yang-Yang function} or \emph{Yang-Yang action} can be defined as
\begin{equation}\label{sY:eq:xxz YY}
	Y(\Vector{\lambda}) \coloneqq L \sum_{m=1}^M \hat{p}(\lambda_m) - 2\pi\sum_{m=1}^M I_m \lambda_m + \frac{1}{2} \sum_{m,n=1}^M \hat{\Theta}(\lambda_n-\lambda_m) \ ,
\end{equation}
where
\begin{equation}\label{sY:eq:hats}
\begin{aligned}
	& \hat{p}(\lambda) = \int^\lambda \varphi_{1/2}(\mu) \, \D \mu \ , \qquad \hat{\Theta}(\lambda) = \int^\lambda \varphi_1(\mu) \, \D \mu \ , \\
	& \qquad \varphi_s(\lambda) \coloneqq \choice{\I}{\frac 1\I} \log \frac{\sinh(\lambda+\I\gamma s)}{\sinh(\lambda-\I\gamma s)} \ .
\end{aligned}
\end{equation}
\begin{exercise}
Check that the solutions of the \textsc{bae}~\eqref{sY:eq:xxz BAE via lambda log form} are precisely the critical points of \eqref{sY:eq:xxz YY}.
\end{exercise}

Yang and Yang employed $Y(\Vector{\lambda})$ to show that the \textsc{bae} admit a class of real solutions~\cite{YY66}.

\begin{thm}\label{sY:th:thm}
When $0<\Delta<1$ the \textsc{bae}~\eqref{sY:eq:xxz BAE via lambda log form} have a unique real solution $\Vector{\lambda}\in\mathbb{R}^M$ for any $\Vector{I}$.
\end{thm}

\begin{proof}
The \emph{uniqueness} is proven with a convexity argument, see also \cite[\textsection{II.1}]{KBI93}, which only works when $\Vector{\lambda}\,$ is real. Using
\begin{equation}
	\varphi_s'(\lambda) = \choice{\I}{\frac 1\I} \, \frac{-\I \sin(\gamma s)}{\sinh(\lambda+\I\gamma s)\sinh(\lambda-\I\gamma s)} = \choice{}{-} \frac{ 2 \, \sin(\gamma s)}{\cosh(2 \lambda)- \cos(2 \gamma s)}
\end{equation}
it is easy to see that the Hessian matrix~$\partial_m\partial_n Y$ is \choice{positive}{negative} when $0<\gamma<\pi/2$: then
\begin{equation}
	\sum_{m,n=1}^M \frac{\partial^2 Y(\Vector{\lambda})}{\partial \lambda_m \partial \lambda_n} \,v^m \,v^n = \sum_{m=1}^M  L \, \varphi_{1/2}'(\lambda_m) \, v_m^2 + \frac{1}{2} \sum_{m,n=1}^M \varphi_1'(\lambda_m -\lambda_n) \, (v_m-v_n)^2 \choice{>}{<} 0 
\end{equation}
for any $\Vector{\lambda}, \Vector{v}\in\mathbb{R}^M$. This implies that $Y$ has at most one critical point, which, if it exists, is a global \choice{minimum}{maximum}. The \emph{existence} of this critical point, however, requires some care: one has to show that the critical point does not run away to infinity, as happens e.g.\ for the convex function $y(\lambda)=\choice{}{-}\E^\lambda$ on $\mathbb{R}$. The proof can be found in \cite[\textsection4]{YY66}.
\end{proof}

Now by construction the Bethe ansatz is symmetric in the $\lambda_m$, so solutions to the \textsc{bae} that only differ by a permutation (relabelling) of the $\lambda_m$ correspond to the same Bethe vector; such solutions should only be counted once. Since the \textsc{bae} \eqref{sY:eq:xxz BAE via lambda log form} are symmetric under simultaneous interchange of $\lambda_m\leftrightarrow \lambda_n$ and $I_m\leftrightarrow I_n$, this amounts to taking into account only Bethe quantum numbers~$\Vector{I}$ with $0\leq I_1 \leq \cdots \leq I_M\leq L-1$. Moreover, by the Pauli exclusion principle, Bethe vectors vanish whenever two rapidities coincide, so in fact it suffices to consider $0\leq I_1 < \cdots < I_M \leq L-1$. There are precisely $\binom{L}{M}$ such choices for $\Vector{I}$, and by the theorem there is a unique \emph{real} solution~$\Vector{\lambda}\in\mathbb{R}^M$ for each of these. Since the corresponding energies~\eqref{s2:eq:xxz E via lambda} are different, this goes a long way towards a proof of the completeness for the regime $0<\Delta<1$ of the \textsc{xxz} spin chain.

A few years later Yang and Yang applied the same technique to analyze the (simpler) one-dimensional Bose gas~\cite{YY69}; see also \cite[\textsection{I.2}]{KBI93}. It is worth mentioning that $Y(\Vector{\lambda})$ also features in \emph{Gaudin's hypothesis}, which says that the (square) norm of the Bethe wave function can be computed as the determinant of the Hessian matrix $\partial_m\partial_n Y$, see e.g.~\cite[\textsection{X}]{KBI93}.

\paragraph{\textsc{xxx} case.} In the isotropic case the functions $\varphi_s$ in the Yang-Yang function have a simple primitive (cf.~Exercise~\ref{s2:ex:isotropic limit} at the end of Section~\ref{s2:results}), and the integrals in \eqref{sY:eq:hats} can be performed explicitly:
\begin{equation}
	\int^\lambda \varphi_s|_{\Delta=1}(\mu) \, \D \mu  = \choice{\I}{\frac 1\I} \big[ (\lambda+\I s) \log(\lambda+\I s) - (\lambda-\I s) \log(\lambda-\I s) \big] + \text{const} \ ,
\end{equation}
where the constant on the right-hand side does not affect the \textsc{bae}~\eqref{sY:eq:YY potential}. However, the Yang-Yang function is no longer convex for $\Delta=1$, and Theorem~\ref{sY:th:thm} does not apply. However, when the degeneracies of the \textsc{xxx} spin chain are lifted by turning on `generic' quasiperiodic boundary conditions or inhomogeneities (see the end of Section~\ref{s4:ABA}), completeness of the Bethe ansatz \emph{can} be proven~\cite{MTV07}.

\section{Computations for the $M$-particle sector}\label{sM}

In this appendix the \textsc{cba} from Section~\ref{s2:method} is worked out for the $M$-particle sector of the \textsc{xxz} spin chain to derive the results quoted in Section~\ref{s2:results}. Recall that the \textsc{cba}
\begin{equation}\label{sM:eq:CBA M}
	\Psi_{\Vector{p}}(\Vector{l}) = \sum_{\pi\in S_M} A_\pi(\Vector{p}) \E^{\choice{-}{}\I\Vector{p}_\pi \!\cdot \Vector{l}} \ , \qquad\qquad l_1 < \cdots < l_M \ 
\end{equation}
for the wave functions in the $M$-particle sector,
\begin{equation}\label{sM:eq:M vector via coord basis}
	\ket{\Psi_M;p} = \sum_{1\leq l_1<\cdots<l_M \leq L} \!\! \Psi_p(l_1,\To,l_M) \,  \ket{l_1,\To,l_M} \in \cH_M \ ,
\end{equation}
yields eigenstates of the Hamiltonian if the equations
\begin{equation}\label{sM:eq:xxz eqns}
	\bra{\Vector{l}}H_\textsc{xxz}\ket{\Psi_M;\Vector{p}} = E_M(\Vector{p}) \, \Psi_{\Vector{p}}(\Vector{l}) \ , \qquad 1 \leq l_1 < \cdots < l_M \leq L \ .
\end{equation}
can be solved. Our goal is to find the energies~$E_M(\Vector{p})$, the coefficients~$A_\pi(\Vector{p})$ and equations determining the values of the parameters~$\Vector{p}$. Let us write $N_{\Vector{l}}\in\mathbb{N}$ for the number of pairs of neighbouring excitations in the configuration~$\Vector{l}$ of excited spins, i.e.\ $N_{\Vector{l}} \coloneqq \# \{ l_n \mid l_{n+1} = l_n + 1 \}$, so $0\leq N_{\Vector{l}} \leq M-1$ (unless $M=L=N_{\Vector{l}}$). For example, Figure~\ref{s2:fg:spin chain} from Section~\ref{s2:spin chains} has $N_{\Vector{l}}=1$. In terms of this notation the strategy (cf.\ Section~\ref{s2:method}) is as follows:
\begin{enumerate}
	\item[0.] Compute the left-hand side of \eqref{sM:eq:xxz eqns}.
	\item Solve~\eqref{sM:eq:xxz eqns} for the energy contribution $\varepsilon_M(\Vector{p})\coloneqq E_M(\Vector{p})-E_0$ as a function of~$\Vector{p}$ by considering configurations~$\Vector{l}$ with well-separated excitations ($N_{\Vector{l}}=0$).
	\item Solve~\eqref{sM:eq:xxz eqns} for the $A_\pi(\Vector{p})/A_e(\Vector{p})$ as functions of~$\Vector{p}$ by considering $\Vector{l}$ with $N_{\Vector{l}}\geq 1$. (Luckily it will turn out that it suffices to consider a single pair of neighbouring excitations.)
	\item Impose periodic boundary conditions to get the \textsc{bae} for the allowed values of~$\Vector{p}$.
\end{enumerate}
As in Section~\ref{s2:results} we put $\hbar=J=1$. At the end of this appendix we comment on the six-vertex case and give the strategy to derive the results given in Section~\ref{s3:results}.

\paragraph{Warm-up: $M=2$.} It is instructive to work out the strategy for the two-particle sector before tackling the general case.

\subparagraph{Step 0.} Expand the vector $\ket{\Psi_2;\Vector{p}}$ via the coordinate basis of~$\cH_2$ as in~\eqref{sM:eq:M vector via coord basis}:
\begin{equation}\label{sM:eq:M=2 vector via coord basis}
	\ket{\Psi_2;\Vector{p}} = \sum_{1\leq l_1<l_2 \leq L} \Psi_{\Vector{p}}(l_1,l_2) \, \ket{l_1,l_2} \in\cH_2 \ .
\end{equation}
We compute $H_\textsc{xxz}\ket{l_1,l_2}$ using \eqref{s2:eq:su(2) ladder} and~\eqref{s2:eq:Ham xxz via S^pm}. The first two terms in \eqref{s2:eq:Ham xxz via S^pm} give $\sum_k S^\pm_k S^\mp_{k+1} \ket{l_1,l_2} = \ket{l_1\pm 1,l_2} + \ket{l_1,l_2\mp 1}$, where one of the two vectors on the right-hand side vanishes when the excitations are next to each other. The third term, $\sum_k S^z_k S^z_{k+1}$, multiplies $\ket{l_1,l_2}$ by $(L/4-2)$ if $l_1$ and $l_2$ are well separated, and by $(L/4-1)$ if they are neighbours.

\subparagraph{Step 1.} In the well-separated case, $N_{\Vector{l}}=0$, \eqref{sM:eq:xxz eqns} gives
\begin{equation}\label{sM:eq:xxz eqns well separated M=2}
\begin{aligned}
	2\,\varepsilon_2(\Vector{p})\,\Psi_{\Vector{p}}(l_1,l_2) = 4\,\Delta\,\Psi_{\Vector{p}}(l_1,l_2) & -\Psi_{\Vector{p}}(l_1-1,l_2) -\Psi_{\Vector{p}}(l_1+1,l_2) \\ & -\Psi_{\Vector{p}}(l_1,l_2-1) -\Psi_{\Vector{p}}(l_1,l_2+1) \ .
\end{aligned}
\end{equation}
These difference equations have to be satisfied by the wave function for all pairs $l_1 < l_2 - 1$. Plugging in the \textsc{cba}
\begin{equation}\label{sM:eq:CBA M=2}
	\Psi_{p_1,\,p_2}(l_1,l_2) = A_e(p_1,p_2)\, \E^{\choice{-}{}\I (p_1 l_1 + p_2 l_2)} + A_\tau (p_1,p_2)\, \E^{\choice{-}{}\I  (p_1 l_2 + p_2 l_1)} \ , \qquad\qquad l_1 < l_2 \ ,
\end{equation}
immediately yields the result 
\begin{equation}\label{sM:eq:xxz energy M=2}
	\varepsilon_2(p_1,p_2) = 2\,\Delta-\cos p_1 -\cos p_2 = \varepsilon_1(p_1) + \varepsilon_1(p_2)  \ .
\end{equation}

\begin{equation}\label{sM:eq:xxz eqns neighbouring M=2}
	2\,\varepsilon_2(\Vector{p})\,\Psi_{\Vector{p}}(l,l+1) = 2\,\Delta\,\Psi_{\Vector{p}}(l,l+1) - \Psi_{\Vector{p}}(l-1,l+1) - \Psi_{\Vector{p}}(l,l+2) \ .
\end{equation}
These equations can be solved using a trick exploiting the similarity between \eqref{sM:eq:xxz eqns neighbouring M=2} and~\eqref{sM:eq:xxz eqns well separated M=2}. 
\begin{exercise}
Check that \eqref{sM:eq:xxz eqns well separated M=2} is satisfied by $\Psi_{\Vector{p}}$ from \eqref{sM:eq:CBA M=2} and $\varepsilon_2(\Vector{p})$ from \eqref{s2:eq:xxz energy M=2} \emph{independently} of the values of $l_1$ and $l_2$.
\end{exercise}
In particular, given \eqref{sM:eq:CBA M=2} and \eqref{sM:eq:xxz energy M=2}, \eqref{sM:eq:xxz eqns well separated M=2} holds true \emph{even} when $N_{\Vector{l}}=1$. When $l_2=l_1+1$, however, the right-hand side of \eqref{sM:eq:xxz eqns well separated M=2} features wave functions with equal arguments, which are not defined. Although such $\Psi_{\Vector{p}}(l,l)$ are not physical --- they do not enter \eqref{sM:eq:M=2 vector via coord basis} --- the trick is to \emph{extend} the Bethe wave function to the diagonal $l_1 = l_2$ using the formula~\eqref{sM:eq:CBA M=2}. In this way nothing changes for $N_{\Vector{l}}=0$, while \eqref{sM:eq:xxz eqns well separated M=2} with $l_2=l_1+1$ gives extra information that we can use to solve~\eqref{sM:eq:xxz eqns neighbouring M=2}.

To see if the so-extended \textsc{cba} does indeed satisfy \eqref{sM:eq:xxz eqns neighbouring M=2} we subtract \eqref{sM:eq:xxz eqns well separated M=2} with~$l_2=l_1+1$ from \eqref{sM:eq:xxz eqns neighbouring M=2} to get
\begin{equation}\label{sM:eq:xxz eqns new M=2}
	2\,\Delta \, \Psi_{\Vector{p}}(l,l+1) = \Psi_{\Vector{p}}(l,l) + \Psi_{\Vector{p}}(l+1,l+1) \ .
\end{equation}
Up to an overall normalization the coefficients in the \textsc{cba}~\eqref{sM:eq:CBA M=2} are determined by these equations: \eqref{sM:eq:xxz eqns new M=2} is satisfied if the two-body $S$-matrix~\eqref{s2:eq:S-matrix} is the $1\times1$ matrix given by the result~\eqref{s2:eq:xxz S-matrix} from Section~\ref{s2:results}.

\subparagraph{Step 3.} We already obtained the \textsc{bae} for $M=2$ in Section~\ref{s2:method}, see~\eqref{s2:eq:BAE M=2}.

\paragraph{General $M$.} We carry out the strategy following 1\cite[\textsection8.4]{Bax07}. The reader is advised to compare each step with the corresponding step for the case $M=2$. 

\subparagraph{Step 0.} The equations \eqref{sM:eq:xxz eqns} are computed for the $M$-particle sector as for $M=2$. The result can be compactly written as
\begin{equation}\label{sM:eq:xxz eqns M}
	2\,\varepsilon_M(\Vector{p}) \, \Psi_{\Vector{p}}(\Vector{l}) = 2\bigl(M-N_{\Vector{l}}\bigr) \Delta \Psi_{\Vector{p}}(\Vector{l}) - \sideset{}{'}\sum_{\Vector{k}} \Psi_{\Vector{p}}(\Vector{k}) \ ,
\end{equation}
where the prime indicates that the sum runs over the $2\,(M-N_{\Vector{l}})$ configurations~$\Vector{k}$ obtained from $\Vector{l}$ by letting any single excited spin hop to an unexcited neighbour. 
\begin{exercise}
Check that \eqref{sM:eq:xxz eqns M} contains both \eqref{sM:eq:xxz eqns well separated M=2} and \eqref{sM:eq:xxz eqns neighbouring M=2} for $M=2$.
\end{exercise}

\subparagraph{Step~1.} For $N_{\Vector{l}}=0$ we have to solve
\begin{equation}\label{sM:eq:xxz eqns M for N=0}
	2\,\varepsilon_M(\Vector{p}) \, \Psi_{\Vector{p}}(\Vector{l}) = 2M \Delta \Psi_{\Vector{p}}(\Vector{l}) - \sideset{}{'}\sum_{\Vector{k}} \Psi_{\Vector{p}}(\Vector{k}) \ ,
\end{equation}
where the sum runs over $2M$ configurations. Plugging in the \textsc{cba}~\eqref{sM:eq:CBA M} it is easy to see that \eqref{sM:eq:xxz eqns M for N=0} is satisfied provided the energy contribution is given by the result~\eqref{s2:eq:xxz energy M} quoted in Section~\ref{s2:results}.

\subparagraph{Step~2.} Here the real work begins. Repeating the trick of extending the \textsc{cba} to $l_1 \leq l_2 \leq \cdots \leq l_M$ the equations for $N_{\Vector{l}}=0$ can again be used to simplify those for $N_{\Vector{l}}\geq 1$. Indeed, the extended $\Psi_{\Vector{p}}(\Vector{l})$ satisfies \eqref{sM:eq:xxz eqns M for N=0}, still involving a sum over $2M$ configurations, also for all (nonphysical) $\Vector{l}$ with at least one pair $l_m = l_{m+1}$. Subtracting that equation from \eqref{sM:eq:xxz eqns M} we obtain
\begin{equation}\label{sM:eq:xxz eqns new M}
	2\,N_{\Vector{l}}\ \Delta\, \Psi_{\Vector{p}}(\Vector{l}) = \sideset{}{''}\sum_{\Vector{k}} \Psi_{\Vector{p}}(\Vector{k}) \ ,
\end{equation}
where the sum now runs over the $2N_{\Vector{l}}$ configurations obtained from $\Vector{l}$ by moving one excitation in any single pair of neighbours on top of the other one. For example, when $N_{\Vector{l}}=1$, say with $l_{n+1}=l_n+1$, \eqref{sM:eq:xxz eqns new M} boils down to a simple generalization of~\eqref{sM:eq:xxz eqns new M=2}:
\begin{equation}\label{sM:eq:xxz eqns new M for N=1}
\begin{aligned}
	2\,\Delta\, \Psi_{\Vector{p}}(l_1,\To,l_n,l_n+1,\To,l_M) = \ &  \Psi_{\Vector{p}}(l_1,\To,l_n,l_n,\To,l_M) \\ & + \Psi_{\Vector{p}}(l_1,\To,l_n+1,l_n+1,\To,l_M) \ .
\end{aligned}
\end{equation}

Usually \eqref{sM:eq:xxz eqns new M} contains many more equations than there are unknowns (the functions $A_\pi$ and the values of $\Vector{p}$), so our task seems daunting. Luckily the equations that we have to solve are simplified by the following observation, exploiting the similarity between \eqref{sM:eq:xxz eqns new M} and \eqref{sM:eq:xxz eqns new M for N=1}. Suppose that we would be able to solve \eqref{sM:eq:xxz eqns new M for N=1} not just for $N_{\Vector{l}}=1$, but for \emph{any} $N_{\Vector{l}}\geq 1$. Because \eqref{sM:eq:xxz eqns new M} can be recognized as the sum of $N_{\Vector{l}}$ copies of \eqref{sM:eq:xxz eqns new M for N=1}, one for each pair of neighbouring excitations, then all other equations in \eqref{sM:eq:xxz eqns new M} would automatically be satisfied as well! Although this observation does not change the number of equations that we have to solve, the new equations all have the same form, which moreover is very similar to that of \eqref{sM:eq:xxz eqns new M=2}.

Thus we focus on \eqref{sM:eq:xxz eqns new M for N=1} for some $1\leq n\leq M-1$, where $\Vector{l}$ may or may not contain additional pair of neighbours. (The special case $n=M$, for which $l_M=L$ is next to $l_1=1$, corresponds to periodic boundary conditions. This is tackled in step~3 below.)
\begin{exercise}
Let us abbreviate $\Vector{k} \coloneqq (l_1,\To,l_n,l_n,\To,l_M)$. Plug the \textsc{cba}~\eqref{sM:eq:CBA M} into \eqref{sM:eq:xxz eqns new M for N=1} to find
\begin{equation}\label{sM:eq:xxz eqns new M for N=1 via s}
	\sum_{\pi\in S_M} s(p_{\pi(n)},p_{\pi(n+1)}) \, A_\pi(\Vector{p}) \E^{\choice{-}{}\I \Vector{p}_\pi \!\cdot \Vector{k}} = 0 \ , \qquad s(p,p') \coloneqq 1 -2\,\Delta \E^{\choice{-}{}\I p'} + \E^{\choice{-}{}\I (p+p')} \ .
\end{equation}
\end{exercise}
Not all $\E^{\choice{-}{}\I \Vector{p}_\pi \!\cdot \Vector{k}}$, $\pi\in S_M$, are independent. Indeed, since $k_n=k_{n+1}$ these exponentials only contain the quasimomenta $p_{\pi(n)}$ and $p_{\pi(n+1)}$ in the combination $p_{\pi(n)}+p_{\pi(n+1)}$. Writing $\tau_n\coloneqq(n,n+1)\in S_M$ for the transposition switching $n\leftrightarrow n+1$, this means that $\E^{\choice{-}{}\I \Vector{p}_\pi \!\cdot \Vector{k}}=\E^{\choice{-}{}\I \Vector{p}_{\pi'}\cdot\Vector{k}}$ when $\pi= \pi'\circ\tau_n$. Thus the terms in \eqref{sM:eq:xxz eqns new M for N=1 via s} come in pairs, and we find 
\begin{equation}\label{sM:eq:A via S}
	\frac{A_{\pi\tau_n}(\Vector{p})}{A_\pi(\Vector{p})} = - \frac{s(p_{\pi(n)},p_{\pi(n+1)})}{s(p_{\pi(n+1)},p_{\pi(n)})} = S(p_{\pi(n)},p_{\pi(n+1)}) \ .
\end{equation}
This is the $M$-particle generalization of equation \eqref{s2:eq:xxz S-matrix} for the two-body $S$-matrix. In summary, for each $1\leq n\leq M-1$ we obtain $M!/2$ equations \eqref{sM:eq:A via S}.

Any permutation $\pi\in S_M$ admits a (non-unique) decomposition as a product of transpositions interchanging neighbouring $p_m$, so repeated application of \eqref{sM:eq:A via S} allows us to express $A_\pi(\Vector{p})$ in terms of products of two-body $S$-matrices and an overall normalization $A_e(\Vector{p})$, with $e\in S_M$ denoting the identity. The non-uniqueness of this decomposition leads to new compatibility conditions, which are always satisfied since the two-body $S$-matrix is a scalar quantity. 
\begin{exercise}
Check that $\pi=(1,3)\in S_M$ can be decomposed as $\pi=\tau_1\circ\tau_2\circ\tau_1=\tau_2 \circ\tau_1 \circ\tau_2$. Carefully apply \eqref{sM:eq:A via S} to find $A_{(1,3)}(\Vector{p})$. Compare the result with  Figure~\ref{s2:fg:factorized scattering} from Section~\ref{s3:results}.
\end{exercise}

Up to an overall normalization, the (unique) general solution of~\eqref{sM:eq:A via S} is
\begin{equation}\label{sM:eq:A solution}
	\frac{A_\pi(\Vector{p})}{A_e(\Vector{p})} \ = \ \sgn(\pi)\prod_{1\leq m < m' \leq M} \frac{s(p_{\pi(m')},p_{\pi(m)})}{s(p_{m'},p_m)} \ = \ \prod_{(m,m') \in \inv(\pi)} S(p_m,p_{m'}) \ ,
\end{equation}
where $\inv(\pi)\coloneqq \{1\leq m<m'\leq M \mid \pi(m)>\pi(m') \}$ is the set of inversions for $\pi$. Thus we have derived the result \eqref{s2:eq:A solution}.
\begin{exercise}
As a warm-up take $\pi=\tau_n$ and check that \eqref{sM:eq:A solution} solves \eqref{sM:eq:A via S}. For the general case verify the solution by splitting the product into four parts, depending on whether or not $m$ and~$m'$ lie in $\{n,n+1\}$. Use $\sgn(\pi)=(-1)^{\#\inv(\pi)}$ to verify the second equality in \eqref{sM:eq:A solution}.
\end{exercise}

\subparagraph{Step 3.} It remains to impose periodic boundary conditions on the Bethe wave function. Indeed, when working with the coordinate basis~\eqref{s2:eq:M-part basis} for $\cH_M$ a subtlety arises because of periodicity. To avoid linear dependence amongst the vectors in \eqref{sM:eq:M vector via coord basis} we have ordered the positions of the excited spins as~$l_1<l_2<\cdots<l_M$. However, the circle $\mathbb{Z}_L$ does not possess an ordering. In the above we have implicitly chosen representatives in $\{1,2,\To,L\}\subseteq\mathbb{Z}$ for the sites~$l_m \in \mathbb{Z}_L$ (in a more pictorial language: we have cut open the circle between sites $L$ and~$1$). Of course we could have chosen to cut $\mathbb{Z}_L$ at any other point; for example, a cut just after~$l_1$ corresponds to choosing representatives $\{l_1+1,l_1 +2,\To,l_1+L\}$ and yields the ordering $l_2<l_3<\cdots<l_M<l_1+L$. Independence of the choice of representatives for~$\mathbb{Z}_L$ thus requires $\ket{l_2,\To,l_M,l_1+L} = \ket{l_1,\To,l_M}$. This just expresses the periodic boundary conditions~$\Vector{S}_{l+L}=\Vector{S}_l$ in terms of the coordinate basis. 

The upshot is that any vector $\ket{\Psi_M} \in \cH_M$ may be expressed in terms of the coordinate basis~\eqref{s2:eq:M-part basis} as in \eqref{sM:eq:M vector via coord basis} provided periodicity is imposed on the wave function as $\Psi(l_M-L,l_1,\cdots,l_{M-1}) = \Psi(l_1,\To,l_M)$ or, equivalently,
\begin{equation}\label{sM:eq:cyclicity coeff}
	\Psi(l_2,\To,l_M,l_1+L) = \Psi(l_1,\To,l_M) \ , \qquad\qquad l_1<\cdots<l_M \ .
\end{equation}

Write $\sigma=(12\cdots M)\in S_M$ for the cyclic permutation. For the Bethe wave function~\eqref{sM:eq:CBA M} we get $\Psi_{\Vector{p}}(l_2,\To,l_M,l_1+L) = \sum_{\pi'}  A_{\pi'}(\Vector{p}) \, \E^{\choice{-}{}\I p_{\pi'(1)} L} \, \E^{\choice{-}{}\I \Vector{p}_{\pi'}\cdot\Vector{l}}$, where $\pi'\coloneqq \pi\circ\sigma^{-1}$. Dropping these primes and equating coefficients in \eqref{sM:eq:cyclicity coeff} we obtain $M!$ \textsc{bae}:
\begin{equation}\label{sM:eq:BAE M general}
	\E^{\choice{-}{}\I p_{\pi(1)} L} = \frac{A_\pi(\Vector{p})}{A_{\pi\sigma}(\Vector{p})} \ , \qquad\qquad \pi \in S_M \ .
\end{equation}
\begin{exercise}
Plug in the solution \eqref{sM:eq:A solution} for $A_\pi$. Check that factors with $2\leq m'<m\leq L$ in the numerator cancel those with $1\leq m'<m\leq L-1$ in the denominator. Note that the result is symmetric in the $p_{\pi(2)},\To,p_{\pi(M)}$ so that the equations are the same for any two~$\pi,\pi'\in S_M$ with the same value~$\pi(1)=\pi'(1)$. Thus it suffices to consider transpositions of the form $\pi=(1,n)\in S_M$. Check that this yields the \textsc{bae} \eqref{s2:eq:xxz BAE M} for the $M$-particle sector:
\begin{equation}\label{sM:eq:xxz BAE M}
	\E^{\choice{-}{}\I p_m L} = (-1)^{M-1} \prod_{\substack{n=1 \\ n\neq m}}^M \frac{s(p_n,p_m)}{s(p_m,p_n)} = \prod_{\substack{n =1 \\ n\neq m}}^M S(p_n,p_m) \ , \qquad\qquad 1 \leq m \leq M \ .
\end{equation}
\end{exercise}
Thus we have derived all results quoted in Section~\ref{s2:results}.

\begin{exercise}
Work out the \textsc{cba} for the \textsc{xxz} in an external magnetic field, see Exercise~\ref{s2:ex:external magnetic field} in Section~\ref{s2:spin chains}, starting with the cases $M=0,1,2$ before attacking the general case.
\end{exercise}

\paragraph{Six-vertex case.} To conclude this appendix we briefly turn to the strategy to work out the \textsc{cba} for the six-vertex model. Thus, the goal now is to use 
\begin{equation}\label{sM:eq:6v CBA M}
	\svPsi_{\Vector{z}}(\Vector{l}) = \sum_{\pi\in S_M} \svA_\pi(\Vector{z}) \, \Vector{z}_\pi {}^{\Vector{l}} \ , \qquad \Vector{z}_\pi {}^{\Vector{l}} \coloneqq \prod_{m=1}^M (z_{\pi(m)})^{l_m} \ , \qquad\qquad l_1 < \cdots < l_M \ .
\end{equation}
This produces eigenvectors for the transfer matrix provided we can solve the equations
\begin{equation}\label{sM:eq:6v eqns}
	\sum_{\Vector{k}} \bra{\Vector{l}} \, t \, \ket{\Vector{k}} \, \svPsi_{\Vector{z}}(\Vector{k}) = \bra{\Vector{l}} \, t \, \ket{\svPsi_M; \Vector{z}} = \Lambda_M(\Vector{z}) \, \svPsi_{\Vector{z}}(\Vector{l}) \ , \qquad\qquad 1\leq l_1 < \cdots < l_M \leq L \ .
\end{equation}
The left-hand side of \eqref{sM:eq:6v eqns} is much more complicated than its \textsc{xxz}-analogue in \eqref{sM:eq:xxz eqns M}. This requires more care in formulating the strategy for using the \textsc{cba} in this case:
\begin{enumerate}
	\item[0.] Rewrite the left-hand side of \eqref{sM:eq:6v eqns} by summing geometric series in the $z_m$,
	\begin{equation}\label{sM:eq:geometric}
		\sum_{k=n}^N z^k = \frac{z^n}{1-z} - \frac{z^{N+1}}{1-z} \ .
	\end{equation}
	\item Focus on the \emph{wanted} terms, involving $\Vector{z}_\pi {}^{\Vector{l}}$ like the \textsc{cba}, to find $\Lambda_M(\Vector{z})$.
	\item Consider unwanted \emph{internal} terms, containing $(z_n z_{n+1})^{l_n}$ or $(z_n z_{n+1})^{l_{n+1}}$. Demand that these cancel to get \eqref{sM:eq:xxz eqns new M for N=1 via s}, with $\svp_m = \choice{}{-} \I \log(z_m)$ instead of $p_m$, and proceed as for the \textsc{xxz} spin chain to find $\svA_\pi(\Vector{z})$.
	\item Demand that the unwanted \emph{boundary} terms, obtained in \eqref{sM:eq:geometric} when $n=1$ or $N=L$, cancel to get \eqref{sM:eq:BAE M general} in terms of the $z_m$. The \textsc{bae} for the allowed values of~$\Vector{z}$ are then found as above.
\end{enumerate}
The details can be found in \cite[\textsection8.3--8.4]{Bax07} (cf.\ the footnote at the end of Section~\ref{s3:method}).
\begin{exercise}
Compare this strategy with that for the \textsc{xxz} model keeping in mind Exercise~\ref{s3:ex:t vs H} in Section~\ref{s3:method}.
\end{exercise}

\section{Solving the \textsc{fcr}}\label{sR}

Theorem~\ref{s4:th:FCR thm} from Section~\ref{s4:YBA} tells us that if we can find an $R$-matrix satisfying the fundamental commutation relations (\textsc{fcr}) then the corresponding transfer matrices commute, resulting in conserved quantities for the \textsc{xxz} and six-vertex models, cf.\ Sections \ref{s3:results} and~\ref{s4:conserved quantities}. Consider two monodromy matrices depending on different sets of vertex weights, such that the corresponding Lax operators are given by
\begin{equation}\label{sR:eq:Lax operator matrix}
	L_{al} \ = \
	\tikz[baseline={([yshift=-.5*11pt*0.8+8pt]current bounding box.center)}, 
		scale=0.8,font=\scriptsize]{
		\draw[->] (0,1) node[left]{$a$} -- (2,1);
		\draw[->] (1,0) node[below]{$l$} -- (1,2);
	} 
	\ = \ \begin{pmatrix} a & & & \\ & b & c & \\ & c & b & \\ & & & a \end{pmatrix}_{al} \ , \qquad
	L_{bl} \ = \ 
	\tikz[baseline={([yshift=-.5*11pt*0.8+8pt]current bounding box.center)}, 
		scale=0.8,font=\scriptsize]{
		\draw[->] (0,1) node[left]{$b$} -- (2,1);
		\draw[->] (1,0) node[below]{$l$} -- (1,2);
	}
	\ = \ \begin{pmatrix} a' & & & \\ & b' & c' & \\ & c' & b' & \\ & & & a' \end{pmatrix}_{bl}
\end{equation}
with respect to the standard bases of $V_a \otimes V_l$ and $V_b \otimes V_l$, as in Section~\ref{s4:conserved quantities}. The entries are vertex weights of two different six-vertex models. 

In this appendix we follow Baxter~\cite[\textsection9.6--9.7]{Bax07} to obtain an $R$-matrix solving the \textsc{fcr} 
\begin{equation}\label{sR:eq:FCR}
	R_{ab} \, L_{al} \, L_{bl} = L_{bl} \, L_{al} \, R_{ab} \ ,
\end{equation}
or in graphical notation
\begin{equation}\label{sR:eq:FCR graphical}
	\tikz[baseline={([yshift=-.5*11pt*0.8+8pt]current bounding box.center)},
	scale=0.8,font=\scriptsize,triple/.style={postaction={draw,-,shorten >=.05},double,double distance=4pt,-implies}]{
		\draw[->] (0,1) node[left]{$b$} -- (2,1) -- (3,2);
		\draw[->] (0,2) node[left]{$a$} -- (2,2) -- (3,1);
		\draw[->] (1,0) node[below]{$l$} -- (1,3);
	}
	\qquad = \qquad
	\tikz[baseline={([yshift=-.5*11pt*0.8+8pt]current bounding box.center)},
	scale=0.8,font=\scriptsize,triple/.style={postaction={draw,-,shorten >=.05},double,double distance=4pt,-implies}]{
		\draw[->] (0,1) node[left]{$b$} -- (1,2) -- (3,2);
		\draw[->] (0,2) node[left]{$a$} -- (1,1) -- (3,1);
		\draw[->] (2,0) node[below]{$l$} -- (2,3);
	} \ \ .
\end{equation}
As a byproduct we will find conditions on vertex weights $w$ and~$w'$ necessary to get commutating transfer matrices: from \eqref{sR:eq:FCR graphical} we will rederive the condition $\Delta(a,b,c)=\Delta(a',b',c')$ from~Section~\ref{s3:results} in the more algebraic setting from Section~\ref{s4}.

\paragraph{Reduction by symmetries.} Since $V_a\otimes V_b \otimes V_l$ has dimension eight, \eqref{sR:eq:FCR} a~priori consists of $8\times 8 =64$ equations:
\begin{equation}\label{sR:eq:FCR graphical components}
	\tikz[baseline={([yshift=-.5*11pt*0.8]current bounding box.center)},
	scale=0.8,font=\small,triple/.style={postaction={draw,-,shorten >=.05},double,double distance=4pt,-implies}]{
		\draw (0,1) node[left]{$\beta'$} -- (2,1) -- (3,2) node[right]{$\delta'$};
		\draw (0,2) node[left]{$\beta$} -- (2,2) -- (3,1) node[right]{$\delta$};
		\draw (1,0) node[below]{$\alpha$} -- (1,3) node[above]{$\gamma$};
	}
	\ \ = \ \
	\tikz[baseline={([yshift=-.5*11pt*0.8]current bounding box.center)},
	scale=0.8,font=\small,triple/.style={postaction={draw,-,shorten >=.05},double,double distance=4pt,-implies}]{
		\draw (0,1) node[left]{$\beta'$} -- (1,2) -- (3,2) node[right]{$\delta'$};
		\draw (0,2) node[left]{$\beta$} -- (1,1) -- (3,1) node[right]{$\delta$};
		\draw (2,0) node[below]{$\alpha$} -- (2,3) node[above]{$\gamma$};
	}
	\ \ , \qquad\qquad \alpha,\To,\delta' = \pm 1 \ .
\end{equation}
It is reasonable to look for solutions $R_{ab} \in \End(V_a \otimes V_b)$ that preserve the two symmetries of the six-vertex model: the ice rule (line conservation) and spin-reversal symmetry. Thus we assume that the $R$-matrix is of the same form as the Lax operators~\eqref{sR:eq:Lax operator matrix},
\begin{equation}\label{sR:eq:R-matrix matrix}
	R_{ab} \ = \
	\tikz[baseline={([yshift=-.5*11pt*0.8]current bounding box.center)},scale=0.8,font=\scriptsize]{
		\draw[->] (0,1) node[left]{$a$} -- (1,0);
		\draw[->] (0,0) node[left]{$b$} -- (1,1);
	}
	\ = \
	\begin{pmatrix}	a'' & & & \\ & b'' & c'' & \\ & c'' & b'' & \\ & & & a'' \end{pmatrix}_{ab} \ ,
\end{equation}
where the three entries have to be determined. Let us emphasize once more that, both for the Lax operator and for the $R$-matrix, the order of the `outgoing' labels is reversed in the coefficients, see \eqref{s4:eq:Lax operator acting} and \eqref{s4:eq:R acting}.

Due to parity reversal these equations come in 32 equal pairs. Line conservation requires the occupancy number to be preserved: $\alpha+\beta+\beta'=\gamma+\delta+\delta'$. This further reduces the number of nontrivial equations to~10. Using parity reversal we may restrict our attention to the ten cases with at least two incoming occupancies. Simultaneously switching $\alpha\leftrightarrow\gamma$, $\beta\leftrightarrow\delta$ and $\beta'\leftrightarrow\delta'$ in \eqref{sR:eq:FCR graphical components} yields 
\begin{equation}
	\tikz[baseline={([yshift=-.5*11pt*0.8]current bounding box.center)},
	scale=0.8,font=\small]{
		\draw (0,1) node[left]{$\delta'$} -- (2,1) -- (3,2) node[right]{$\beta'$};
		\draw (0,2) node[left]{$\delta$} -- (2,2) -- (3,1) node[right]{$\beta$};
		\draw (1,0) node[below]{$\gamma$} -- (1,3) node[above]{$\alpha$};
	}
	\ \ = \ \
	\tikz[baseline={([yshift=-.5*11pt*0.8]current bounding box.center)},
	scale=0.8,font=\small]{
		\draw (0,1) node[left]{$\delta'$} -- (1,2) -- (3,2) node[right]{$\beta'$};
		\draw (0,2) node[left]{$\delta$} -- (1,1) -- (3,1) node[right]{$\beta$};
		\draw (2,0) node[below]{$\gamma$} -- (2,3) node[above]{$\alpha$};
	} \ \ .
\end{equation}
When we rotate this over $180^\circ$ we precisely recover~\eqref{sR:eq:FCR graphical components}. But the vertex weights are invariant under such a rotation (see Figure~\ref{s3:fg:six vertices} in Section~\ref{s3:six-vertex model}), so it follows that the four equations with $\alpha=\gamma$, $\beta=\delta$ and $\beta'=\delta'$ are automatically satisfied while the six remaining equations come in equal pairs. Thus we are left with just three independent equations. Paying attention to which crossing corresponds to $L_{al}$, $L_{bl}$ and $R_{ab}$ these equations read
\begin{align}
	a \, b' c'' + c \, c' b'' \ = \ \ 
	\tikz[baseline={([yshift=-.5*11pt*0.8]current bounding box.center)},
	scale=0.8,font=\small]{
		\draw[dotted] (0,1) -- (1,1) (2.5,1.5) -- (3,1);
		\draw[very thick] (0,2) -- (1,2) -- (1,3) (1,0) -- (1,1) (2.5,1.5)  -- (3,2);
		\draw (1,1) -- (2,1) -- (2.5,1.5) -- (2,2) -- (1,2) -- (1,1);
	}
	\ \ & = \ \
	\tikz[baseline={([yshift=-.5*11pt*0.8]current bounding box.center)},
	scale=0.8,font=\small]{
		\draw[dotted] (0,1) -- (.5,1.5) (2,1) -- (3,1);
		\draw[very thick] (0,2) -- (.5,1.5) (2,0) -- (2,1) (3,2) -- (2,2) -- (2,3);
		\draw (.5,1.5) -- (1,1) -- (2,1) -- (2,2) -- (1,2) -- (.5,1.5);
	}
	\ \ = \ b \, a' c'' \ , \label{sR:eq:FCR eq 1} \\
	a \, c' b'' + c \, b' c'' \ = \ \ 
	\tikz[baseline={([yshift=-.5*11pt*0.8]current bounding box.center)},
	scale=0.8,font=\small]{
		\draw[dotted] (1,0) -- (1,1) (2.5,1.5) -- (3,2);
		\draw[very thick] (0,2) -- (1,2) -- (1,3) (0,1) -- (1,1) (2.5,1.5)  -- (3,1);
		\draw (1,1) -- (2,1) -- (2.5,1.5) -- (2,2) -- (1,2) -- (1,1);
	}
	\ \ & = \ \
	\tikz[baseline={([yshift=-.5*11pt*0.8]current bounding box.center)},
	scale=0.8,font=\small]{
		\draw[dotted] (2,0) -- (2,1) (2,2) -- (3,2);
		\draw[very thick] (0,1) -- (.5,1.5) -- (0,2) (2,1) -- (3,1) (2,2) -- (2,3);
		\draw (.5,1.5) -- (1,1) -- (2,1) -- (2,2) -- (1,2) -- (.5,1.5);
	}
	\ \ = \ b \, c' a'' \ , \label{sR:eq:FCR eq 2} \\
	a \, c' c'' + c \, b' b'' \ = \ \ 
	\tikz[baseline={([yshift=-.5*11pt*0.8]current bounding box.center)},
	scale=0.8,font=\small]{
		\draw[dotted] (1,0) -- (1,1) (2.5,1.5) -- (3,1);
		\draw[very thick] (0,2) -- (1,2) -- (1,3) (0,1) -- (1,1) (2.5,1.5)  -- (3,2);
		\draw (1,1) -- (2,1) -- (2.5,1.5) -- (2,2) -- (1,2) -- (1,1);
	}
	\ \ & = \ \
	\tikz[baseline={([yshift=-.5*11pt*0.8]current bounding box.center)},
	scale=0.8,font=\small]{
		\draw[dotted] (2,0) -- (2,1) (2,1) -- (3,1);
		\draw[very thick] (0,1) -- (.5,1.5) -- (0,2) (2,2) -- (3,2) (2,2) -- (2,3);
		\draw (.5,1.5) -- (1,1) -- (2,1) -- (2,2) -- (1,2) -- (.5,1.5);
	}
	\ \ = \ c \, a' a'' \ , \label{sR:eq:FCR eq 3} 
\end{align}

\paragraph{Solving the \textsc{fcr}.} Let us first show how \eqref{sR:eq:FCR eq 1}--\eqref{sR:eq:FCR eq 3} result in constraints on $L_{al}$ and $L_{bl}$ that are necessary for the existence of an $R$-matrix satisfying the \textsc{fcr} (and thus yield commuting transfer matrices). Eliminating the (doubly primed) entries of the $R$-matrix we recover the constraint $\Delta(a,b,c) = \Delta(a',b',c')$ from \eqref{s3:eq:commuting t's}, where $\Delta$ was defined in~\eqref{s3:eq:Delta(a,b,c)}. Recall from the discussion following~\eqref{s3:eq:commuting t's} that, for fixed value of the function~$\Delta$, the vertex weights of the six-vertex model are parametrized (up to an overall normalization) by the spectral parameter. Thus we are once more led to Lax operators depending on spectral parameters as in Theorem~\ref{s4:th:FCR thm}. In Section~\ref{s3:results} we found the condition $\Delta(a,b,c) = \Delta(a',b',c')$ for commuting transfer matrices through the results of the \textsc{cba}; presently we obtained it by purely algebraic methods!

Now we turn to the $R$-matrix itself. Notice that \eqref{sR:eq:FCR eq 1} and~\eqref{sR:eq:FCR eq 2} only differ by the position of the single and double primes, while \eqref{sR:eq:FCR eq 3} is symmetric in this respect. It follows that if we instead eliminate the (primed) entries of~$L_{bl}$ from \eqref{sR:eq:FCR eq 1}--\eqref{sR:eq:FCR eq 3} we similarly obtain the further condition $\Delta(a,b,c) = \Delta(a'',b'',c'')$. Thus, again up to an normalization (which drops out of the \textsc{fcr} at any rate), the entries of the $R$-matrix differ from those of the two Lax operators only by the value of the spectral parameter~$w$. Using the explicit parametrization~\eqref{s4:eq:parametrization} we find that \eqref{sR:eq:FCR eq 1}--\eqref{sR:eq:FCR eq 3} are satisfied provided that the spectral parameters are related by the difference property $\sinh(w) = \sinh(u-v)$, which holds for~$w=u-v$. Explicitly we thus obtain 
\begin{equation}\label{sR:eq:R-matrix entries}
	a'' = \rho'' \sinh(u-v+\I \gamma) \ , \qquad b'' = \rho'' \sinh(u-v) \ , \qquad c'' = \rho'' \sinh(\I \gamma) \ ,
\end{equation}
where $\cos\gamma=\Delta(a,b,c)=\Delta(a',b',c')$, for the entries of the $R$-matrix of the six-vertex model.


\bibliography{paper_references}
\addcontentsline{toc}{section}{References}

\end{document}